\documentclass[a4paper, 12pt]{article}

\usepackage{natbib}
\bibpunct{(}{)}{;}{a}{}{,} 

\usepackage{amsthm, amsmath, amssymb, mathrsfs, multirow, url, subfigure, latexsym, amsfonts,bm}
\usepackage{graphicx} 
\usepackage{ifthen} 
\usepackage[usenames]{color}

\usepackage{float}
\usepackage{appendix}
\usepackage{mathtools}  

\PassOptionsToPackage{hyphens}{url}\usepackage{hyperref}
\usepackage[dvipsnames]{xcolor}

\usepackage{xurl}

\usepackage[margin=1in]{geometry}
\usepackage{fancyhdr}
\theoremstyle{plain} 
\newtheorem{thm}{Theorem}

\newtheorem{lemma}{Lemma}

\newtheorem{assump}{Assumption}

\theoremstyle{definition}

\theoremstyle{remark} 

\newtheorem{rem}{Remark}

\def\var{{\rm {\mathbb Var\,}}}

\def\iid{\stackrel{\textrm{i.i.d.}}{\sim}}

\newcommand{\ee}{\mathbb{E}}

\newcommand*\dd{\mathop{}\!\mathrm{d}}
\newcommand{\ig}{\operatorname{IG}}
\newcommand{\g}{\operatorname{G}}
\newcommand{\ind}{\mathbf{1}}

\usepackage{tikz,xcolor,hyperref}

\definecolor{lime}{HTML}{A6CE39}
\DeclareRobustCommand{\orcidicon}{%
	\begin{tikzpicture}
	\draw[lime, fill=lime] (0,0) 
	circle [radius=0.16] 
	node[white] {{\fontfamily{qag}\selectfont \tiny ID}};
	\draw[white, fill=white] (-0.0625,0.095) 
	circle [radius=0.007];
	\end{tikzpicture}
	\hspace{-2mm}
}

\foreach \x in {A, ..., Z}{%
	\expandafter\xdef\csname orcid\x\endcsname{\noexpand\href{https://orcid.org/\csname orcidauthor\x\endcsname}{\noexpand\orcidicon}}
}

\title{Fast and scalable non-parametric Bayesian inference for Poisson point processes}
\author{Shota Gugushvili\footnote{Biometris, Wageningen University \& Research,
		Postbus 16,
		6700 AA Wageningen,
		The Netherlands, {shota.gugushvili@wur.nl}} \orcidA{}\negthickspace\negthinspace ,\; Frank van der Meulen\footnote{Delft Institute of Applied Mathematics, 
		Delft University of Technology,
		Van Mourik Broekmanweg 6, 
		2628 XE Delft,
		The Netherlands, f.h.vandermeulen@tudelft.nl} \orcidB{}\negthickspace\negthinspace , \; Moritz Schauer\footnote{Department of Mathematical Sciences,
		Chalmers University of Technology and University of Gothenburg,
		SE-412 96 G\"{o}teborg,
		Sweden, smoritz@chalmers.se} \orcidC{} \negthickspace\negthinspace \; \\
		and Peter Spreij\footnote{Korteweg-de Vries Institute for Mathematics, 
		University of Amsterdam,
		P.O.\ Box 94248, 
		1090 GE Amsterdam,
		The Netherlands,
		and
		Institute for Mathematics, Astrophysics and Particle Physics, Radboud University, Nijmegen, spreij@uva.nl} \orcidD{}}
\date{
\today
}

\begin{document}

\maketitle 

\thispagestyle{fancy}

\begin{abstract}    
We study the problem of non-parametric Bayesian estimation of the intensity function of a Poisson point process. The observations are $n$ independent realisations of a Poisson point process on the interval $[0,T]$. We propose two related approaches. In both approaches we model the intensity function as piecewise constant on $N$ bins forming a partition of the interval $[0,T]$.  In the first approach the  coefficients of the intensity function are assigned independent gamma priors, leading to a closed form posterior distribution. On the theoretical side,  we prove that as $n\rightarrow\infty,$ the posterior asymptotically concentrates around the ``true", data-generating intensity function at an  optimal rate for $h$-H\"older regular intensity functions ($0 < h\leq 1$).

In the second approach we employ a gamma Markov chain prior on the coefficients of the intensity function. The posterior distribution is no longer available in  closed form, but inference can be performed using a straightforward version of the Gibbs sampler. Both  approaches  scale well with sample size, but the second is much less sensitive to the choice of $N$.

Practical performance of our methods is first demonstrated via synthetic data examples. 
We  compare our second method with other existing approaches on the UK coal mining disasters data. Furthermore, we apply it to the US mass shootings data and Donald Trump's Twitter data.

\smallskip

\end{abstract}

\section{Introduction}
\label{intro}
Poisson point processes (see \cite{kingman93}) are among the basic modelling tools in areas as different as astronomy, biology, geology, image analysis, medicine, operations research, physics, reliability theory, social media and others; see, e.g., \cite{bremaud81}, \cite{moller04} and \cite{streit10}. Their probabilistic properties are completely determined by their intensity measure $\Lambda$, or its density, the intensity function $\lambda$.  In this paper, we study a non-parametric Bayesian approach to estimation of the intensity function $\lambda$ for univariate Poisson point processes. We propose two related approaches that lead to simple implementations, and also enjoy favourable performance in practice.

\subsection{Setting}
We introduce our statistical setting in greater detail. We restrict our attention to a Poisson point process $X$ on an interval $[0,T]$ of the real line $\mathbb{R}$, which we equip with the Borel $\sigma$-field $\mathcal{B}([0,T])$. Such Poisson point processes are also often called non-homogeneous Poisson processes. The process $X$ is a random integer-valued measure on $[0,T]$ (we assume the underlying (complete) probability space $(\Omega,\mathcal{F},\mathbb{Q})$ in the background), such that
\begin{enumerate}
	\item for any disjoint subsets $B_1,B_2,\ldots,B_m\in\mathcal{B}([0,T]),$ the random variables $X(B_1),$ $X(B_2),\ldots,X(B_m)$ are independent, and
	\item for any $B\in\mathcal{B}([0,T]),$ the random variable $X(B)$ is Poisson distributed with parameter $\Lambda(B)$. Here $\Lambda$ is a finite measure on $([0,T],\mathcal{B}([0,T])),$ called the intensity measure of the process $X.$ Moreover, it is assumed that $\Lambda$ admits a density $\lambda$ with respect to the Lebesgue measure on $\mathcal{B}([0,T])$.
\end{enumerate}
Intuitively, the process $X$ can be thought of as random scattering of points in $[0,T],$ where the way the scattering occurs is determined by properties (i)--(ii).

A popular assumption in the statistical literature (see, e.g., \cite{karr86} and \cite{kutoyants98}) is that one has independent copies $X_1,\ldots,X_n$ of the process $X$ at her disposal.  Based on these observations, an estimator of the intensity function $\lambda$ has to be constructed. Intensity functions that are periodic with a known period also lead to this statistical setting.

\subsection{Literature overview}
\label{subsec:literature}

Theoretical results for a kernel-type estimator of the intensity function are given in \cite{kutoyants98}, where further references can be found. See also \cite{diggle85}, where a practical implementation is discussed in a closely related problem of estimation of the intensity of a stationary Cox process. Some recent references treating various types of the kernel method are \cite{baddeley16}, \cite{diggle14}, and \cite{lieshout18}.

Works dealing with non-parametric\footnote{As put characteristically in \cite{efron16}, page 172: ``Nonparametric" is another name for ``very highly parametrized''.} Bayesian intensity function estimation include, among others,
\cite{adams09}, \cite{arjas98}, \cite{hensman15}, \cite{john18}, \cite{komsamo15}, \cite{moller98} and  \cite{rao11}, and concentrate primarily on computational aspects. On the other hand, \cite{harry15} is a theoretical contribution analysing the approach in \cite{adams09} (cf.~also \cite{gugu13}). \cite{vz13} deals both with practical and theoretical aspects of the problem. Another recent paper dealing with derivation of posterior contraction rates in the Poisson point process setting using Gaussian priors is \cite{grant_leslie19}.

Advantages of a non-parametric Bayesian approach over the kernel method are succinctly summarised in the discussion given in \cite{adams09}. One obvious drawback of a ``naive" kernel estimator is that it is inconsistent at the boundary of the set $[0,T]$ on which the process is defined. In practice this estimator will also behave poorly close to boundary points. This has to do with the well-known boundary bias problem of the kernel estimator. Simulation examples in \cite{adams09} demonstrate that even after the correction for edge effects following the method in \cite{diggle85} and with an optimal choice of the bandwidth parameter, the kernel method is still  outperformed by the Bayesian approach of \cite{adams09}. Secondly, if a kernel of order higher than two is used,  a  non-negative estimate of the intensity function is not guaranteed by the kernel method. On the other hand, the Bayesian approaches studied in the above cited works are often computationally intense, in that their practical implementation is based on advanced Markov Chain Monte Carlo (MCMC) or optimisation methods, and are also not trivial to implement.

\subsection{Our contribution}
In this paper we propose two simple non-parametric Bayesian approaches to estimation of the intensity function. We model the intensity function as piecewise constant on the domain of its definition $[0,T].$ In our first approach we  equip the coefficients of the intensity function with independent gamma priors, and obtain a posterior distribution that is known in closed form and is also of gamma type. Posterior inference with this approach is computationally elementary, with no need to recourse to simulation methods such as MCMC, or optimisation techniques such as variational Bayes. Hence the method scales extremely well with sample size. This method can be thought of as a simple to use Bayesian tool for preliminary, exploratory data analysis, that is much akin to a histogram or a regressogram; see \cite{wasserman06}. The approach requires the choice of the number of coefficients of the intensity function. A simple practical method to select this hyperparameter is to use the empirical Bayes method. 

On the theoretical side, we derive the contraction rate of the posterior distribution around the ``true" intensity function $\lambda_0$ under $\mathbb{P}_{\lambda_0}^{(n)},$ where $\mathbb{P}_{\lambda_0}^{(n)}$ denotes the law of the observation $X^{(n)}=(X_1,\ldots, X_n)$ under the true parameter $\lambda_0.$
This concerns taking a sequence of shrinking neighbourhoods of $\lambda_0$ and determining the fastest rate, at which these neighbourhoods can shrink, while still capturing most of the posterior mass (the precise definition will be given below). This rate can be thought of as an analogue of the convergence rate of a frequentist estimator. 
As we will demonstrate, our approach attains the optimal posterior contraction rate for estimating an $h$-H\"older regular intensity function, $0<h\leq 1$. 
We stress the fact that the data-generating $\lambda_0$ in our approach is not necessarily assumed to be piecewise constant with known break points.

Inspired by ideas in the audio signal processing literature, see, e.g., \cite{cemgil07}, \cite{peeling08} and \cite{dikmen10}, as well as inference in diffusion models, see \cite{gugu18micro} and \cite{gugu18}, we next propose a second non-parametric Bayesian method. This method extends the first and relies on the gamma Markov chain (GMC) prior. Specifically, as in our first approach, we model a priori the intensity function as piecewise constant, but now its coefficients under a prior form a gamma Markov chain. Unlike our first method, the posterior is not available in closed form. However, this second method can be implemented in a straightforward way using the Gibbs sampler.
We initialise our Gibbs sampler using the information obtained from our first method, to ensure quick convergence of the resulting Markov chain. 
We provide a comparison of both methods via simulations, and show that the first one is more sensitive to the appropriate choice of the number of bins and is outperformed by the second one in practice. Finally, like the first method, also the second one scales well with sample size: once the number of Poisson events falling within each bins has been determined, the computational complexity per MCMC step of the second method is linear in the number of bins. 

The methodology developed in this paper is applied on three real data examples: the UK coal mining disasters data, the US mass shootings data and Donald Trump's Twitter data. The first dataset is a classical benchmark in point process inference and is used to compare our method with results from \cite{adams09} and \cite{lloyd15}.

\subsection{Structure of the paper}
The rest of the paper is organised as follows: in Section \ref{sec:likelihood} we introduce in detail our first Bayesian procedure, study its frequentist asymptotics, and discuss one method for practical selection of a hyperparameter (the number of coefficients), which governs properties of the prior. In Section \ref{sec:gamma} we introduce our second Bayesian estimation technique based on the GMC prior. In Section \ref{sec:simulations} we present simulations examples, while in Section \ref{sec:realdata} we  apply our approach on real datasets. In Section \ref{sec:conclusions} we summarise contributions of our paper. Appendix \ref{appendix}  contains the proofs of technical results. Finally, in Appendices \ref{sec:gp} and \ref{sec:revjump} we include some results obtained based on  Gaussian processes and equipping a prior on the model index respectively.

\subsection{Notation}
For two sequences $\{a_n\}$ and $\{b_n\}$ of positive real numbers, the notation $a_n\lesssim b_n$  (or $b_n\gtrsim a_n$) means that there exists a constant $C>0$ that is independent of $n$ and such that $a_n\leq C b_n.$ We write $a_n\asymp b_n$ if both $a_n\lesssim b_n$ and $a_n\gtrsim b_n$ hold. $\|\cdot\|_2$ denotes  the $L_2$-norm with respect to the Lebesgue measure on the Borel sets of $[0,T]$. 
We denote a prior (depending on the hyperparameter $N$) by $\Pi_N$. Given data $X^{(n)}$, the corresponding posterior measure is denoted by  $\Pi_N(\cdot\mid X^{(n)})$, and $\ee_{\Pi_N}[\cdot \mid X^{(n)}]$ stands for expectation with respect to the posterior measure, while $\var_{\Pi_N}(\cdot \mid X^{(n)})$ is the corresponding variance. The gamma distribution with shape parameter $\alpha$ and rate parameter $\beta$ ($\alpha,\beta >0$) is denoted by $\operatorname{G}(\alpha,\beta)$. Its density is given by
$x \mapsto \frac{\beta^{\alpha}}{\gamma(\alpha)} x^{\alpha-1}e^{-\beta x}, \quad x>0$,
where $\Gamma$ is the gamma function. The inverse gamma distribution with shape parameter $\alpha$ and scale parameter $\beta$ is denoted by $\ig(\alpha,\beta)$. Its density is 
$x \mapsto \frac{\beta^{\alpha}}{\Gamma(\alpha)} x^{-\alpha-1}e^{-\beta / x}, \quad x>0$.
A normal (Gaussian) density with mean $\mu$ and standard deviation $\sigma$ is denoted by $\phi(\cdot;\mu,\sigma)$. We use the notation $\operatorname{Uniform}(a,b)$ for a uniform distribution on $[a,b]$, whereas $\operatorname{Exp}(\beta)$ stands for an exponential distribution with mean $1/\beta$.  In conformance with standard Bayesian notation, we will often use lowercase letters for random variables and write a density of a random variable $x$ as $p(x)$. A conditional density of $x$ given $y$ will be written as $p(x\mid y)$, while conditioning of $x$ on $y$ will be denoted by $x\mid y$. Finally, $\lceil a \rfloor$ will stand for the integer nearest to the real number $a$.

\section{Independent gamma prior}
\label{sec:likelihood}

In this section we present in detail our first method of estimation of the intensity function.

\subsection{Likelihood}
In a Bayesian approach to estimation of $\lambda$ one starts with putting a prior $\Pi$ on $\lambda$. This might be thought of as reflecting one's prior knowledge or beliefs on $\lambda.$ In formal terms, the prior is a measure $\Pi$ defined on the parameter set $\Theta$ (a set of functions $\lambda\colon[0,T]\mapsto\mathbb{R}_{+}$) equipped with a $\sigma$-field $\sigma(\Theta)$. By Proposition~6.11 in \cite{karr86} or Theorem 1.3 in \cite{kutoyants98}, the law of $X$ under the parameter value $\lambda$ admits a density $p(\cdot\,;\lambda)$ with respect to the measure induced by a standard Poisson point process of rate $1.$ This density is given by
\begin{equation*}
p(\xi;\lambda)=\exp\left( \int_{[0,T]} \log \lambda(x)\mathrm{d}\xi(x) - \int_{ [0,T]  } (\lambda(x)-1) \mathrm{d}x \right),
\end{equation*}
where $\xi=\sum_{i=1}^m \delta_{x_i}$ is a realisation of $X$ (here $\delta_{x_i}$ denotes the Dirac measure at $x_i$). 
We assume independent observations $X_1,\ldots, X_n$ with the same distribution as $X$ and write 
$ X_j = \sum_{i=1}^{m_j} \delta_{X_{ij}}$ for $j=1,\ldots, n$. 
If we define $X^{(n)}=(X_1,X_2,\ldots,X_n)$, then it follows that the likelihood $L(X^{(n)};\lambda)$  can be written as
\begin{equation}
\label{likelih}
L(X^{(n)};\lambda)=\prod_{j=1}^n \exp\left( \int_{ [0,T] } \log \lambda(x)\mathrm{d}X_{j}(x) - \int_{ [0,T]  } \left(\lambda(x)-1\right) \mathrm{d}x \right).
\end{equation}

\subsection{Prior}\label{subsection:prior}
Fix a positive integer $N$. Let $0=b_0<b_1<\cdots<b_{N-1}<b_N=T$ be a grid of points on the interval $\mathcal{X}=[b_0,b_N]$, for instance a uniform grid. Using this grid, define the bins $B_k=[b_{k-1},b_{k}),$ $k=1,\ldots,N-1$, with the last bin $B_N = [b_{N-1},b_N]$. In order to define a prior on $\lambda$, introduce a collection of piecewise constant functions $\lambda$ on bins $B_k$,
\begin{equation}
\label{lambda}
\lambda(x)=\sum_{k=1}^N \psi_k \mathbf{1}_{B_k}(x), \quad x\in[0,T].
\end{equation}
\begin{assump}
	\label{assump:prio}
	Assume that the coefficients $\psi_k \iid \operatorname{G}(\alpha,\beta)$. Define the prior $\Pi_N$ on $\lambda$ to be the law of the random function \eqref{lambda}.
\end{assump}

We will refer to the prior $\Pi_N$ as the independent gamma prior. If the grid on $\mathcal{X}=[0,T]$ is taken to be uniform, the prior has three hyperparameters: $\alpha,$ $\beta,$ and $N.$ 
Depending on the amount of available data and the degree of prior knowledge on $\lambda$, the hyperparameters $\alpha,\beta$ can for instance be chosen to render the $\operatorname{G}(\alpha,\beta)$ prior diffuse (non-informative). This corresponds to the case when one has little information on the magnitude and shape of $\lambda$. 
The hyperparameter $N$ can be viewed as a smoothing parameter. We discuss one practical method for its choice in Subsection \ref{sec:ebayes}.

Using a prior with piecewise constant realisations as in \eqref{lambda} is not unnatural, as follows by a comparison to a histogram-like frequentist procedure studied in detail from the theoretical point of view in \cite{henderson03} and \cite{leemis04}. In fact, in practical applications it is often the case that Poisson point processes are only partially observed through aggregate counts of points over successive intervals (cf.~\cite{henderson03}). The intensity function cannot be learned beyond the resolution of these intervals, which lends support to using priors with piecewise constant realisations to model the intensity function. 

\subsection{Posterior}
Since the function $\lambda$ in our case is parametrised by the coefficients $\psi_1,\ldots,\psi_N$, the posterior distribution of the intensity function $\lambda$ given the data $X^{(n)}$ can be equivalently described through the posterior distribution $\psi_1,\ldots,\psi_N \mid X^{(n)}$. Transition from the prior to the posterior can be thought of as updating our prior opinion on $\lambda$ upon seeing the data $X^{(n)}.$
\begin{lemma}\label{lem:indepposterior}
	\label{lem:gamma}
	Define 
	\begin{equation}\label{eq:def-Hk}	
	H_k = \sum_{j=1}^n \sum_{i=1}^{m_j} \ind_{\{X_{ij} \in B_k\}}
	\end{equation}
	and $\Delta_k= b_k - b_{k-1}$ for $k=1,\ldots,N.$ Then  $\psi_1,\ldots,\psi_N$ are a posteriori independent and 
	\begin{equation}
	\label{post:gamma}
	\psi_k \mid   X^{(n)} \sim \operatorname{G}(\alpha  + H_k, \beta + n\Delta_k ), \quad k=1,\ldots,N.
	\end{equation}	
\end{lemma}
Thus, the posterior for $\lambda$ is known in closed form. The posterior mean is given by
$ \hat{\lambda}(x)=\sum_{k=1}^N \hat{\psi}_k \mathbf{1}_{B_k}(x)$, $x\in [0,T]$,
with
\begin{equation}
\label{psik}
\hat{\psi}_k = \frac{ H_k + \alpha }{ n\Delta_k+\beta }, \quad k=1,\ldots,N.
\end{equation}
Marginal $1-\gamma$ credible bands for $\lambda$ can be obtained by producing $1-\gamma$ credible intervals for $\psi_k,$ using the lower and upper $\gamma/2$-quantiles of the gamma distribution. This gives Bayesian uncertainty quantification for estimation of $\lambda.$

If we were to ignore the hyperparameters $\alpha,\beta$ (or alternatively, if we were to consider the limit $n\rightarrow\infty$ in such a way that $n\min_k\Delta_k\rightarrow\infty$), then the posterior mean $\hat{\lambda}$ would have coincided with a frequentist estimator of the intensity function from \cite{henderson03}; cf.~also \cite{leemis04}. Thus our approach sheds additional light on the latter method, in that it yields its Bayesian interpretation. 

Returning to an interpretation of the prior distribution on $\lambda$, our view aligns with that in \cite{martin19}: ``[I]n high-dimensional problems\ldots the role [of prior distributions] is simply to facilitate efficient posterior inference, and, therefore, only priors whose corresponding posterior has good properties are used". That the latter is indeed the case, is the subject of the next subsection. Computational efficiency is clear from the above discussion and will be further demonstrated below.

\subsection{Bayesian asymptotics}
\label{sec:asymptotics}

In this subsection we perform the asymptotic frequentist analysis of the Bayesian procedure we introduced above. This concerns the study of the asymptotic properties of the posterior measure as the sample size $n\rightarrow\infty$.

We first formalise our assumption on the bins $B_k$.

\begin{assump}
	\label{assump:bins}
	Assume that the grid $\{b_k\}$ defining the bins $B_k$ is uniform: $b_k=Tk/N,$ $k=0,\ldots,N$, so that the bins are of equal width $\Delta_k=\Delta=T/N$.
\end{assump}

The assumption that the grid $\{b_k\}$ is uniform is made for simplicity in the proofs, and can in fact be relaxed. It is not necessary for our method to work in synthetic and real data examples.

The next condition deals with the ``true", data-generating intensity function $\lambda_0$. It places a modest smoothness assumption on $\lambda_0$, which will often be satisfied in practice.

\begin{assump}
	\label{lambda0}
	The intensity function $\lambda_0\colon[0,T]\rightarrow (0,\infty)$ is H\"older continuous: there exist constants $L>0$ and $0<h\leq 1$, such that
	\[
	|\lambda_0(x)-\lambda_0(y)| \leq L |x-y|^h, \quad \forall x,y\in [0,T].
	\]
\end{assump}

Our first result shows that the posterior mean $\hat{\lambda}$ is a consistent estimator of $\lambda_0$, and establishes its convergence rate. The expectation $\ee[\,\cdot\,]$ here and in the sequel is always under the law of the observations with the ``true" parameter value $\lambda_0$.

\begin{thm}
	\label{thm:mean}
	Let Assumption \ref{lambda0} hold, and assume $N \asymp n^{1/(2h+1)}$. Then
	\[
	\label{mse}
	\ee [\| \hat{\lambda} - \lambda_0  \|_2^2] \lesssim n^{-2h/(2h+1)}.
	\]
\end{thm}

The right-hand side is the optimal rate for estimating an $h$-H\"older-regular intensity function, see \cite{kutoyants98}.

The next result gives a posterior contraction rate in the $L_2$-metric. Whereas Theorem \ref{thm:mean} dealt with the `centre' of the posterior distribution, the theorem below deals with the entire posterior distribution.

\begin{thm}
	\label{thm:post}
	Let the assumptions of Theorem \ref{thm:mean} hold, and let $\varepsilon_n\asymp n^{-h/(2h+1)}$. Then, for any sequence $M_n\rightarrow \infty$,
	\[
	\ee [\Pi_N(\|\lambda-\lambda_0\|_2 \ge M_n \varepsilon_n \mid X^{(n)})] \rightarrow 0
	\]
	as $n\rightarrow\infty$.
\end{thm}
The first conclusion that follows from Theorem \ref{thm:post} is that the proposed Bayesian procedure is consistent: as the sample size $n\rightarrow\infty$, the posterior puts most of its mass on (shrinking) $L_2$-neighbourhoods around the true parameter $\lambda_0$. Furthermore, the rate $\varepsilon_n\asymp n^{-h/(2h+1)}$ in  Theorem \ref{thm:post} is the optimal posterior contraction rate for $h$-H\"older-regular intensity functions. 

Note that the posterior contraction rate in Theorem \ref{thm:post} does not involve a $\log n$ factor, that often appears when studying the frequentist asymptotics of non-parametric Bayesian procedures. The reason for this is that our proof does not appeal to the powerful and general-purpose machinery from \cite{ghosal00}, which might yield a superfluous log factor in the contraction rate. Instead, it relies on elementary and direct calculations and estimates.

Although not concerned with direct inference in Poisson point process models, a work related to ours is \cite{grant19}. There the authors study the problem of adaptively placing sensors along an interval to detect stochastically-generated events. A histogram-type Bayesian prior on the intensity function plays a significant role in the developments in that paper. The authors derive an $O(M^{2/3})$ bound on Bayesian regret in $M$ rounds, which in their work is comparable to the statements of Theorems \ref{thm:mean} and \ref{thm:post} that we gave above.

\subsection{Empirical Bayes for bin number selection}
\label{sec:ebayes}

According to the asymptotic results in Subsection~\ref{sec:asymptotics}, the proposed Bayesian approach is guaranteed to be consistent and asymptotically optimal for estimating an $h$-H\"older intensity function (with $0<h\le 1$) if the number of bins satisfies $N\asymp n^{1/(2h+1)}$. However, this choice of the hyperparameter $N$ depends on a proportionality constant. In practice the resulting performance of our Bayesian procedure may turn out to be suboptimal for a given sample size and given dataset, if the constant is not chosen appropriately. 
Essentially $N$ plays the role of a smoothing parameter. 
In general, choice of a smoothing parameter constitutes the biggest challenge in non-parametric estimation (see, e.g., \cite{loader99}). \cite{chaudhuri00} goes as far as to propose a method `agnostic' of such a choice. 
Frequentist papers \cite{henderson03} and \cite{leemis04}, that are related to our work, concentrate on asymptotics of the kernel estimator of the intensity function and do not provide specific practical guidance for selecting the number of bins.

The main idea behind our approach in this section is that the marginal likelihood can be viewed as model evidence in Bayesian statistics. Maximising it over $N$, roughly speaking, selects a model that is most compatible with the data at hand.  This can be seen as an instance of the well-known empirical Bayes method (see, e.g., \cite{efron10}).

The marginal likelihood is given by
\begin{equation}\label{eq:marginal_lik}
\begin{split}
\operatorname{ML}_N(X^{(n)})&=\int_{[0,\infty)^N} L( X^{(n)};\psi_1,\ldots,\psi_N ) \prod_{k=1}^N \pi(\psi_k)\dd\psi_1 \cdots \dd\psi_N\\
&=e^{Tn} \frac{\beta^{\alpha N}}{\Gamma(\alpha)^N} \prod_{k=1}^N \frac{\Gamma(\alpha+H_k)}{(n\Delta_k+\beta)^{\alpha+H_k}}.
\end{split}
\end{equation}
Viewed as a function $N \mapsto \operatorname{ML}_N (X^{(n)}) $ (with $\alpha$ and $\beta$ fixed), the marginal likelihood can thus be easily evaluated and maximised graphically. For numerical stability, it is advisable to work with $ \operatorname{LML}_N(X^{(n)}) = \log\operatorname{ML}_N(X^{(n)})$.  We study the behaviour of this procedure for selecting $N$ in simulation examples in Section \ref{sec:simulations}.

Alternatively, at first sight, the marginal likelihood can also be used to optimise the hyperparameters $\alpha,\beta$ of the prior (keeping $N$ fixed). Cf.~\cite{clayton87} for a somewhat similar idea in a different context than ours. However, we advise against this approach in our setting, as the resulting procedure is plagued by numerical problems. Instead, a numerically stable empirical Bayes procedure can be obtained by maximisation of the marginal likelihood over $\beta$ for a  fixed $\alpha$ (and $N$). The (unique) maximiser can be found upon setting the derivative of the criterion function with respect to $\beta$ to zero, which leads to the relation
\begin{equation}
\label{eq:beta}
\frac{\alpha}{\beta}=\frac{1}{N} \sum_{k=1}^N \frac{H_k+\alpha}{n\Delta_k+\beta}.
\end{equation}
This has an intuitive interpretation. Namely,  the left-hand side is the prior mean of $\psi_k$, whereas the right-hand side is the  average of the posterior means $\hat{\psi}_k$'s, see equation \eqref{psik}. Thus $\beta$ chosen according to the above rule implies a stability property, whereby the prior mean matches the (average) posterior mean.

As small values of hyperparameters correspond to diffuse priors, one may want to fix $\alpha$ and $\beta$ at small positive values (from \eqref{eq:beta} it follows that for small $\alpha$ also $\beta$ should be small).

\section{Gamma Markov chain prior}
\label{sec:gamma}

In this section we propose an alternative Bayesian approach to intensity function estimation. Ideas we use to that end have appeared in various works in the audio signal modelling literature (see, e.g., \cite{cemgil07}, \cite{dikmen10} and \cite{peeling08}). They were applied in the volatility estimation setting for diffusion processes in \cite{gugu18}, where a prior resembling the one below was employed as a conjugate prior for a Gaussian likelihood. 

Our starting point is the same as in Section~\ref{sec:likelihood}. Namely, as in equation \eqref{lambda}, we model the intensity function as piecewise constant on bins $B_k$ forming a partition of the interval $[0,T]$. However, instead of assuming that the coefficients $\psi_k$ of the intensity function $\lambda$ are a priori independent and gamma distributed, we postulate that under a prior they form a gamma Markov chain (GMC). This chain is defined as follows: introduce auxiliary variables $\zeta_k,k=2,\ldots,N$, use the time ordering $\psi_1,\zeta_2,\psi_2,\ldots,\zeta_N,\psi_N$, and set
\begin{equation}
\label{formula:prior}
\psi_1 \sim \g(\alpha_{1},\beta_{1}), \quad \zeta_{k} \mid \psi_{k-1} \sim \ig(\alpha_{\zeta},\alpha_{\zeta} \psi_{k-1}), \quad \psi_{k} \mid \zeta_{k} \sim \g\left(\alpha_{\psi},\frac{\alpha_{\psi}}{\zeta_{k}}\right), 
\end{equation}
where $k\in \{2,\ldots, N\}$. 
The name of the chain reflects the fact that its transition distributions are (inverse) gamma. The parameters $\alpha_1,$ $\beta_1$, $\alpha_{\zeta}$ and $\alpha_{\psi}$ are the hyperparameters of the GMC prior. The hyperparameters $\alpha_1,\beta_1$ allow one to `release' the chain at the origin. This is important to avoid possible edge effects in non-parametric estimation due to a strong specification of the prior at the time origin $t=0$. Next, a principal aim in using latent variables $\zeta_k$'s in \eqref{formula:prior} is to attain positive correlation between $\psi_k$'s. In the intensity function modelling context this induces smoothing across different bins. Depending on whether the ratio $\alpha_{\zeta}/\alpha_{\psi}$ is less than one, equal to one, or greater than one, the subsequence $\{\psi_k\}$ extracted from the gamma Markov chain exhibits in the limit $N\rightarrow\infty$ a decreasing trend, no trend, or an increasing trend, respectively; cf.~\cite{cemgil07}. This feature is attractive in case one possesses prior information on the monotonicity properties of the ``true" intensity function $\lambda$. Furthermore, large values of the hyperparameters $\alpha_{\zeta},\alpha_{\psi}$ correspond to a strong correlation between $\psi_k$'s and a slow decay of the autocorrelation function, see Figure \ref{fig:psi} for an illustration. The GMC prior with large values of $\alpha_{\zeta},\alpha_{\psi}$ has thus a stronger smoothing effect.
\begin{figure}
	\begin{center}
		\includegraphics[width=0.9\textwidth]{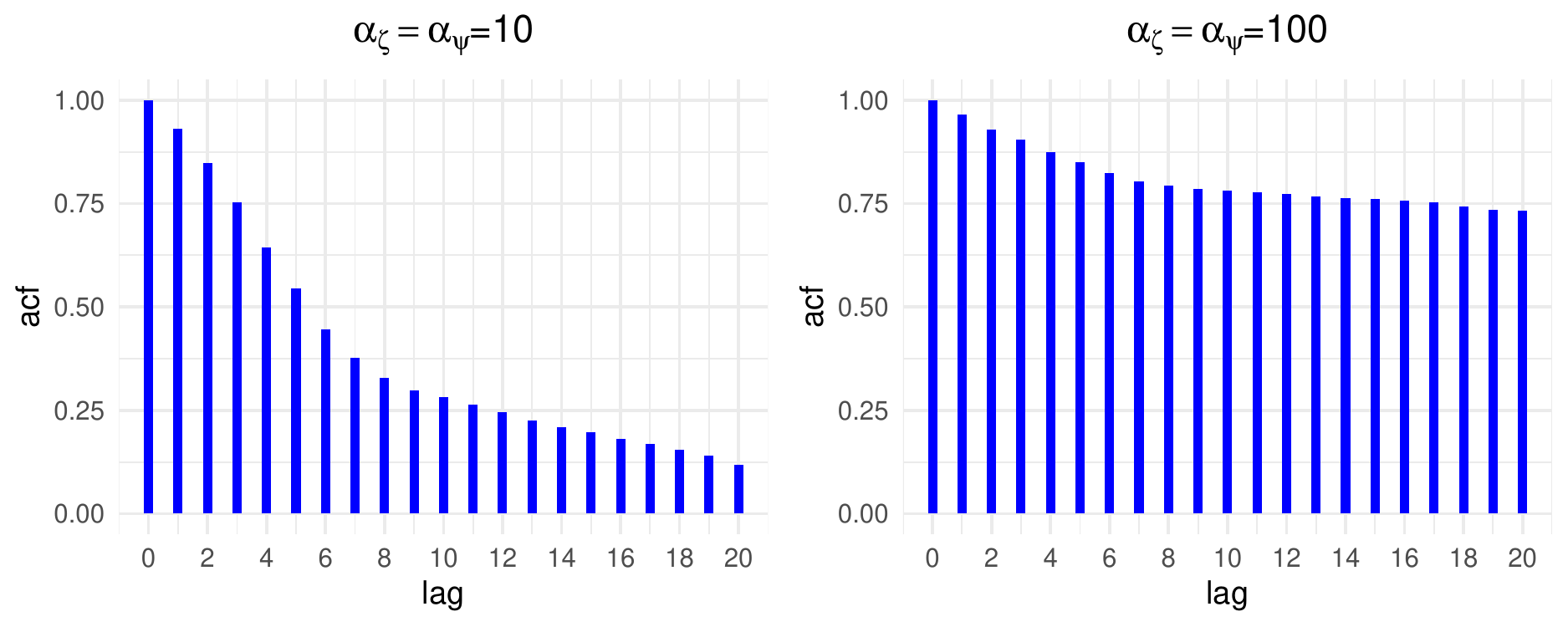}
	\end{center}
	\caption{Sample autocorrelation function of $\{\psi_k\}$ based on two realisations of the gamma Markov chain with $N=1000$. The left plot corresponds to a realisation with $\alpha_{\zeta}=\alpha_{\psi}=10$, the right one to $\alpha_{\zeta}=\alpha_{\psi}=100$. In both cases, $\alpha_1=4$, $\beta_1=1$.}
	\label{fig:psi}
\end{figure}	

\subsection{Sampling from the posterior}
The posterior distribution with the GMC prior is not available in closed form. This necessitates the use of some approximate posterior inference technique. In this subsection we give the details of the Gibbs sampler to draw (dependent) samples from the posterior.

\subsubsection{Full conditional distributions}
\label{sec:blankets}
By the Markov property in \eqref{formula:prior}, the joint distribution of $\{\psi_k\}$ and $\{\zeta_k\}$ factorises as
\begin{equation}
\label{formula:parameters}
p(\psi_1) \prod_{k=2}^N p(\zeta_{k} \mid \psi_{k-1}) p(\psi_k \mid \zeta_k) .
\end{equation}
Using this formula and \eqref{formula:prior}, it can be seen that the full conditional distributions are determined by
\begin{align*}
\zeta_k \mid \psi_k,\psi_{k-1} &\sim \ig\left(\alpha_{\zeta}+\alpha_{\psi}, {\alpha_{\zeta}}{\psi_{k-1}}+{\alpha_{\psi}}{ \psi_k}\right), \quad k=2,\ldots,N,\\
\psi_k \mid \zeta_{k+1},\zeta_{k} &\sim \g\left(\alpha_{\psi}+\alpha_{\zeta},\frac{\alpha_{\psi}}{\zeta_k}+\frac{\alpha_{\zeta}}{ \zeta_{k+1}}\right), \quad k=2,\ldots,N-1, 
\\
\psi_1 \mid \zeta_2 & \sim \g\left(\alpha_1+\alpha_{\zeta},  \beta_1 + \frac{\alpha_{\zeta}}{ \zeta_{2}} \right),  
\\
\psi_N \mid \zeta_N & \sim \g\left(\alpha_{\psi},\frac{\alpha_{\psi}}{\zeta_N}\right).
\end{align*}

\subsubsection{Gibbs sampler}
\label{sec:gibbs}

Combining the preceding results with formula \eqref{eq:likelih}, we deduce that  given the data $X^{(n)}$, the full conditional distributions of $\psi_k$, $\zeta_k$ are
\begin{align}
\zeta_k \mid \psi_k,\psi_{k-1} &\sim \ig\left(\alpha_{\zeta}+\alpha_{\psi}, {\alpha_{\zeta}}{\psi_{k-1}}+{\alpha_{\psi}}{ \psi_k}\right), \quad k=2,\ldots,N, \label{fullcond2} \\
\psi_k \mid \zeta_{k+1},\zeta_{k},X^{(n)} &\sim \g\left(\alpha_{\psi}+\alpha_{\zeta}+H_k,\frac{\alpha_{\psi}}{\zeta_k}+\frac{\alpha_{\zeta}}{ \zeta_{k+1}} + n\Delta_k\right), \quad k=2,\ldots,N-1, \label{fullcond1} \\
\psi_1 \mid \zeta_2,X^{(n)} & \sim \g\left(\alpha_1+\alpha_{\zeta}+H_1,  \beta_1 + \frac{\alpha_{\zeta}}{ \zeta_{2}} + n\Delta_1 \right), \label{fullcondstart} \\
\psi_N \mid \zeta_N,X^{(n)} & \sim \g\left(\alpha_{\psi}+H_N,\frac{\alpha_{\psi}}{\zeta_N}+n\Delta_N\right). \label{fullcondend}
\end{align}
The Gibbs sampler (see \cite{geman84} and \cite{gelfand90}) cycles through formulae \eqref{fullcond2}--\eqref{fullcondend} to generate (approximate) samples from the posterior distribution of $\{\psi_k\}$, $\{\zeta_k\}$ given the data $X^{(n)}$.  One can initialise the sampler e.g.\ by providing values for $\psi_1,\ldots,\psi_N$. 

\begin{rem}
	Comparison of formula \eqref{fullcond1} to Lemma~\ref{lem:gamma} shows that our two methods for estimating the intensity function match each other when diffuse limits $\alpha=\beta\rightarrow 0$ and $\alpha_{\psi}=\alpha_{\zeta}\rightarrow 0$ are taken, so that there is a type of continuous transition between the two methods.
\end{rem}

\subsubsection{Sampler initialisation}
\label{sec:init}

Quick convergence of the Gibbs sampler from Subsection~\ref{sec:gibbs} can be facilitated by a good initialisation of the Markov chain. In our simulations we obtained starting values for $\psi_1,\ldots, \psi_N$ by drawing from their posterior distribution based on the independent gamma prior. However, in the simulation examples we considered, even in those cases when we used rather poor initial values for the sampler, it stabilised quite quickly to the `correct' part of the parameter space (results not shown in this paper).

\subsection{Hyperparameters}

Assume the number of bins $N$ is fixed. The GMC prior in \eqref{formula:prior} depends on hyperparameters $\alpha_1,$ $\beta_1$, $\alpha_{\zeta}$ and $\alpha_{\psi}$. In practice we recommend to use a diffuse prior on $\psi_1$. Hyperparameters $\alpha_1,\beta_1$ can be either both fixed manually, or obtained in a data-driven way from formula \eqref{eq:beta}.
There remain two other hyperparameters $\alpha_{\zeta}$, $\alpha_{\psi}$, which we suggest to equip with a prior, and incorporate updates for these hyperparameters in the Gibbs sampler derived in Subsection~\ref{sec:gibbs}; cf.~\cite{dikmen09}. Taking for concreteness $\alpha_{\zeta}=\alpha_{\psi}$ (we will write $\alpha$ for simplicity, though this leads to a clash with our notation in Section~\ref{sec:likelihood} in the case of the independent gamma prior) and denoting this prior by $\pi$, we find that the joint density of $\alpha$, $\{\psi_k\}$, $\{\zeta_k\}$ is $\pi(\alpha) p(\psi_1) \prod_{k=2}^N p(\zeta_{k} \mid \psi_{k-1}) p(\psi_k \mid \zeta_k)$.
Using this formula, the unnormalised full conditional density of $\alpha$ given the remaining parameters $\{\psi_k\}$, $\{\zeta_k\}$ is seen to be
\begin{multline*}
q(\alpha) = \pi(\alpha) \times\left(\frac{\alpha^{\alpha}}{\Gamma(\alpha)}\right)^{2(N-1)}  \times \prod_{k=2}^N \left(\psi_{k-1} \psi_k \zeta_k^{-2}\right)^{\alpha} 
\times \exp\left(-\alpha\sum_{k=2}^N\frac{1}{\zeta_k}( {\psi_{k-1}} + {\psi_k})\right).
\end{multline*}
The corresponding normalised density is nonstandard. Hence the full conditional of $\alpha$ is not easily accessible. Therefore we will use a Metropolis-within-Gibbs step (see \cite{tierney94}) to update the parameter $\alpha$ when running the Gibbs sampler from Subsection~\ref{sec:gibbs}.

As $\alpha$ is nonnegative, we reparametrise $\alpha$ as $\widetilde{\alpha}=\log(\alpha)$ and note that the unnormalised full conditional density of $\widetilde{\alpha}$ is $\widetilde{\alpha} \mapsto e^{\widetilde{\alpha}}q(e^{\widetilde{\alpha}})=\widetilde{q}(\widetilde{\alpha})$. Once one designs an update algorithm for $\widetilde{\alpha}$, samples for $\alpha$ can be obtained by exponentiation.

We propose to use a Gaussian random walk proposal on $\widetilde{\alpha}$. As a prior on the parameter $\alpha$, we recommend to use a distribution with not too light right tail, next to having sufficient mass near zero. The former because large values of $\alpha$ may be necessary to adequately regularise the non-parametric estimation problem if the number of bins is large. On the other hand, if the prior puts mass close to zero, then the prior allows the method to choose a model similar to the independent gamma prior.

\subsection{Bin number selection}

There remains a choice to be made for the hyperparameter $N$, i.e.\ the number of bins. As with our first approach using independent gamma priors on the coefficients $\psi_k$'s, here too the bin number can in principle be optimised via the marginal likelihood. However, unlike our first approach, the marginal likelihood is not available in closed form. Although for any fixed $N$ it can be estimated from the posterior simulation output (see, e.g., \cite{chib95}, \cite{chib01} and references therein), this is far from trivial. However, as we demonstrate in Section~\ref{sec:simulations}, the inferences are quite stable with respect to the choice of $N$. In fact, for a fixed value of $N$, the GMC method will balance the amount of smoothing by tuning the hyperparameter $\alpha$ from the data. Such a shift from the discrete model selection problem to tuning a single continuous hyperparameter is often advantageous from the computational point of view.

In the examples which deal with small or moderate size datasets (of several hundred Poisson points), we use the rule-of-thumb
\begin{equation}
\label{eq:rule}
N = \min\left(50, \left \lceil \frac{\mathcal{H}}{4} \right \rfloor \right).
\end{equation}
Here
\[
\mathcal{H}=\sum_{j=1}^n \sum_{i=1}^{m_j} \ind_{\{X_{ij}\in [0,T]\}}
\]
is the total number of Poisson points. It is advisable not to make the bins $B_k$ too small. While increasing $\alpha$ gives to realisations from the prior more regularity, the exact trade-off between this and the number of bins is difficult to quantify. The upper truncation with $50$ in \eqref{eq:rule} is somewhat arbitrary. It was sufficient to obtain convincing results in our simulation examples. This upper bound can be easily modified to a larger number and does not involve a significant computational overhead. Cf.~our results in Subsection~\ref{sec:scalability}.

\section{Simulation examples}
\label{sec:simulations}

In this section we study the performance of our non-parametric Bayesian procedures on simulated data examples. We implemented our procedures in {\bf Julia}, see \cite{bezanson17}. The computer code and datasets for replication of our examples forms part of the PointProcessInference package, see \cite{pppjl}\footnote{For additional details see \url{https://github.com/mschauer/PointProcessInference.jl}}. For plotting we used functionalities of the {\bf{ggplot2}} package (see \cite{wickham09}) in {\bf R} (see \cite{R17}). The computations were performed on a MacBook Pro, with a 2.7GHz Intel Core i5 with 8 GB RAM. 

Full specifications used for synthetic data generation and posterior inference are given in each example  below. For the method based on the GMC prior, we ran the Gibbs sampler for $30000$ iterations. The first half of the generated posterior samples was then discarded as a  burn-in, and the posterior inference was based on the second half of the samples. A Gaussian random walk proposal was used to update the hyperparameter $\alpha$, with variance scaled in such a way so as to ensure the acceptance rate between $25\%-50\%$ in the Metropolis-within-Gibbs step.

In general, in the ensuing plots the ``true" intensity function is represented by a solid red line. The posterior mean is given by a solid black line, while $95\%$ marginal credible bands are shaded in light blue.

\subsection{Oscillating exponential function}
\label{subsec:exponential}
In our first example, we consider 
\[
\lambda_0(x)=2 e^{-x/5} (5+ 4\cos(x)) , \quad x\in[0,10].
\]
A principal challenge in inferring this function consists in the fact that it takes small values in the middle part of its domain and has a slope of changing sign.  In Figure~\ref{fig:example1n1} we plot our estimation results with sample size $n=1$ (the total number of Poisson points was $\mathcal{H}=46$). In this figure, as well as in subsequent ones, the rug plot on the top displays the event times. For the independent gamma prior we took $\alpha=\beta=0.1$. For the GMC prior we took $\alpha_1=\beta_1=0.1$ and $\alpha\sim\operatorname{Exp}(0.1)$. For the independent gamma prior, the optimal choice for $N$ obtained by maximising the  log marginal likelihood was $N=4$. See the top panels in  Figure~\ref{fig:example1n1}. In Figure~\ref{fig:expcos-mll} a plot of this log marginal likelihood (ignoring the irrelevant term $e^{Tn}$) versus $N$ is shown. Quick convergence of the Gibbs sampler can be seen from the trace and autocorrelation plots for several parameters in Figure~\ref{fig:sampler}.\begin{figure}
	\begin{center}
		\includegraphics[width=0.9\textwidth]{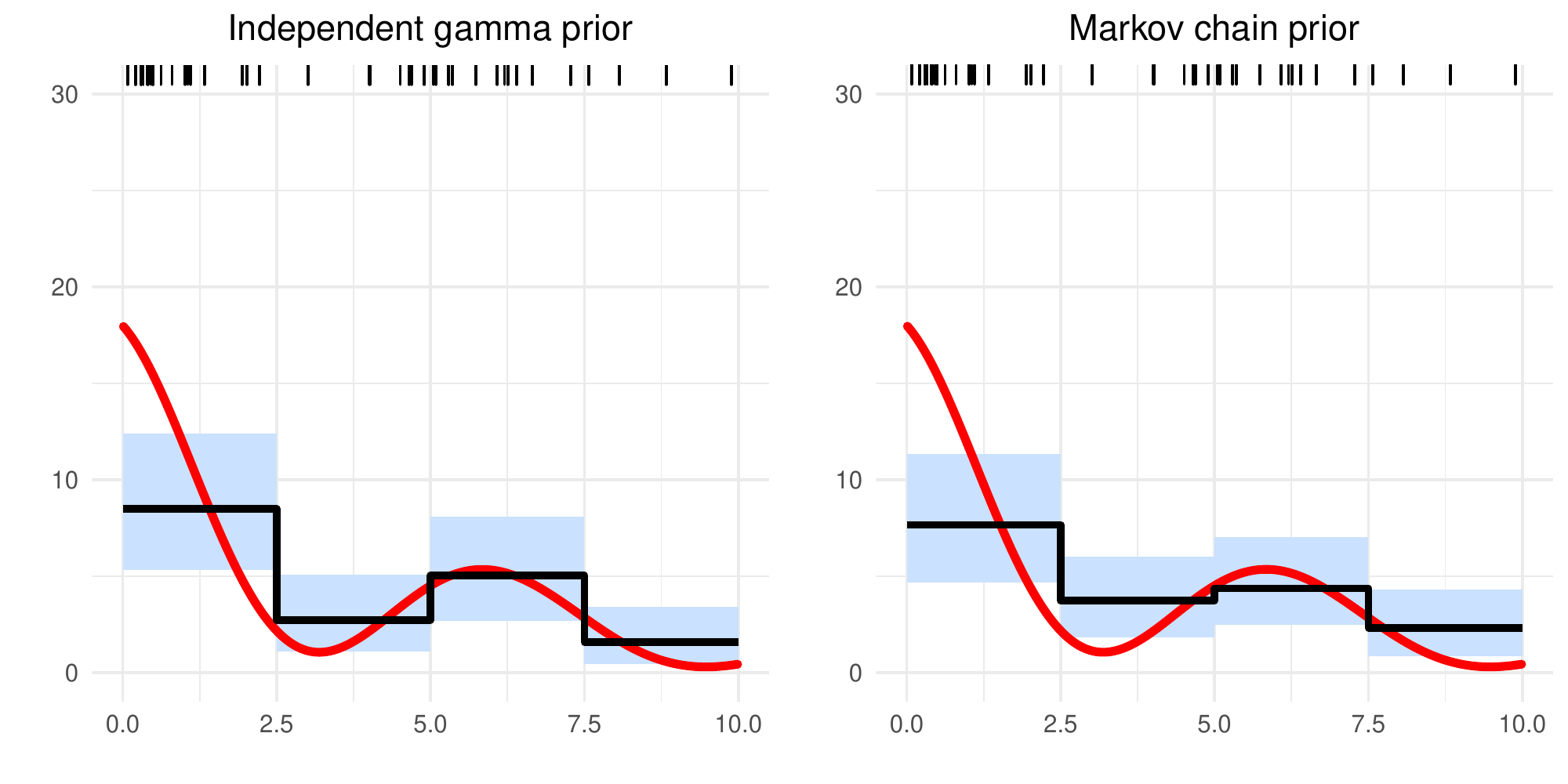}
		\includegraphics[width=0.9\textwidth]{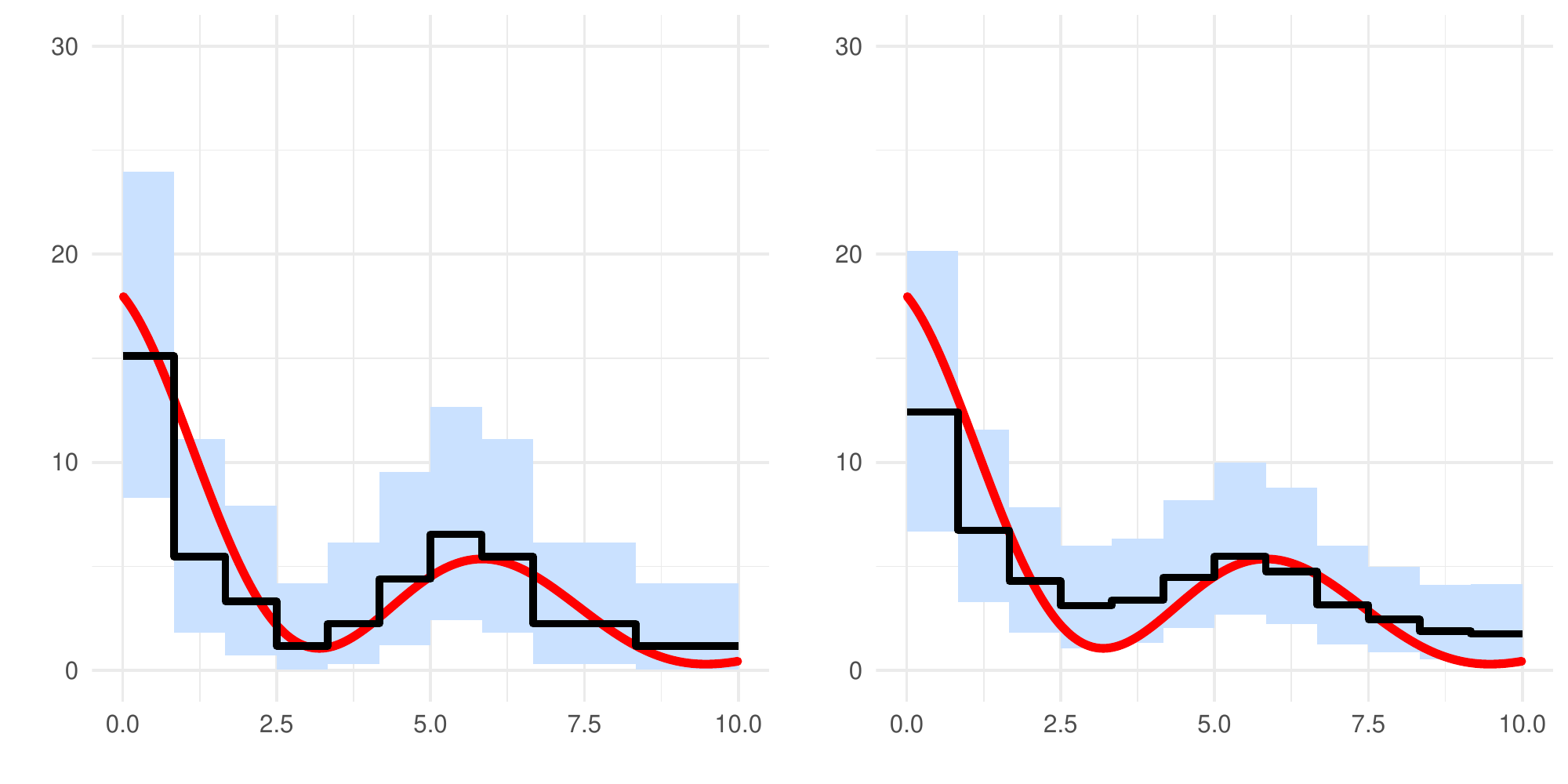}
		\includegraphics[width=0.9\textwidth]{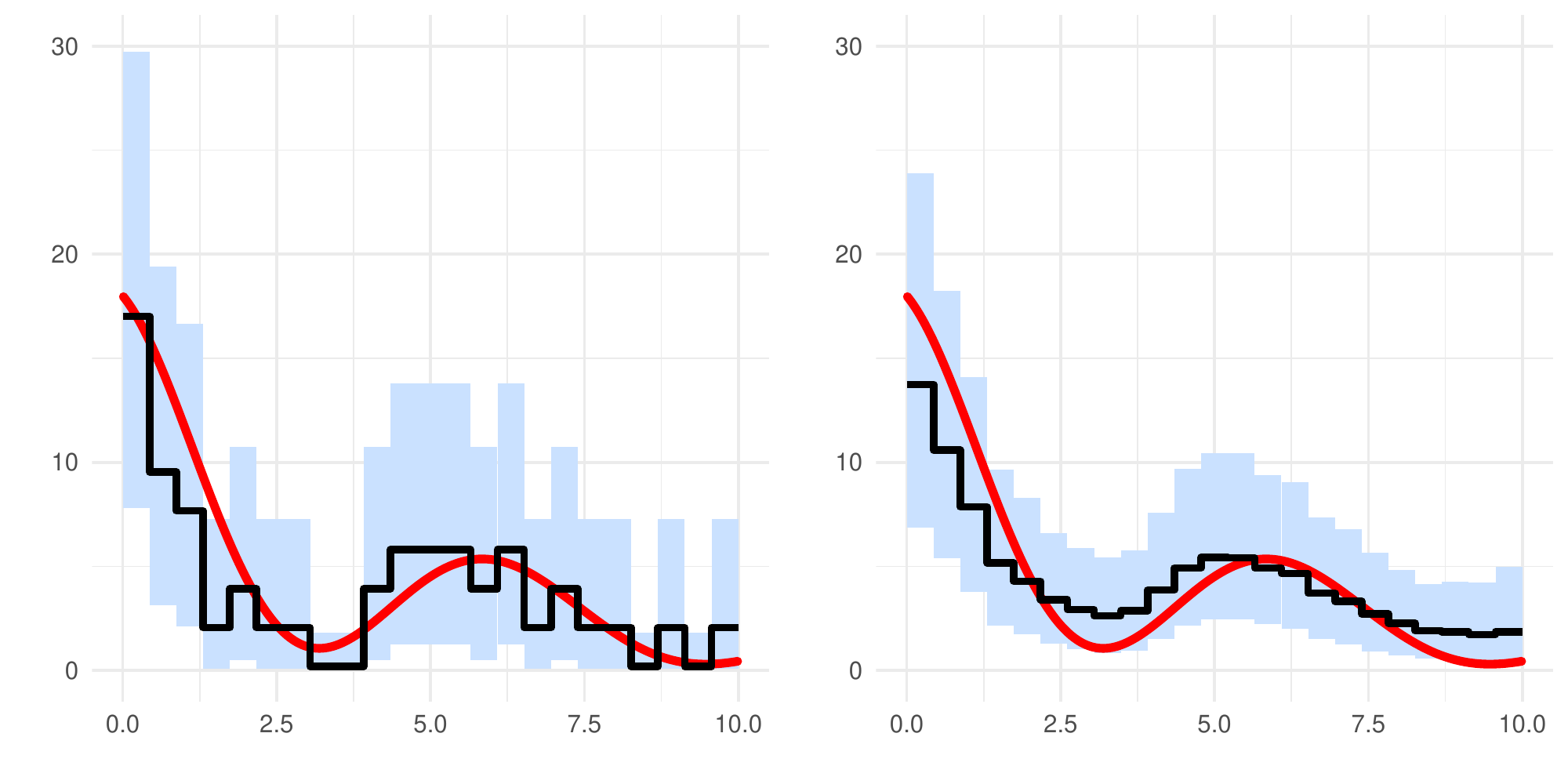}
	\end{center}
	\caption{Estimation results for the oscillating exponential function $\lambda_0$ from Subsection \ref{subsec:exponential} with $n=1$. In the top row, $N=4$ is selected via maximising the marginal likelihood as in Subsection \ref{sec:ebayes}, in the middle row $N=12$ via \eqref{eq:rule}, and in the bottom row $N=\mathcal{H}/2=23$. The prior settings were: $\alpha=0.1$ and $\beta$ determined from \eqref{eq:beta} for the independent gamma prior, and $\alpha_1=\beta_1=0.1$ and $\alpha\sim\operatorname{Exp}(0.1)$ for the GMC prior.}
	\label{fig:example1n1}
\end{figure}

\begin{figure}
	\begin{center}
		\includegraphics[width=0.9\textwidth]{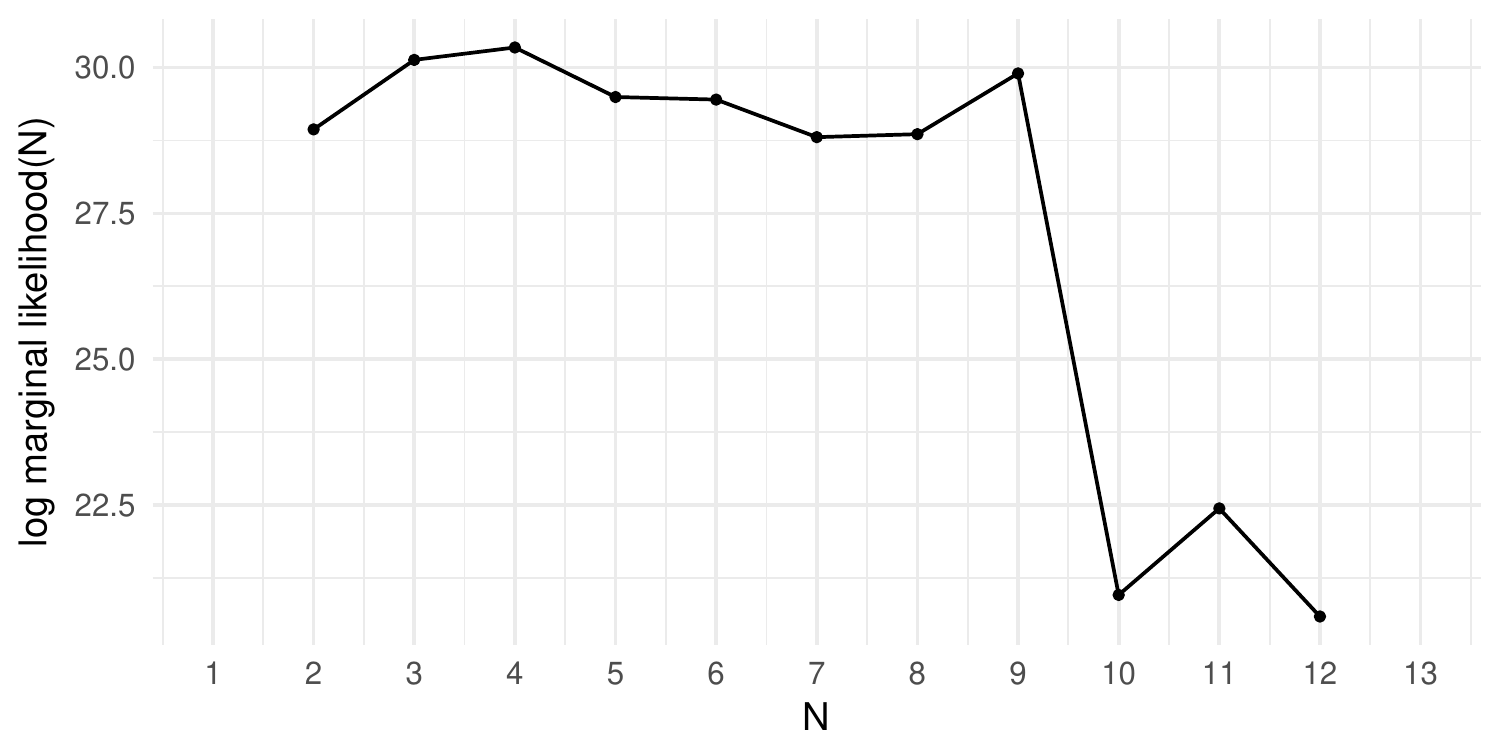}
	\end{center}
	\caption{Simulation example of the oscillating exponential function $\lambda_0$ from Subsection \ref{subsec:exponential} with $n=1$: logarithm of the marginal likelihood as a function of $N$ (up to an additive constant independent of $N$), with $\alpha=\beta=0.1$. The maximum is attained at $N=4$.}
	\label{fig:expcos-mll}
\end{figure}

\begin{figure}
	\begin{center}
		\includegraphics[width=0.9\textwidth]{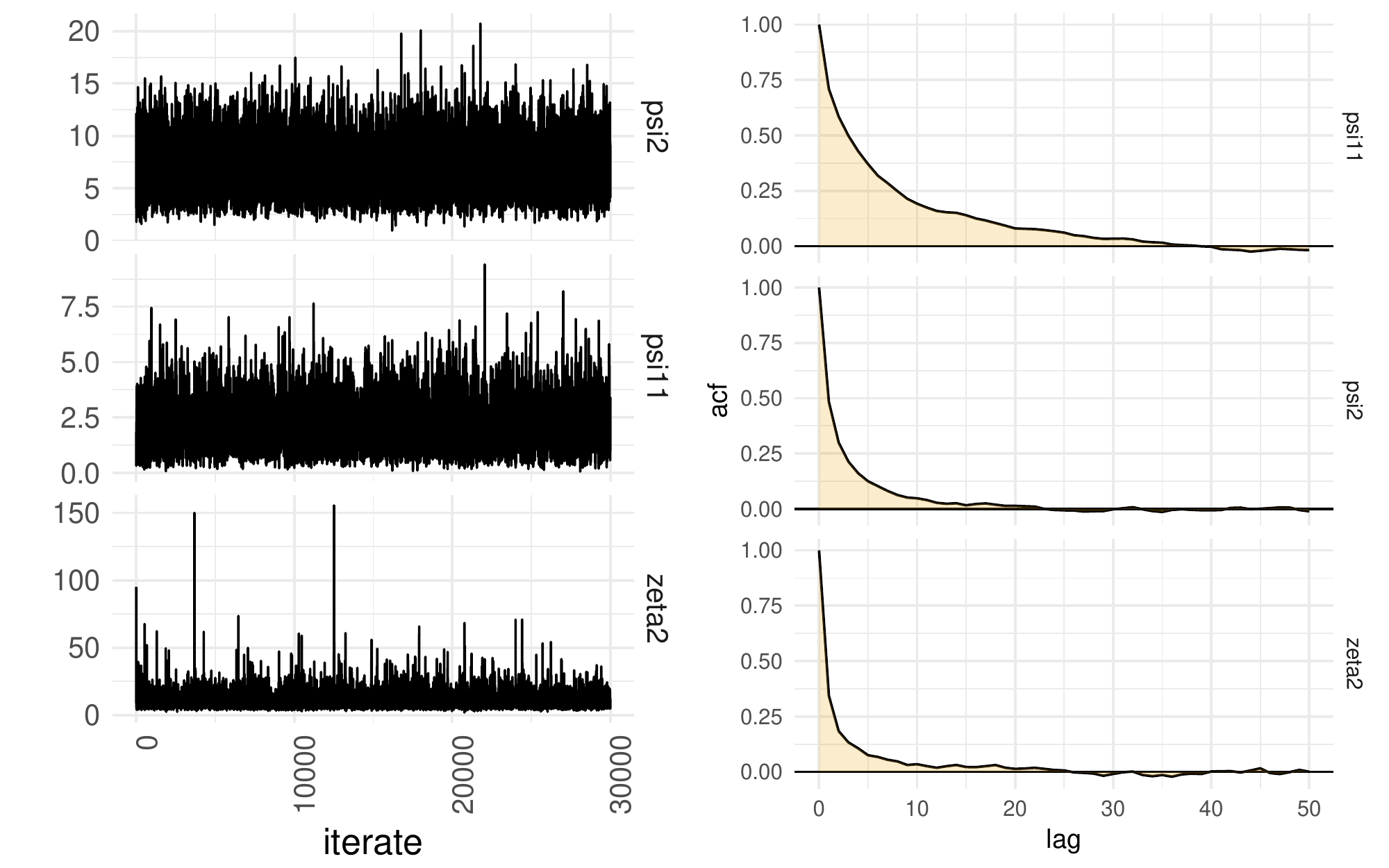}
		\includegraphics[width=0.9\textwidth]{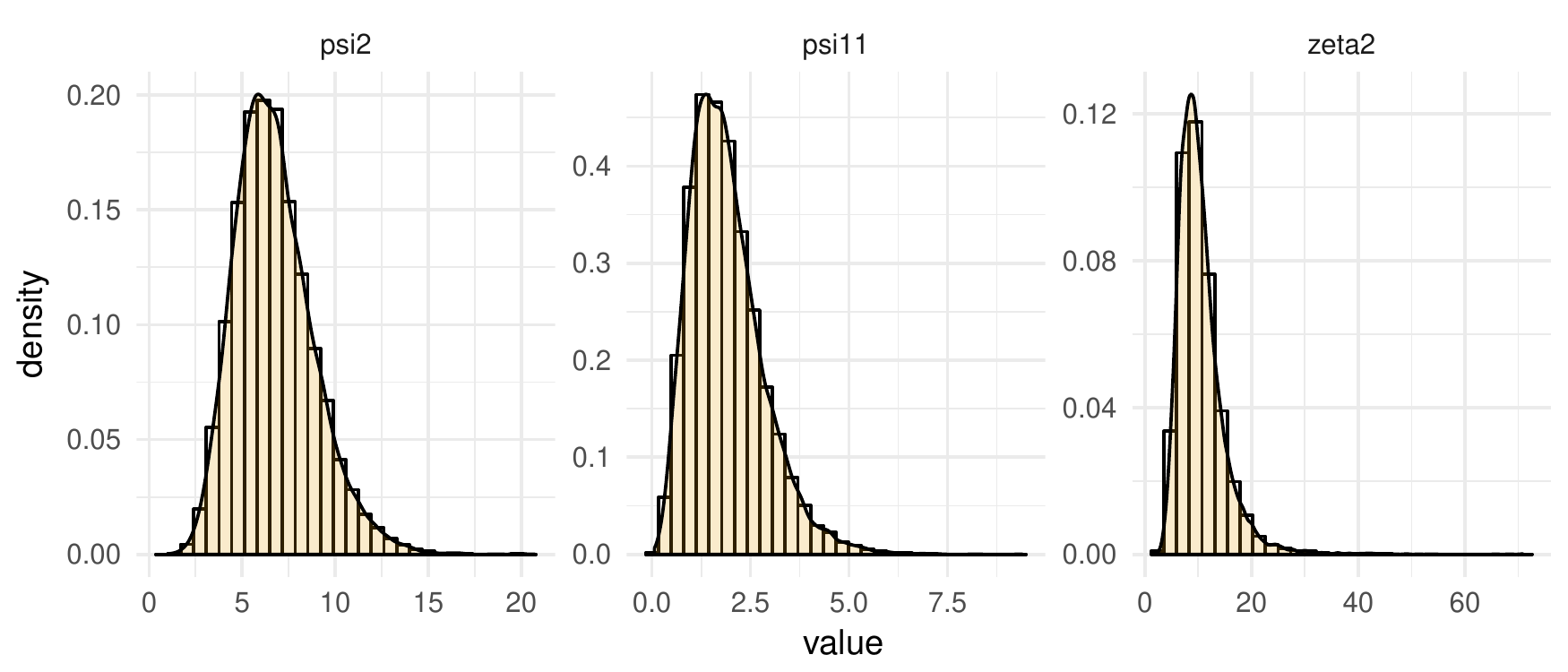}
	\end{center}
	\caption{Posterior simulation output for estimating the oscillating exponential function $\lambda_0$ from Subsection \ref{subsec:exponential}, with $n=1$. Trace plots, autocorrelation plots and histograms for $\psi$ and $\zeta$. For the histograms and autocorrelation plots the first half of the samples has been discarded as burn-in.}
	\label{fig:sampler}
\end{figure}

\begin{figure}
	\begin{center}
		\includegraphics[width=0.9\textwidth]{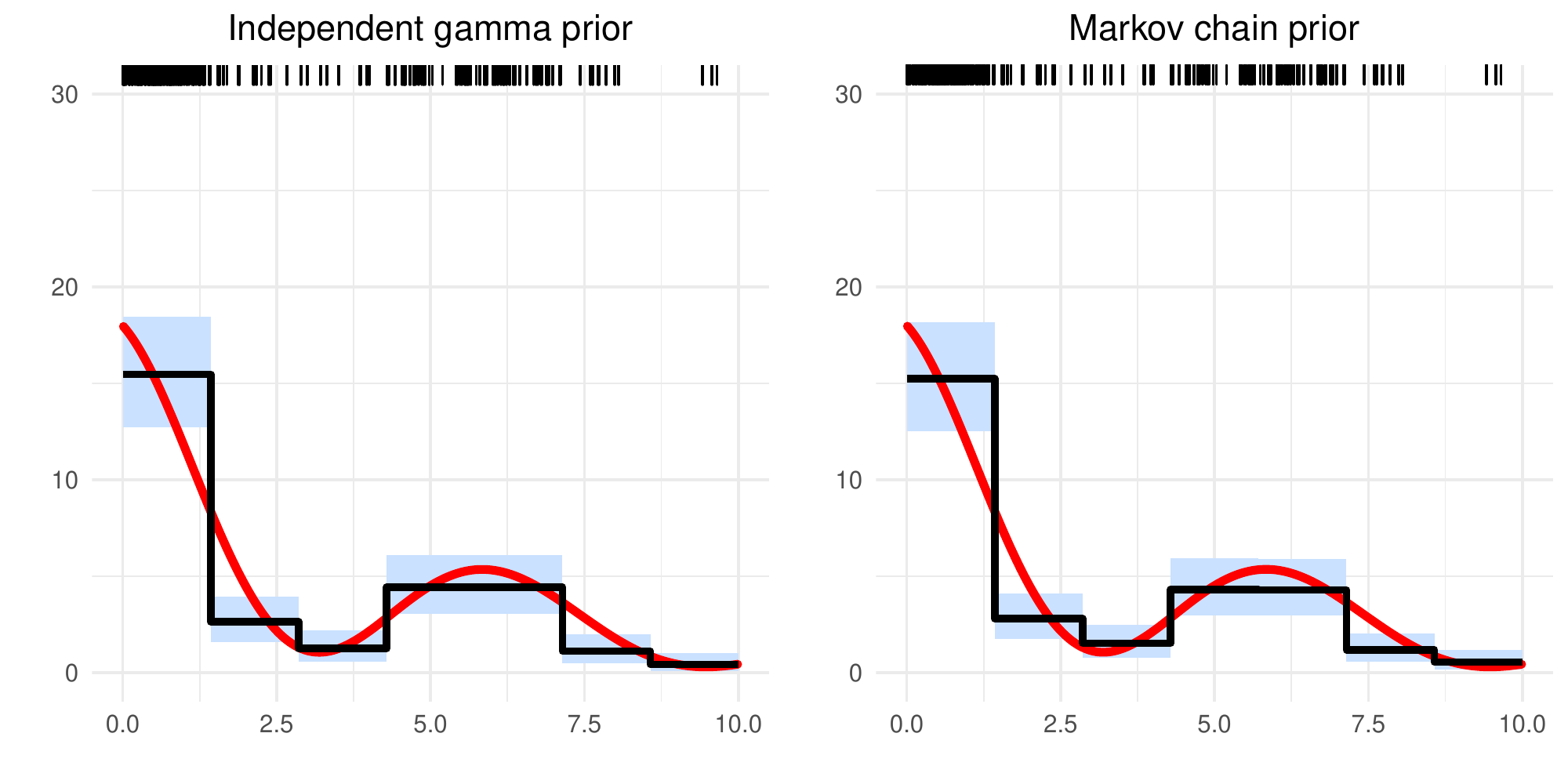}
		\includegraphics[width=0.9\textwidth]{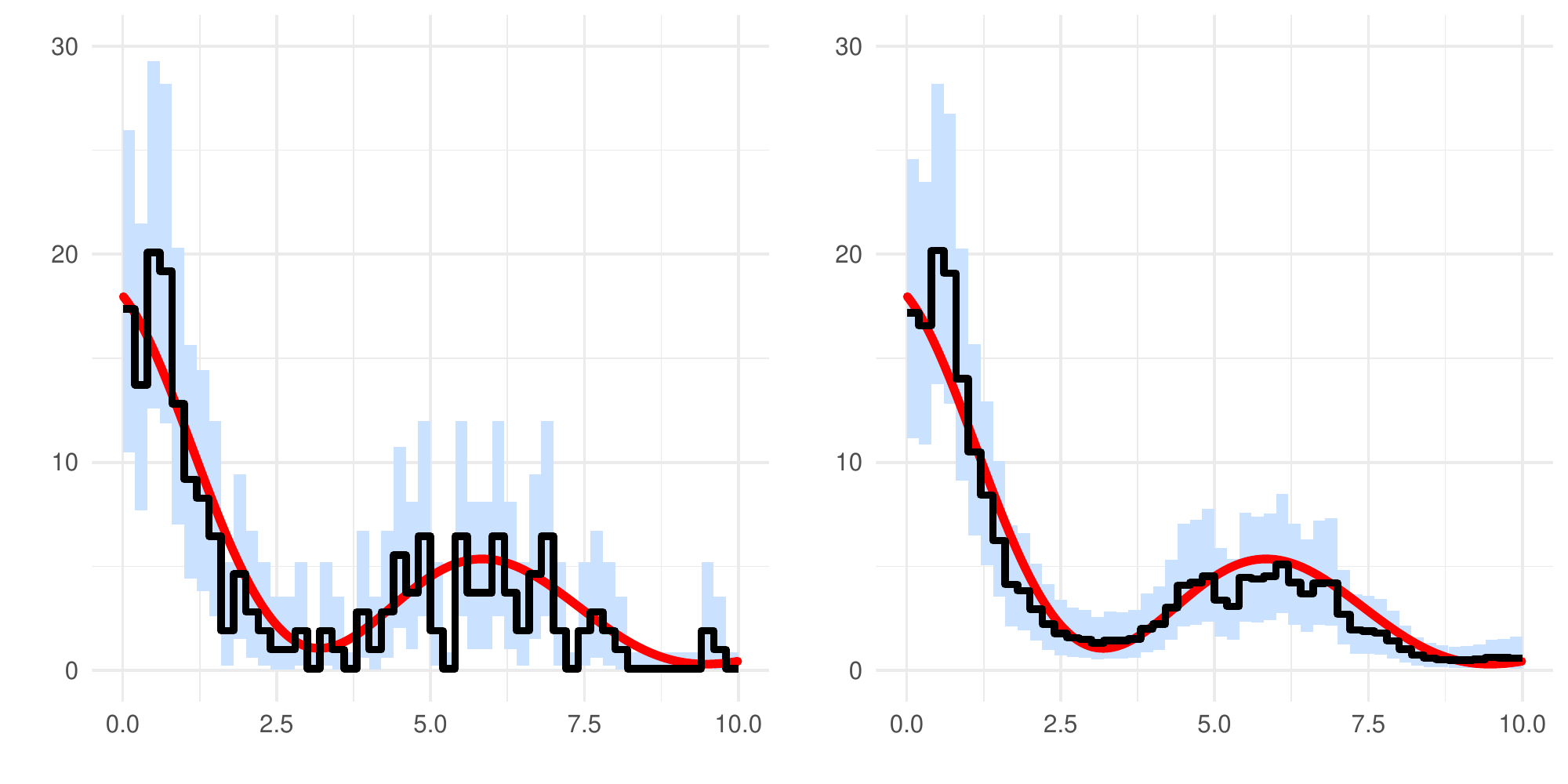}
	\end{center}
	\caption{Estimation results for the oscillating exponential function $\lambda_0$ from Subsection \ref{subsec:exponential} with $n=5$. The top plot corresponds to the method based on the independent gamma prior, with $N=7$ chosen via the empirical Bayes method. The bottom plot corresponds to the method based on the GMC prior, with $N=50$ chosen via \eqref{eq:rule}. Other settings are as in the case $n=1$ as in Figure~\ref{fig:example1n1}.}
	\label{fig:example1n5}
\end{figure}
Estimation results for the sample size $n=5$ are reported in Figure~\ref{fig:example1n5} (the total number of Poisson points was $\mathcal{H}=215$). Here in the top panel $N$ for the independent gamma prior was determined using the empirical Bayes method.  

The following conclusions can be gleaned from the simulation results:
\begin{itemize}
	\item Performance of the empirical Bayes method for selecting $N$ is not particularly encouraging. For moderate sample sizes that we considered, it oversmooths by selecting a visually too small $N$.
	\item For larger $N$, posterior means obtained with the independent gamma prior tend to show more fluctuation than those obtained with the GMC prior.
	\item For larger $N$, marginal posterior bands with the GMC method tend to be narrower than those obtained with the independent gamma prior.
	\item Inferential conclusions with the GMC prior, as reflected in marginal posterior bands, appear to be robust with respect to the choice of $N$, provided this is not chosen exceedingly large or small. In particular, the rule-of-thumb \eqref{eq:rule} appears to work in practice. On the other hand, this robustness property is not shared by the method based on the independent gamma prior.
\end{itemize}

If the independent gamma prior for $\psi_k$ is chosen in a more informative way, determining $N$ via the empirical Bayes method seems to perform better compared to the case when a diffuse gamma prior is used on $\psi_k$. In Figure \ref{fig:example1n1ebayes} we chose $\alpha=2$ and $\beta=1$, which led to the optimal value for $N$ being equal to $9$ (which is close to the value $12$, chosen by the rule-of-thumb \eqref{eq:rule} for the GMC prior). Interestingly enough, in this case the posterior band obtained with the independent gamma prior is somewhat narrower than the one with the GMC prior, although visually the latter reflects uncertainty in estimation results better than the former. A difficulty with the empirical Bayes method is that a sensible informative prior on $\psi_k$ might not be available in practice.
\begin{figure}
	\begin{center}
		\includegraphics[width=0.9\textwidth]{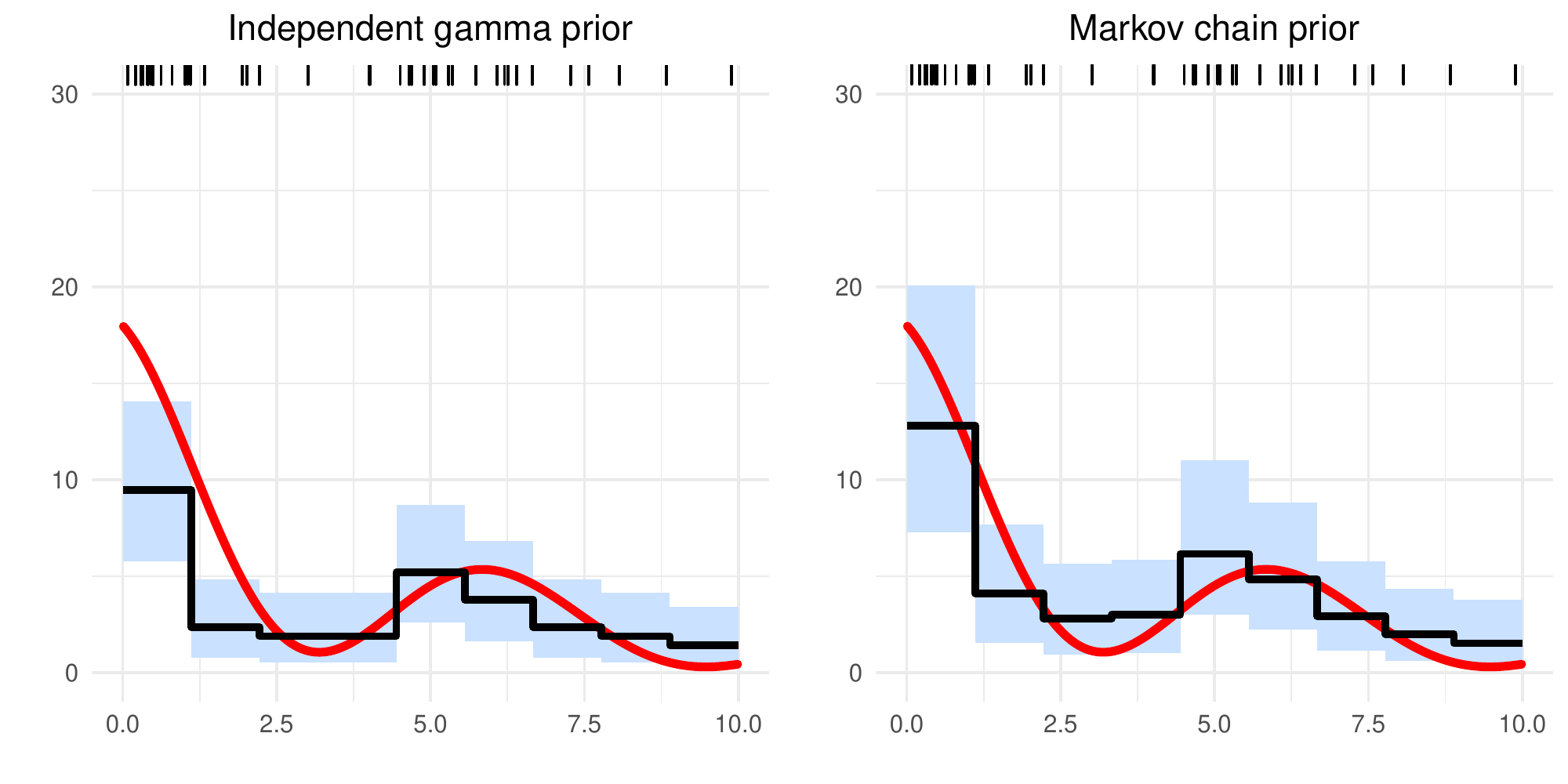}
		\includegraphics[width=0.9\textwidth]{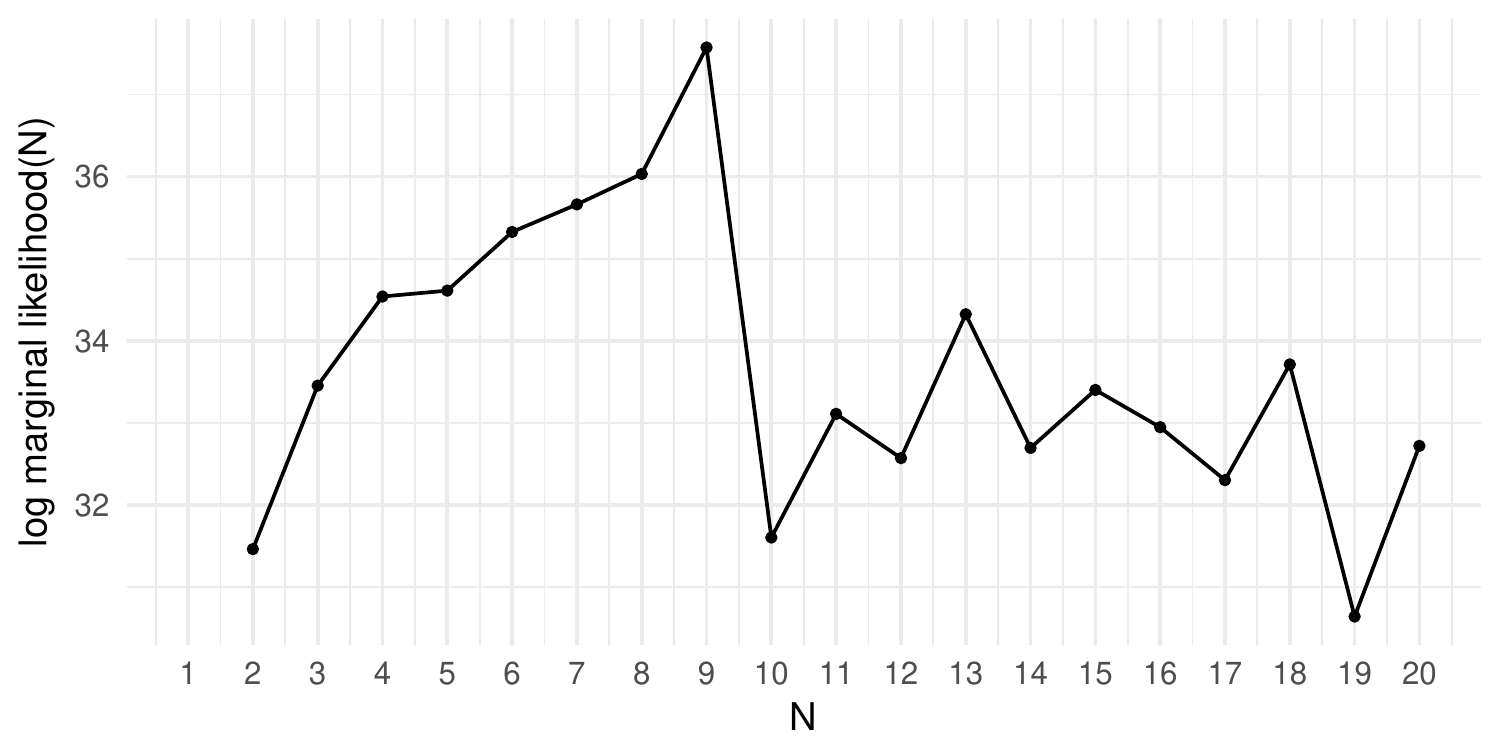}
	\end{center}
	\caption{Estimation results for the oscillating exponential function $\lambda_0$ from Subsection \ref{subsec:exponential} with $n=1$. The prior distribution on each $\psi_k$ is taken to be $\operatorname{G}(2,1)$.  The log marginal likelihood is plotted as a function of $N$ in the bottom panel (up to an additive constant independent of $N$). The optimal number of bins chosen via the empirical Bayes method is $N=9$.}
	\label{fig:example1n1ebayes}
\end{figure}

\subsubsection{Scalability}
\label{sec:scalability}
Our next goal was to illustrate scalability of our methods with big data. In Figures~\ref{fig:example1n4000ebayes}
and \ref{fig:example1n4000} we plot estimation results with a very large sample size $n=4000$, that resulted in $\mathcal{H}=177781$ Poisson points (we omit posterior means from the plots, as they obfuscated the resulting (narrow) marginal credible bands). As noted e.g.\ in \cite{adams09} and \cite{rao11}, samples of this size are far beyond the computational reach of their methods. On the other hand, both our methods perform excellently in terms of estimation quality. That the method based on the independent gamma prior is very fast (for Figure~\ref{fig:example1n4000ebayes} we used $N=48$) comes as no surprise. We also timed the method based on the GMC prior, specifically its part for running the Gibbs sampler, which task was completed in ca.~4 seconds for $N=200$ (used for Figure~\ref{fig:example1n4000}). 

In this example, given the huge amount of data, it makes sense to experiment with a large number of bins. We took $N=1000$ and compared two choices for a prior on $\alpha$ in Figure \ref{fig:example1n4000-1000bins}. In the right panel we took the L\'evy distribution (which is just the $\operatorname{IG}(1/2,1/2)$ distribution), while in the left panel we used the $\operatorname{Exp}(10)$ distribution as prior on $\alpha$. These distributions are very different, but only after zooming in and carefully studying the credible bands can one see a small difference: the band produced with the  L\'evy distribution appears to be smoother. These results were expected: a larger value for $\alpha$ induces a stronger positive correlation in a realisation from the gamma Markov chain prior. With $1000$ bins, the computing time was approximately $9$ seconds.  We conclude that our methods scale well with sample size.

\begin{figure}
	\begin{center}
		\includegraphics[width=0.9\textwidth]{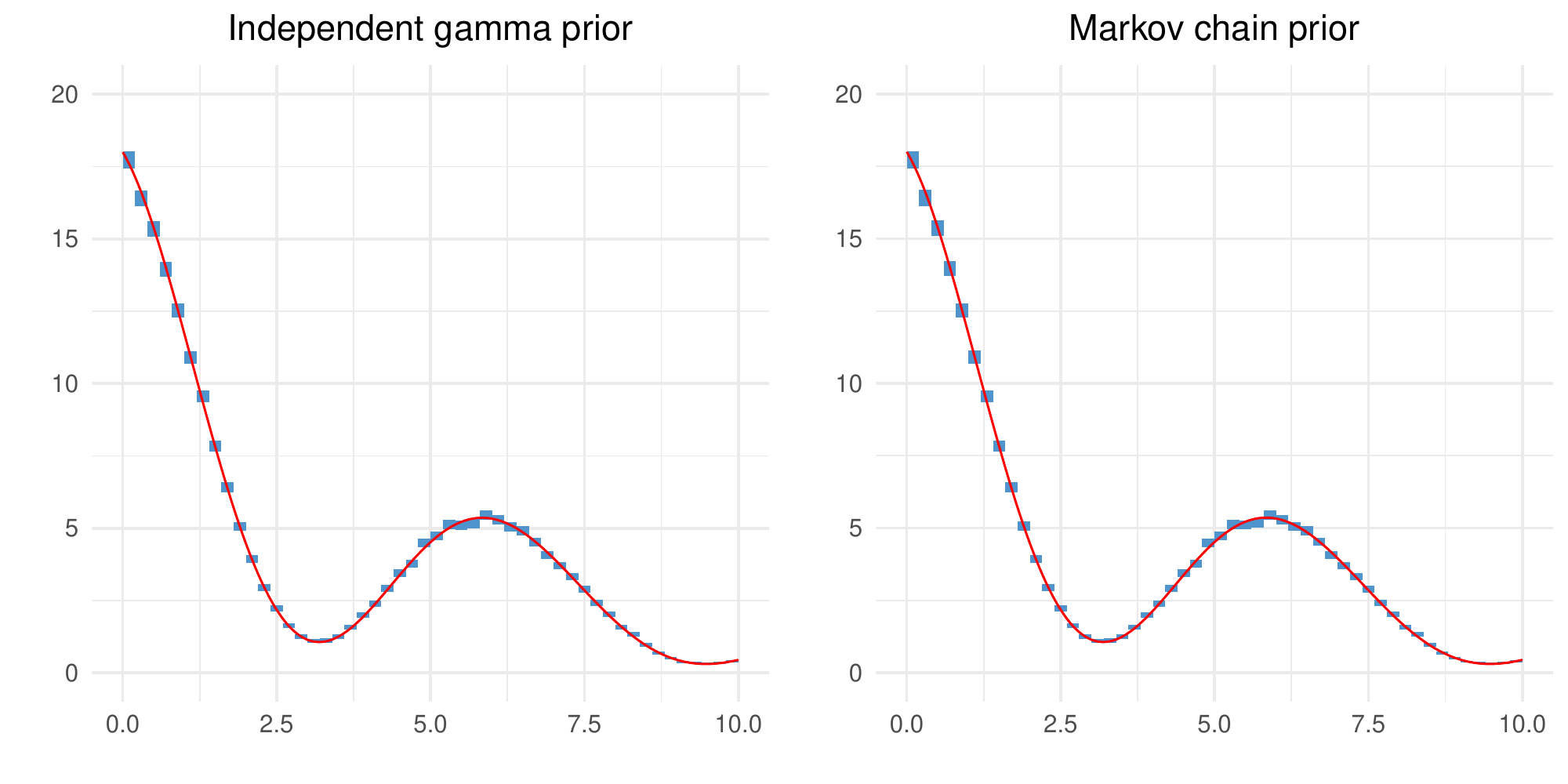}
		\includegraphics[width=0.9\textwidth]{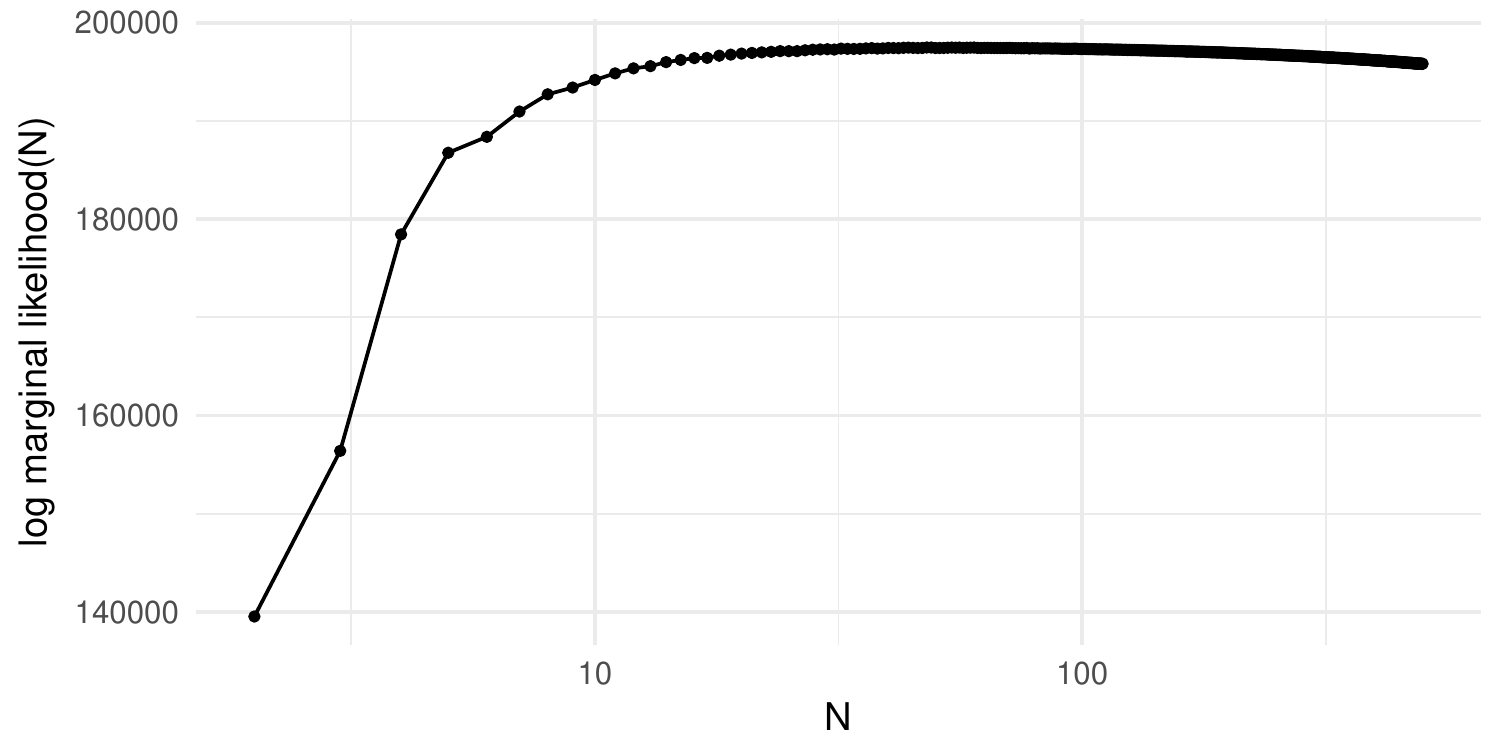}
	\end{center}
	\caption{Estimation results for the oscillating exponential function $\lambda_0$ from Subsection \ref{subsec:exponential} with $n=4000$. The bin number $N=48$ chosen via the empirical Bayes method, other hyperparameters determined as in the case $n=1$. The top panel gives the plots of $\lambda_0$ and the $95\%$ marginal credible bands. The bottom panel gives the plot of the log marginal likelihood as a function of $N$ (up to an additive constant independent of $N$).}\label{fig:example1n4000ebayes}
\end{figure}

\begin{figure}
	\begin{center}
		\includegraphics[width=0.9\textwidth]{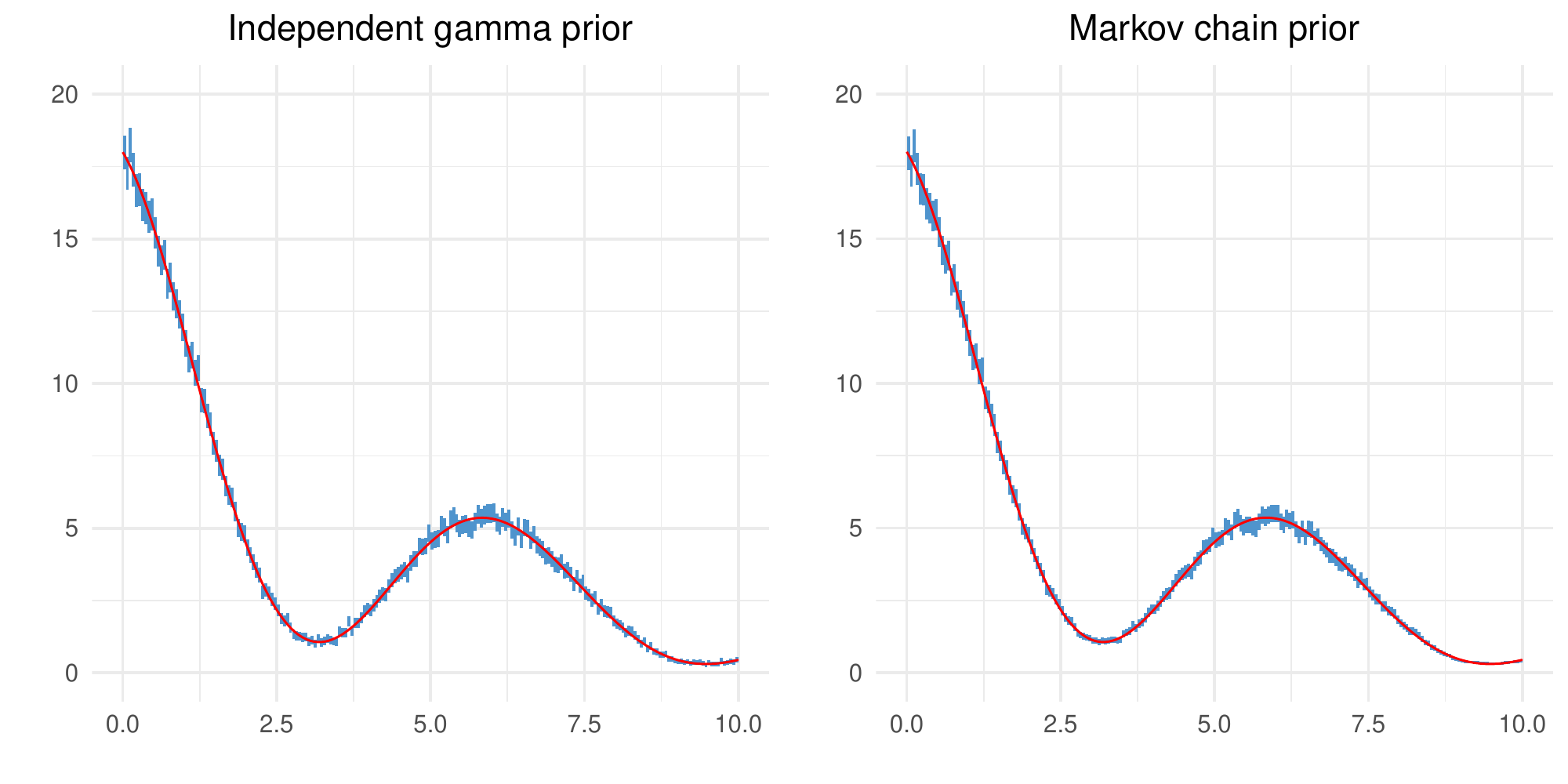}
	\end{center}
	\caption{Estimation results for the oscillating exponential function $\lambda_0$ from Subsection \ref{subsec:exponential} with $n=4000$. $N=200$, other hyperparameters determined as in the case $n=1$.}
	\label{fig:example1n4000}
\end{figure}

\begin{figure}
	\begin{center}
		\includegraphics[width=0.45\textwidth]{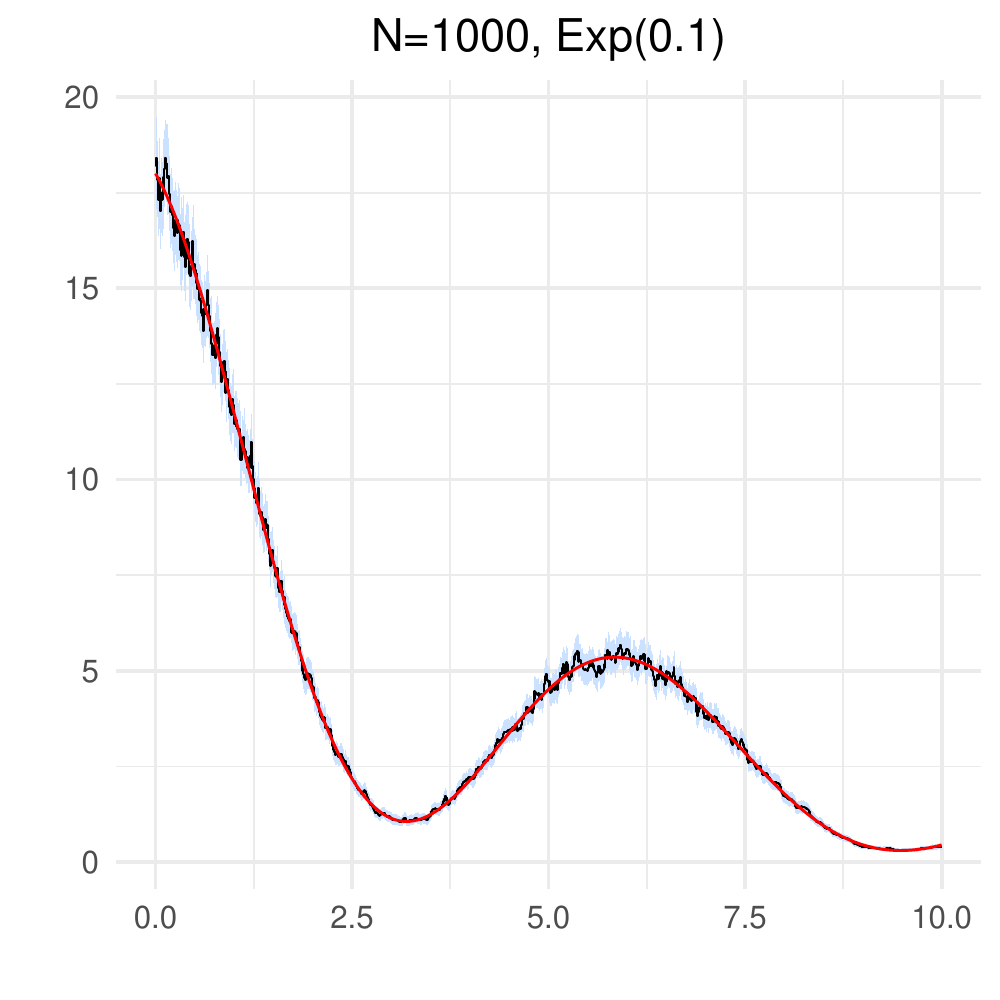}
		\includegraphics[width=0.45\textwidth]{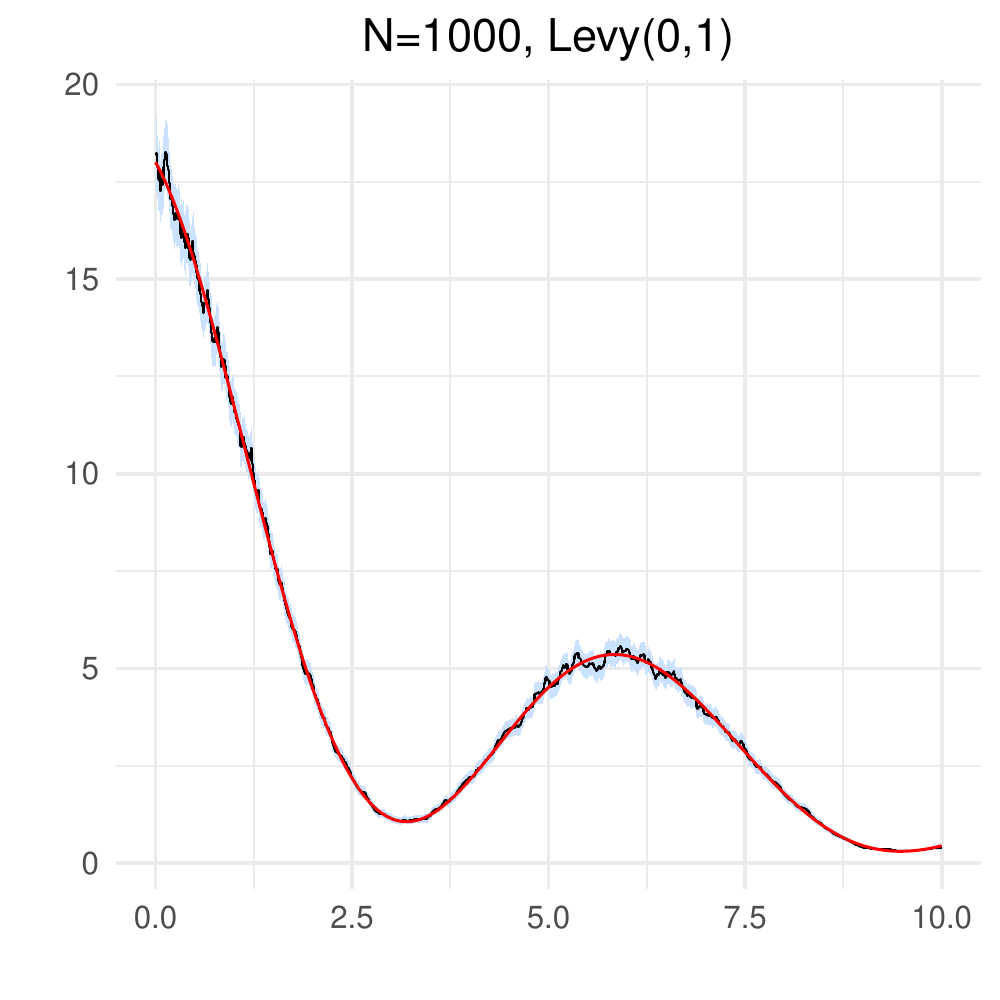}
	\end{center}
	\caption{Estimation results for the oscillating exponential function $\lambda_0$ from Subsection \ref{subsec:exponential} with $n=4000$ using the gamma Markov chain prior. The bin number was fixed at $N=1000$. Both plots show $\lambda_0$ and  $95\%$ marginal credible bands, where the true intensity is the red curve. In the left panel we took the $\operatorname{Exp}(10)$ distribution as prior on $\alpha$, whereas in the right panel we used the standard L\'evy distribution as prior on $\alpha$. }\label{fig:example1n4000-1000bins}
\end{figure}

\subsubsection{Sensitivity}
In Figure~\ref{fig:sensitivity} we illustrate sensitivity of the GMC procedure with respect to the choice of the prior on $\alpha$ by comparing estimation results with three different choices of the prior (we took  $N=12$, which follows from \eqref{eq:rule}). All three priors lead to qualitatively similar estimation results. In in all cases the GMC-based procedure calibrates the parameter $\alpha$ from the data, as evidenced by comparison of the prior density of $\alpha$ to the corresponding posterior density. Cf.~also Figure~\ref{fig:example1n4000-1000bins}.

\begin{figure}
	\begin{center}
		\includegraphics[width=0.3\textwidth]{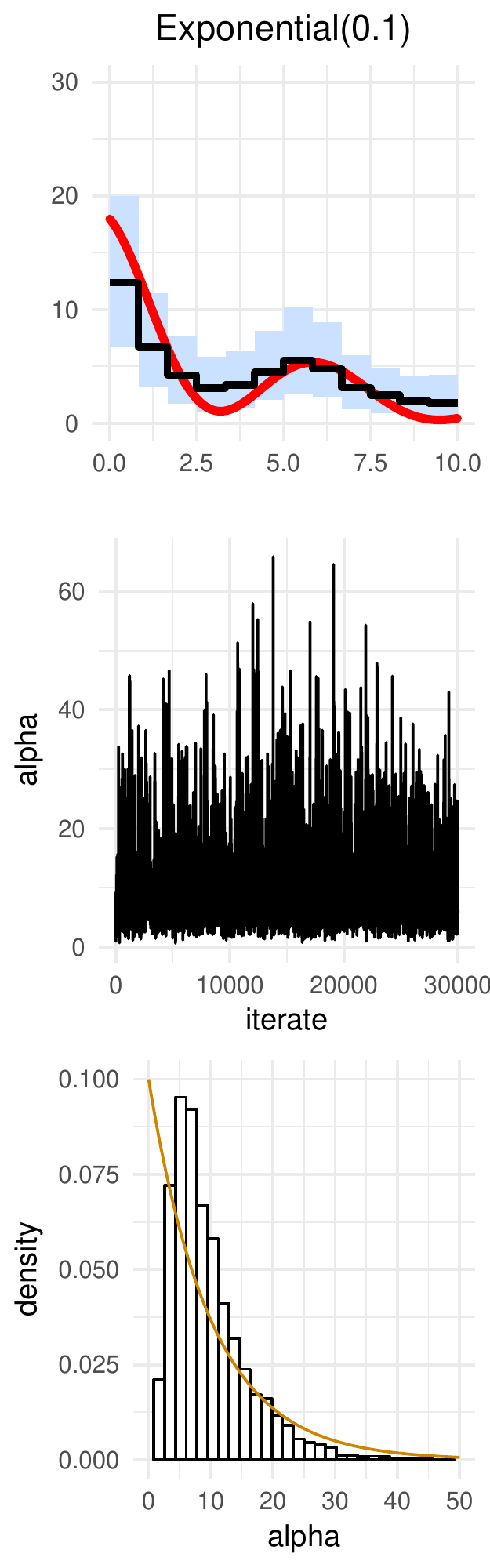}
		\includegraphics[width=0.3\textwidth]{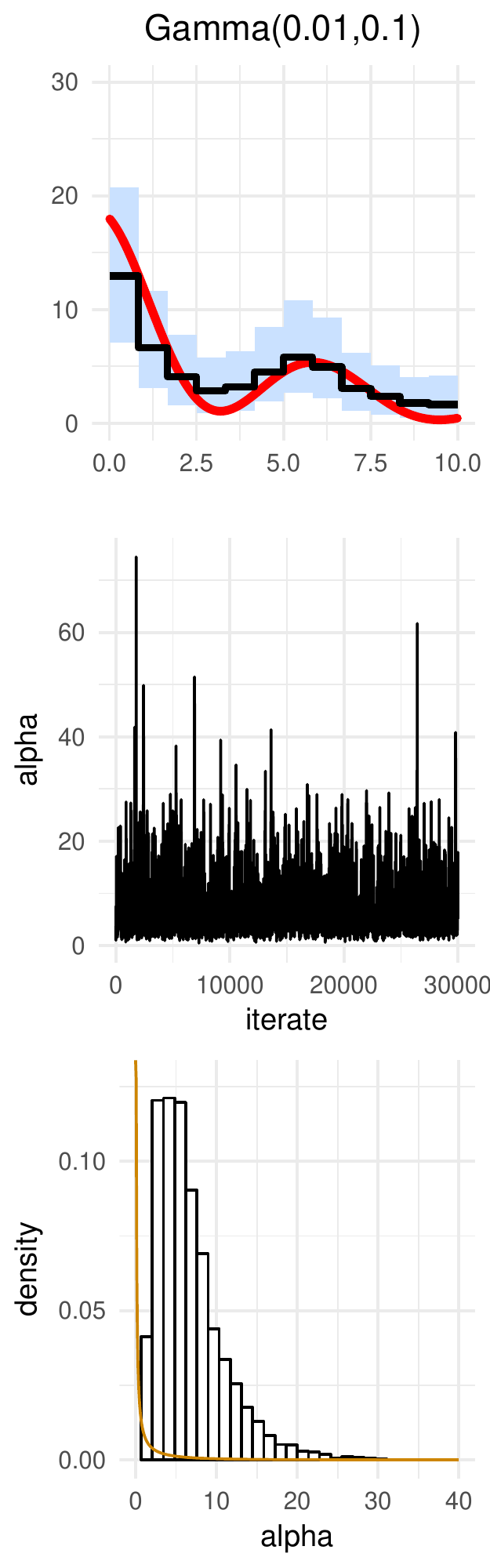}
		\includegraphics[width=0.3\textwidth]{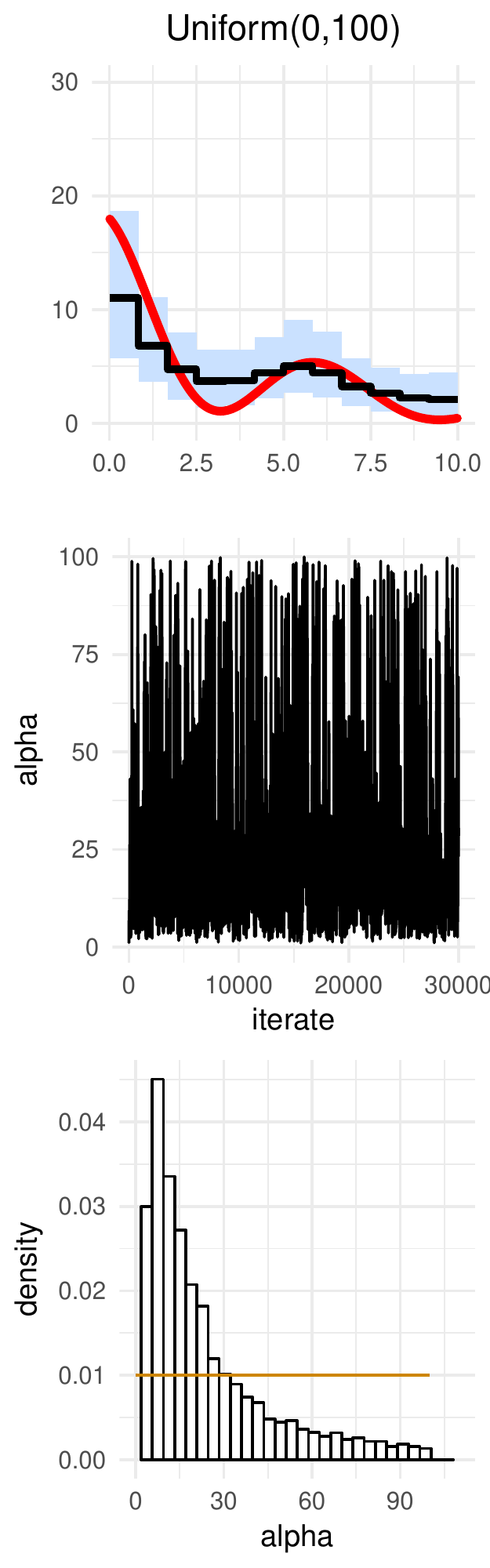}
	\end{center}
	\caption{Estimation results for the oscillating exponential function $\lambda_0$ from Subsection \ref{subsec:exponential} with $n=1$ using the GMC prior. The hyperparameters were $N=12$ and $\alpha_1=\beta_1=0.1$. The plots show the effect the prior on $\alpha$ has on inference results. In the left column, $\alpha\sim\operatorname{Exp}(0.1)$, in the middle $\alpha\sim\operatorname{G}(0.01, 0.1)$, and in the right $\alpha\sim\operatorname{Uniform}(0,100)$.  Top row: true intensity function (red solid), posterior mean (black solid) and $95\%$ marginal credible bands (light blue). Middle row: trace plots of $\alpha$. Bottom row: histogram of posterior samples of $\alpha$ (first half of the samples was discarded as a burn-in), with a prior density of $\alpha$ is added in orange.}
	\label{fig:sensitivity}
\end{figure}

\subsection{Bart Simpson function}
\label{sec:simpson}
As our second example, we consider the Bart Simpson intensity function defined as
\[
\lambda_0(x) = \frac{1}{2} \phi(x;3,1)+\frac{1}{10} \sum_{j=0}^4 \phi\left( x ; \frac{j}{2} +2, \frac{1}{10} \right), \quad x\in[0,6].
\]
This example has been adapted from the non-parametric density estimation context in Chapter 6 in \cite{wasserman06}. The Bart Simpson function is a mixture of Gaussian densities, interpreted as a half times the density $\phi(x;3,1)$  with a number of superimposed peaks. The latter can be thought of as corresponding to bursts of increased activity of a Poisson point process.

We give estimation results with sample sizes $n=200$ and $n=500$ in Figure~\ref{fig:bart200} (the number of Poisson points was $\mathcal{H}=200$ and $\mathcal{H}=491$, respectively) for both the independent gamma prior and the GMC prior. Additionally, in Figure~\ref{fig:bart:ebayes} we report the results when the number of bins with the independent gamma prior is selected via the empirical Bayes method. We note that the same optimal number of bins $N=7$ was obtained when we used a more informative prior $\psi_k \iid \operatorname{G}(2,1)$ too. We do not provide the corresponding plots, as they did not indicate a message different from Figure~\ref{fig:bart:ebayes}.

The following conclusions can be deduced from the figures:
\begin{itemize}
	\item As in the case of the oscillating exponential function from Subsection~\ref{subsec:exponential}, the empirical Bayes method to select $N$ performs visually unsatisfactorily.
	\item Posterior means of both of our estimation methods pick up nicely the peaks and valleys of the Bart Simpson function (provided for the independent gamma prior we use a sufficiently large number of bins $N$; both for the independent gamma prior and the GMC prior $N$ is chosen using our rule-of-thumb~\eqref{eq:rule}).  For the sample size $n=200$ and $n=500$, the posterior mean corresponding to the independent gamma prior does this somewhat better than the one corresponding to the GMC prior. On the other hand, the marginal credible bands for the GMC prior are tighter, and they also show far less spurious variability.
\end{itemize}

\begin{figure}
	\begin{center}
		\includegraphics[width=0.9\textwidth]{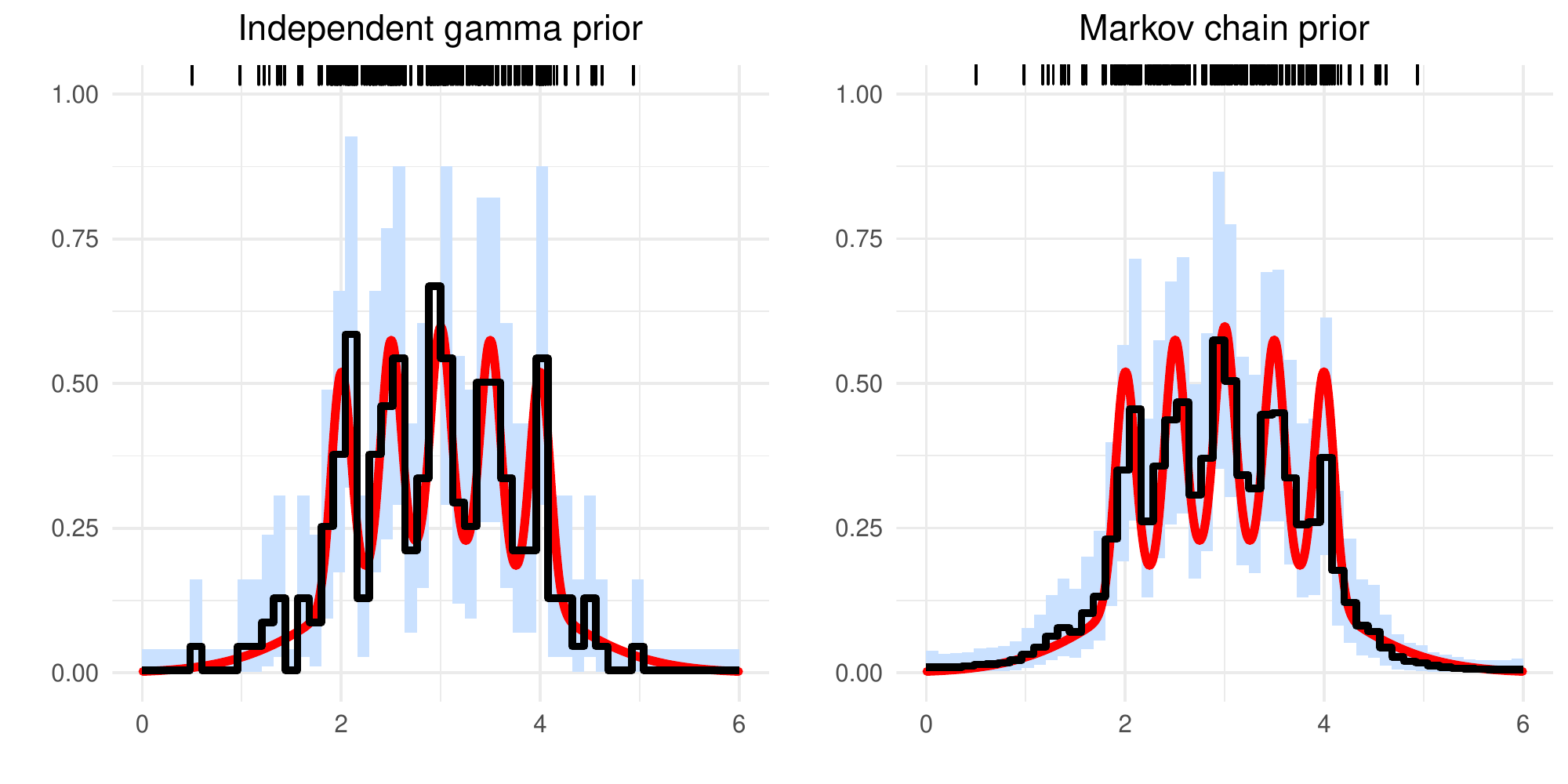}
		\includegraphics[width=0.9\textwidth]{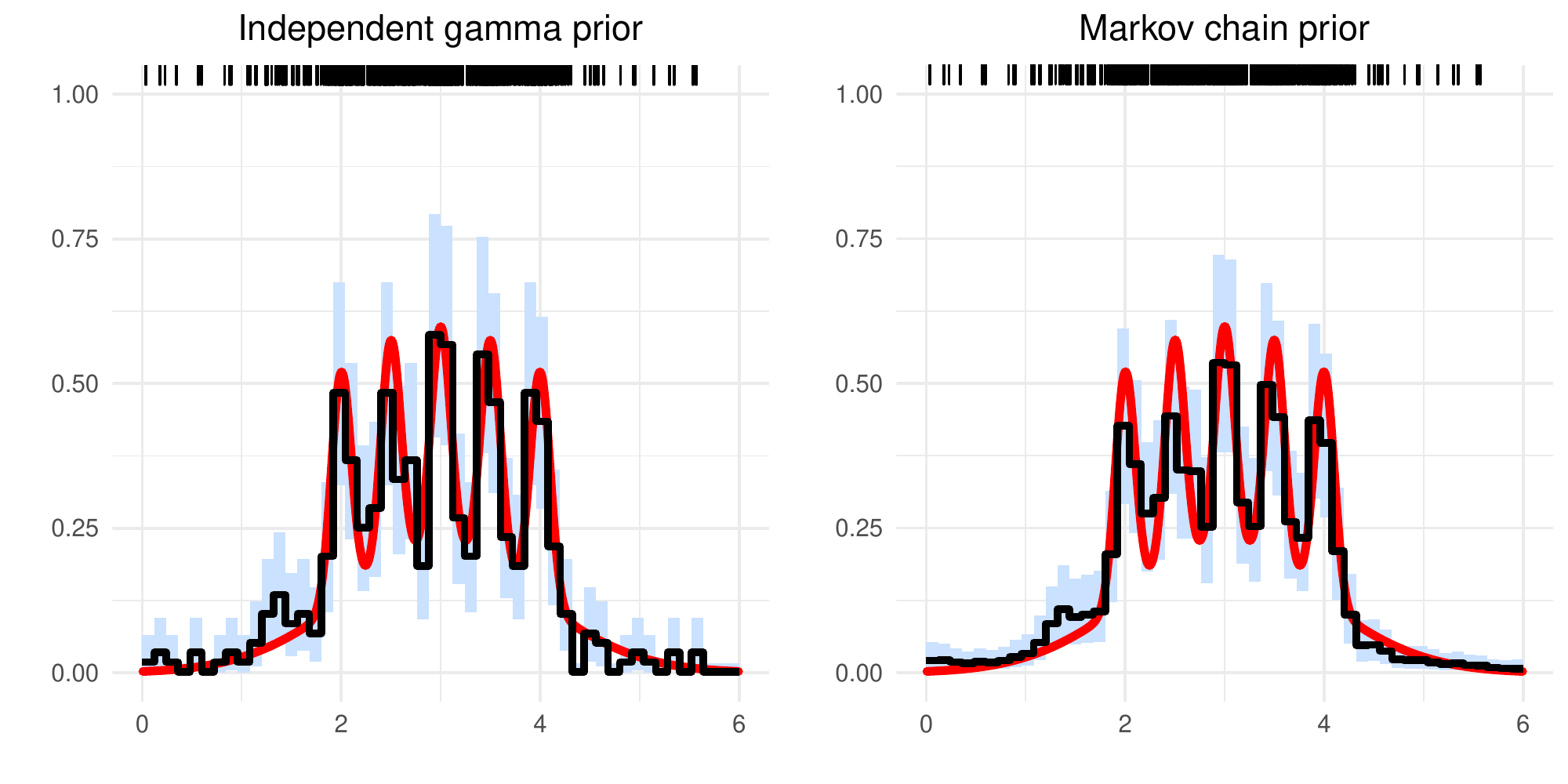}
	\end{center}
	\caption{Estimation results for the Bart Simpson function $\lambda_0$ from Subsection \ref{sec:simpson}. The hyperparameters are $N=50$, $\alpha=\beta=0.1$ for the independent gamma prior, and $N=50$ chosen via \eqref{eq:rule}, $\alpha_1=\beta_1=0.1$ and $\alpha\sim\operatorname{G}(1, 0.1)$ for the GMC prior.  Top row:  $n=200$, bottom row: $n=500$.}
	\label{fig:bart200}
\end{figure}

\begin{figure}
	\begin{center}
		\includegraphics[width=0.9\textwidth]{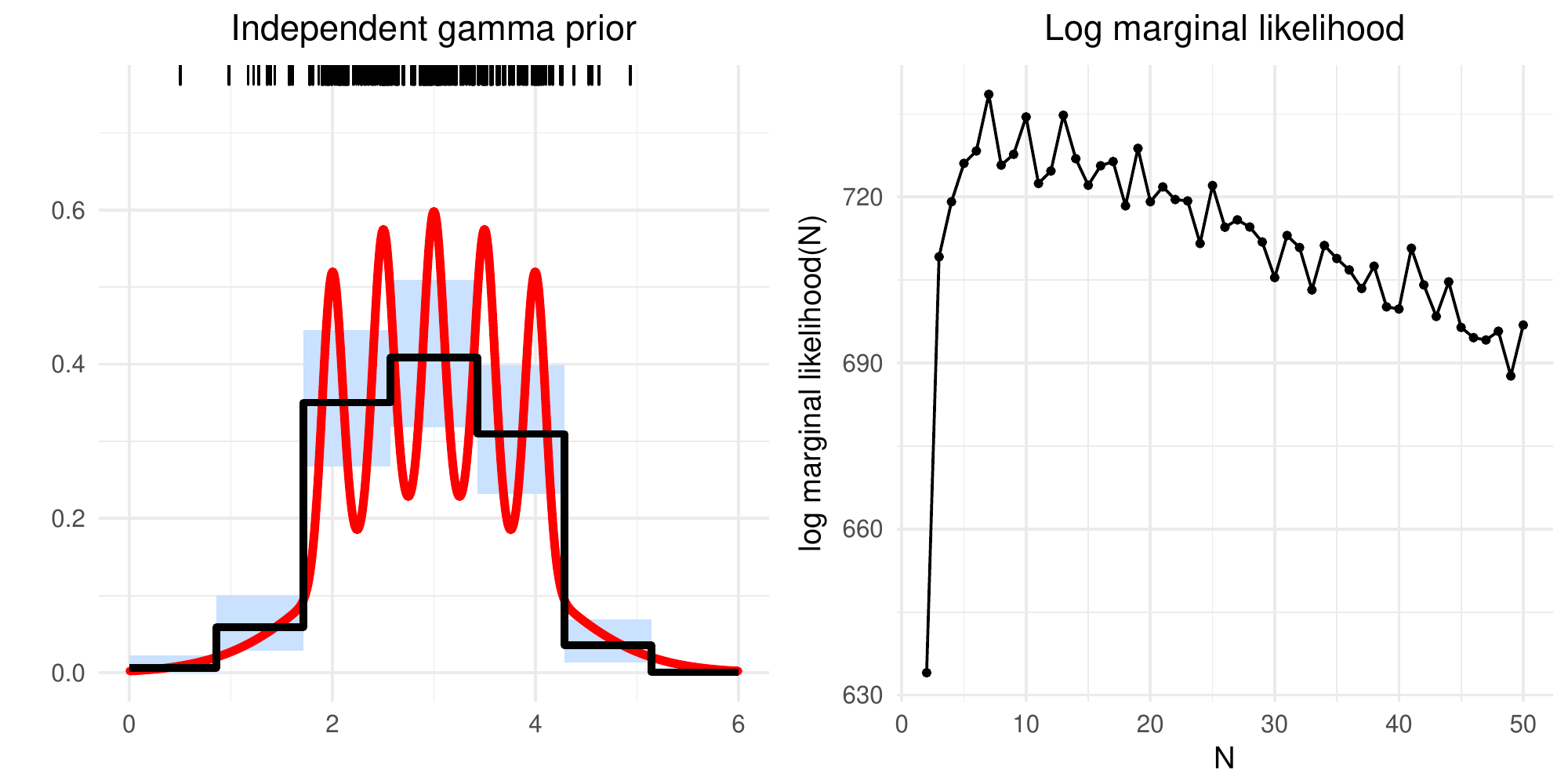}
	\end{center}
	\caption{Estimation results for the Bart Simpson function $\lambda_0$ from Subsection \ref{sec:simpson}, using the independent gamma prior. The sample size was $n=200$, the bin number $N=7$ was determined via the empirical Bayes method from Subsection \ref{sec:ebayes}, and the remaining hyperparameters were $\alpha=\beta=0.1$. The right panel displays the logarithm of the marginal likelihood as a function of $N$ (up to an additive constant independent of $N$).}
	\label{fig:bart:ebayes}
\end{figure}

\subsection{Practical recommendation}
\label{sec:recommendations}

The synthetic data examples considered in previous subsections allow us to formulate a practical recommendation. Specifically, we saw that both our methods are fast, scale with the  sample size and perform well in non-trivial inference problems. However, this good performance is conditional on an appropriate choice of the bin number $N$. In that respect, results we obtained for the independent gamma prior, when $N$ was chosen via the empirical Bayes method, were mixed. On the other hand, the method based on the GMC prior showed substantial stability with respect to a choice of $N$. In particular, a simple rule-of-thumb \eqref{eq:rule} led to good practical results. In light of these observations, we recommend the method based on the GMC prior as a default estimation strategy.

\section{Real data examples}\label{sec:realdata}
In this section, we apply our approach with the GMC prior on three real datasets. In light of Subsection \ref{sec:recommendations}, we do not consider the method based on the independent gamma prior.

Unless explicitly stated otherwise, we took $\alpha_1=\beta_1=0.1$, $N$ as in \eqref{eq:rule}  and $\alpha \sim \operatorname{Exp}(0.1)$. In each case, we ran the  Gibbs sampler for $30000$ iterations, and based the posterior inference on the second half of the generated posterior samples, with the first half dropped as a burn-in.

\subsection{UK coal mining disasters data}
\label{subsec:coal}

The dates of coal mining disasters (defined as accidents with $10$ or more fatal casualties, 191 events in total) in Britain between 15 March 1851 and 22 March 1962 provided in \cite{jarrett79} serve as a modern benchmark for point process inference. The dataset is accessible in {\bf R} as {\sl {coal}} in the {\bf{boot}} package; see \cite{canty17}.

The data have been analysed in \cite{green95} for change-point detection (it should be noted that  change-point estimation is a different inferential task from the one studied in this paper). 
A historical perspective on  change-point analysis for this problem is given in \cite{raftery86}, where it is suggested that an observed decrease in the accident rate over 1851--1962 is mainly due to an abrupt decrease around the years 1887--1895, possibly associated with changes in the coal industry around that time, such as decline in labour productivity (e.g., due to overtime) starting at the end of 1880s, and the emergence of the Miners' Federation in 1889. As a further possible reason, \cite{lloyd15} indicate passing of the Coal Mines Regulation Acts of 1872 and 1887 by the UK parliament with the aim of improving safety for mine workers. A more precise treatment would require adjustment for the coal mining industry size\footnote{Some relevant historical data are available at \url{https://www.gov.uk/government/statistical-data-sets/historical-coal-data-coal-production-availability-and-consumption} (Accessed on 22 June 2019), although the figures start to become annual only from 1913 on.}. This is not attempted neither in the above, nor in the following works, that give  
a machine learning perspective on inference for the coal mining disasters data: \cite{adams09}, \cite{hensman15}, \cite{komsamo15} and \cite{rao11}.

The rule-of-thumb \eqref{eq:rule} led to $N=48$ bins. The resulting posterior mean and the $75\%$ and $95\%$ marginal posterior credible bands are shown in Figure \ref{fig:coal}. The computing time for the Gibbs sampler was under a second. Our estimation results are, broadly speaking, similar to those already reported in the literature. Hence we only present a brief comparison with two state-of-the-art Bayesian methods from \cite{adams09} and \cite{lloyd15} (the corresponding posterior means were read off Figure 3 in \cite{lloyd15} via WebPlotDigitizer 4.1, see \cite{rohatgi17}). Our conclusions are as follows:

\begin{itemize}
	\item We believe that the method from \cite{adams09} does not pick up well enough the behaviour of the intensity function in the neighbourhood of the year 1851, which follows upon a comparison of the posterior mean to the rug plot of data points. Likewise, the posterior mean seems to drop off too sharply starting from mid-1950ies, at the right boundary point of the observation window, this being corroborated also by the method from \cite{lloyd15}. Both these features of the posterior mean from \cite{adams09} are possibly instances of edge effects (see, e.g., \cite{fan96} for a discussion of edge or boundary effects in non-parametric estimation). This boundary behaviour, but only at the left endpoint, is demonstrated by the method from \cite{lloyd15} as well.
	\item The method from \cite{adams09} appears to oversmooth in the years 1900--1925, which also follows by a comparison to the method from \cite{lloyd15}. As a consequence, the method from \cite{adams09} does not differentiate the World War I period in any particular way from the immediately preceding or successive years. This is unlike the World War II period, where it is in global  agreement with the other two methods.
	\item The posterior means from \cite{adams09} and \cite{lloyd15} do pass through the $95\%$ marginal credible band constructed via our GMC method.
	\item The $75\%$ marginal posterior bands for both our method and the method from \cite{lloyd15} look similar, except a disagreement on the first two bins, and a somewhat higher accident rate during the World War II period. As explained above, we believe the former to be due to an edge effect.
\end{itemize}

\begin{figure}
	\begin{center}
		\includegraphics[width=0.9\textwidth]{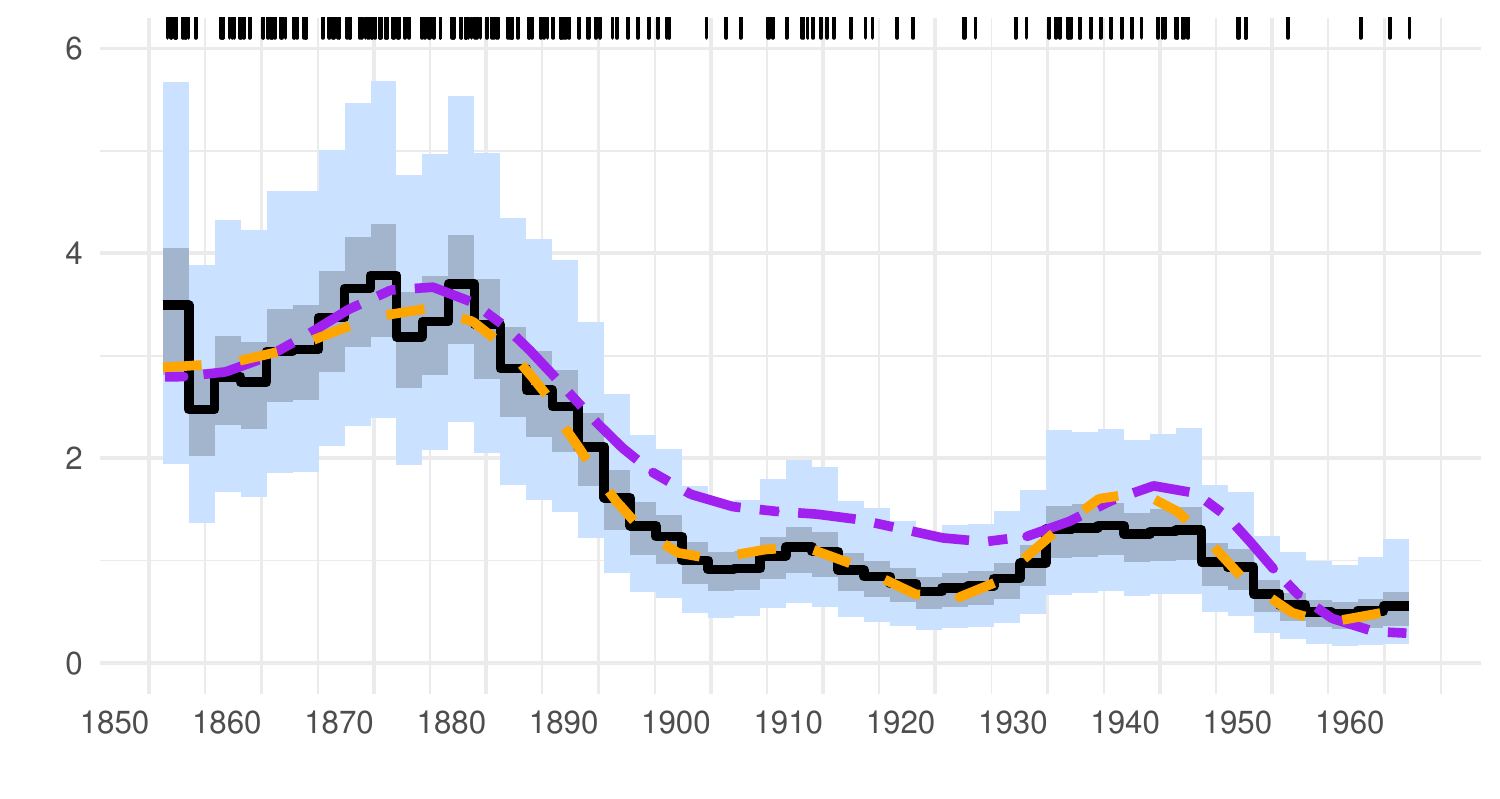}
		\includegraphics[width=0.9\textwidth]{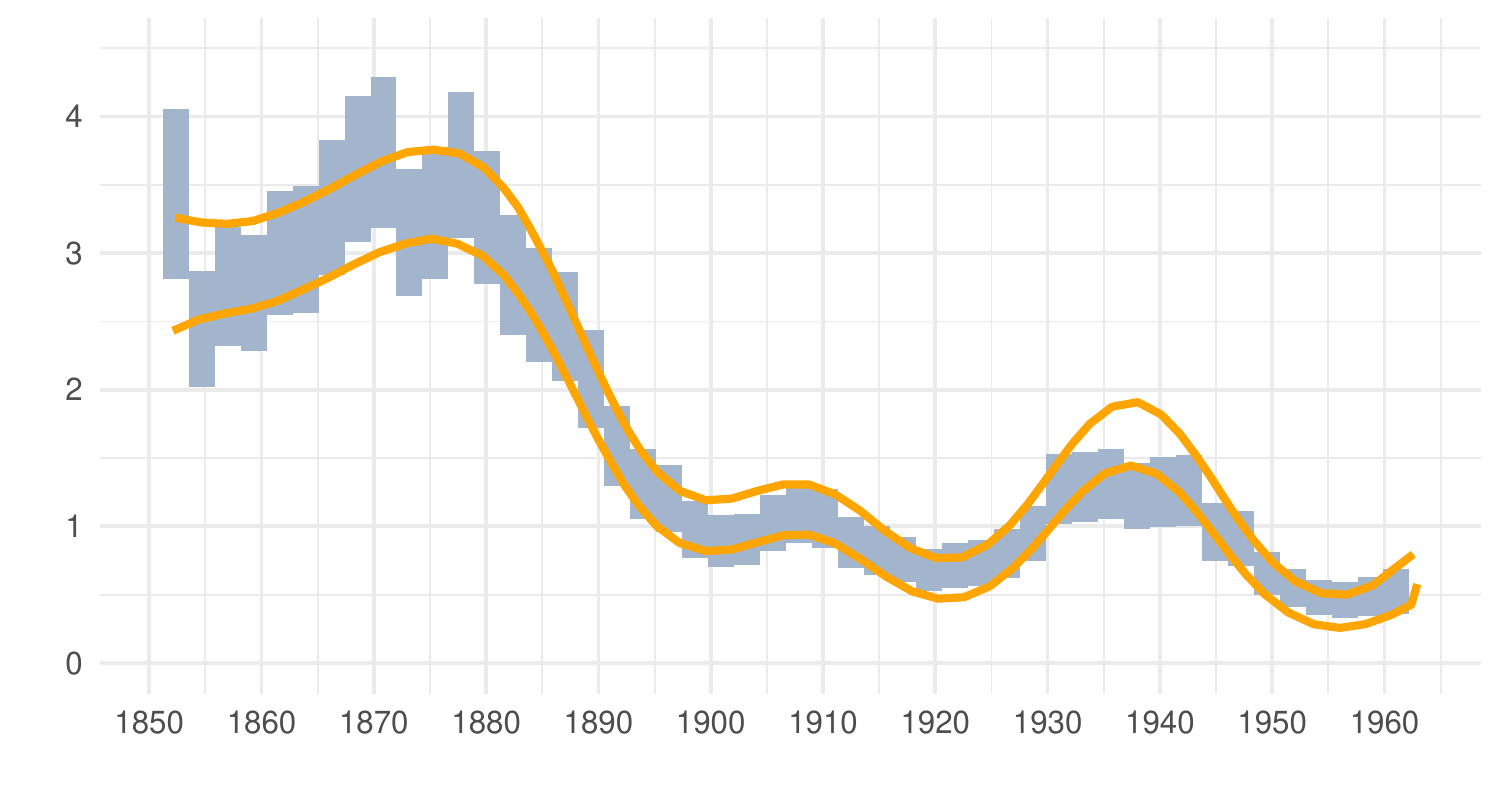}
	\end{center}
	\caption{Coal mining disasters data.  In the top plot displayed are the  posterior mean (black curve) with $75\%$ (grey) and $95\%$ (light blue) credible bands. The dashed purple curve is the estimate from \cite{adams09}, the dashed yellow line is an estimate from \cite{lloyd15}. In the bottom plot the $75\%$ marginal credible band is compared to that from \cite{lloyd15} (visualised via two yellow lines).}
	\label{fig:coal}
\end{figure}

\subsection{US mass shootings data}
\label{subsec:shootings}

In this subsection we will apply our GMC Bayesian method to analyse the US mass shootings data collected by the US nonprofit organisation Mother Jones, see \cite{follman18a} and \cite{follman18b}. The data cover the years 1982--2018 and provide extensive information on mass shootings in the US over that period. A mass shooting is defined as a single attack in public place in which four or more victims were killed (see \cite{follman18a}). This definition is essentially based on the FBI crime classification report from 2005. The definition excludes armed robbery, gang violence, or domestic violence, and focusses on shooting incidents where the motive appears to be indiscriminate mass murder (see \cite{follman12}). We note that for the data collected over 2013--2018, Mother Jones lowered the baseline of four fatal victims to three, reflecting a revised federal law.\footnote{Investigative Assistance for Violent Crimes Act of 2012. Pub.\ L.\ 112--265. 126 Stat. 2435--2436. 14 January 2013. \url{https://www.gpo.gov/fdsys/pkg/PLAW-112publ265/content-detail.html} (Accessed on 24 March 2018)} In our analysis, for uniformity in methods for data collection, we maintained the older definition. However, we acknowledge the fact that a precise definition of a mass shooting is somewhat elusive. This will affect inclusion or exclusion of specific shooting cases in any quantitative analysis.

The dataset used in our analysis consists of 85 shooting incidents over the period of 20 August 1982 -- 14 February 2018, that lead to 4 or more lethal casualties each. We model occurrences of mass shootings as realisations from a Poisson point process. This approach is in line with modelling the coal mining disasters data via a Poisson point process, see Subsection~\ref{subsec:coal}, as well as using the Poisson distribution in \cite{spiegelhalter19}, pages 248--251, to model the UK homicide rates. We set  1 January 1982 and 14 March 2018 as the start and end times of the observation period. 

We note that the US mass shootings data have already been analysed in \cite{cohen14} using a different approach (Shewhart's chart), and data up to 2014. A conclusion reached in \cite{cohen14} is that the rate of mass shooting incidents increased since September 2011. Such a conclusion is further corroborated by the FBI study of active shooter incidents in the 2000--2013 period, see \cite{blair14}.

We present our estimation results in Figure~\ref{fig:shootings} both for  $N=9$ bins and $N=21$ bins. The first setting corresponds to considering bins of $4$ years, whereas the second relies on the rule-of-thumb \eqref{eq:rule}.  The Gibbs sampler runs were completed under a second. Both of the presented  plots suggest, roughly speaking, an increasing intensity function, which can be taken as a substantial evidence of an increasing trend in occurrences of mass shootings. In fact, posterior means indicate a four-fold increase in the Poisson process intensity over 1982--2018. Part of this increase can possibly be correlated with an increase in the US population between 1982 to 2018 (from 232 to 327 million), but nevertheless the relative increase in population is considerably smaller than that in the intensity of mass shooting incidents. Contrast this to Figure \ref{fig:homicides}, where we display the number of homicides by firearms in the US over the years 1999--2017 (data for years prior to 1999 was not accessible to us)\footnote{Centers for Disease Control and Prevention, National Center for Health Statistics. Underlying Cause of Death 1999--2017 on CDC WONDER Online Database, released December, 2018. Data are from the Multiple Cause of Death Files, 1999--2017, as compiled from data provided by the 57 vital statistics jurisdictions through the Vital Statistics Cooperative Program. Accessed at \url{http://wonder.cdc.gov/ucd-icd10.html} on June 22, 2019.}; this has stayed relatively stable. Especially after around 2005 the shooting incidents appear to occur at an ever accelerating rate. A fuller treatment of the question would require inclusion of the population size as a covariate in the model, perhaps in combination with other covariates too. See, e.g., the right panel of Figure \ref{fig:homicides}, where we display the yearly average numbers of homicides per 100,000 inhabitants in the US\footnote{We downloaded the population data from: World Bank, Population, Total for United States [POPTOTUSA647NWDB], retrieved from FRED, Federal Reserve Bank of St. Louis; \url{https://fred.stlouisfed.org/series/POPTOTUSA647NWDB, September 14, 2019.}}. This, however, lies outside the scope of the present work, that deals exclusively with the Poisson point process inference.

\begin{figure}
	\begin{center}
		\includegraphics[width=0.9\textwidth]{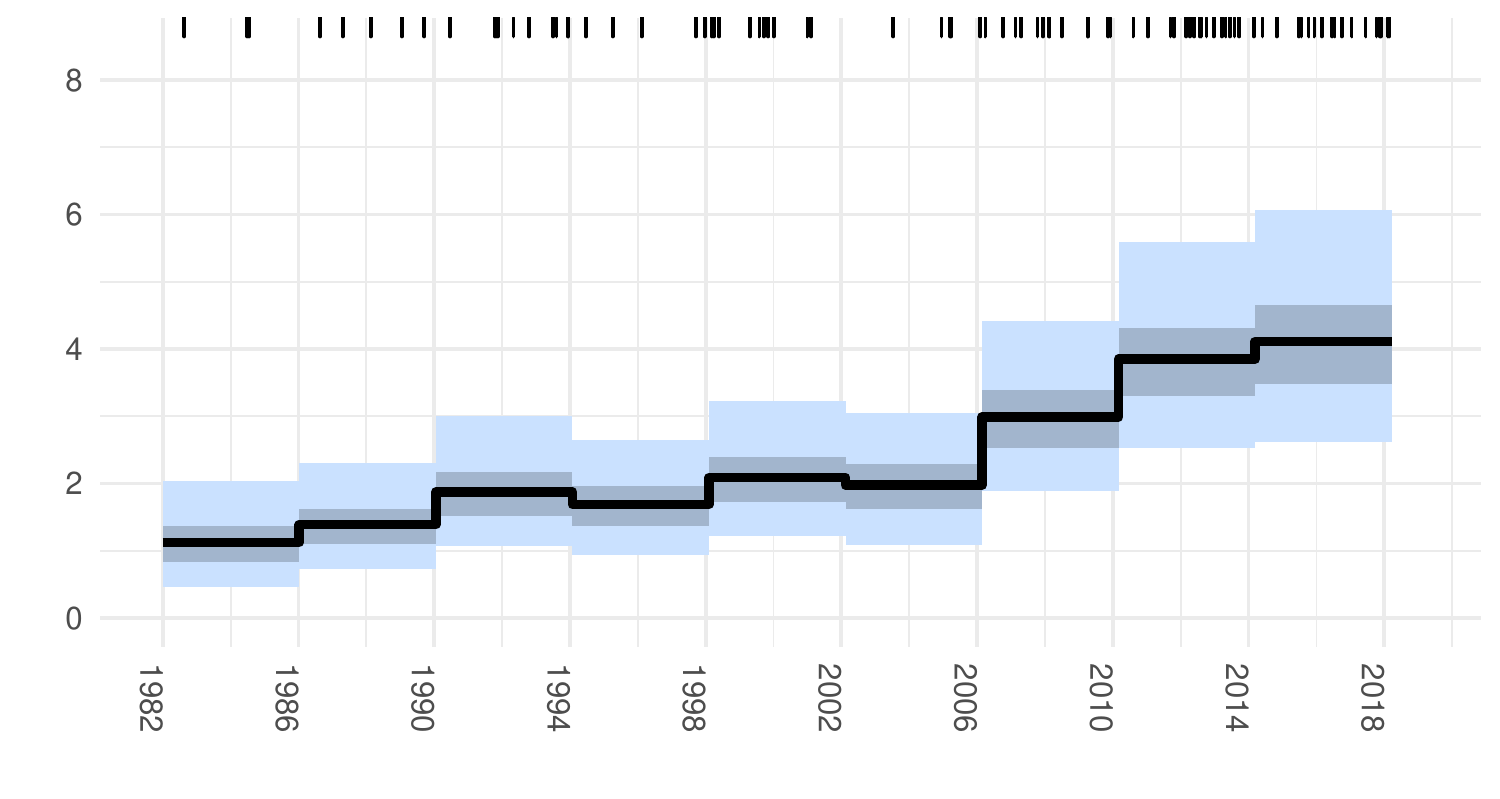}\\
		\includegraphics[width=0.9\textwidth]{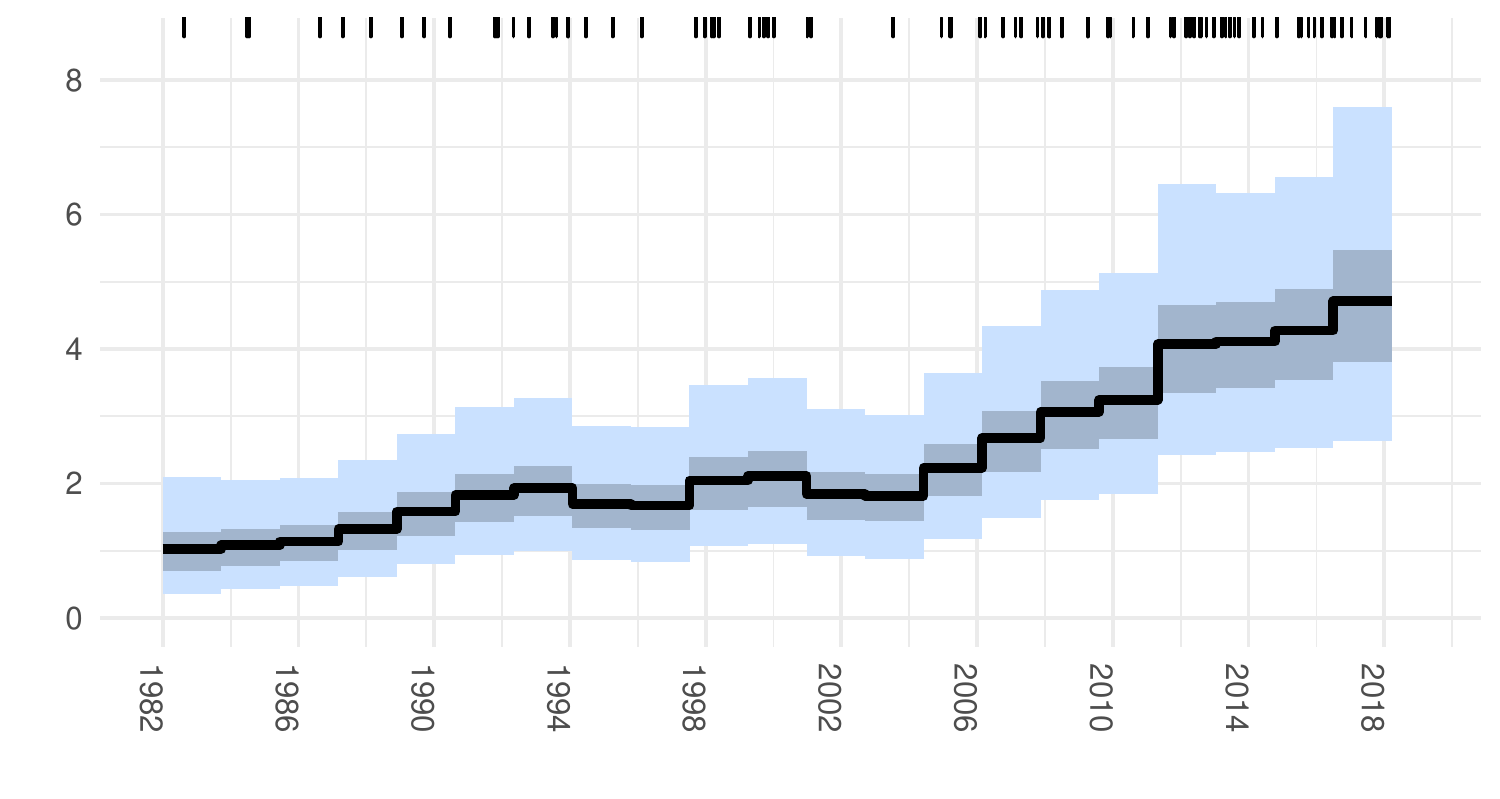}
	\end{center}
	\caption{US mass shootings data. The rug plot on the top displays the event times. Displayed are the  posterior mean (black curve) with $75\%$ (grey) and $95\%$ (light blue) credible bands. Top: $N=9$ bins (one bin for each $4$ years). Bottom: $N=21$ bins (according to the rule for determining $N$ given in equation \eqref{eq:rule}). }
	\label{fig:shootings}
\end{figure}

\begin{figure}
	\begin{center}
		\includegraphics[width=0.9\textwidth]{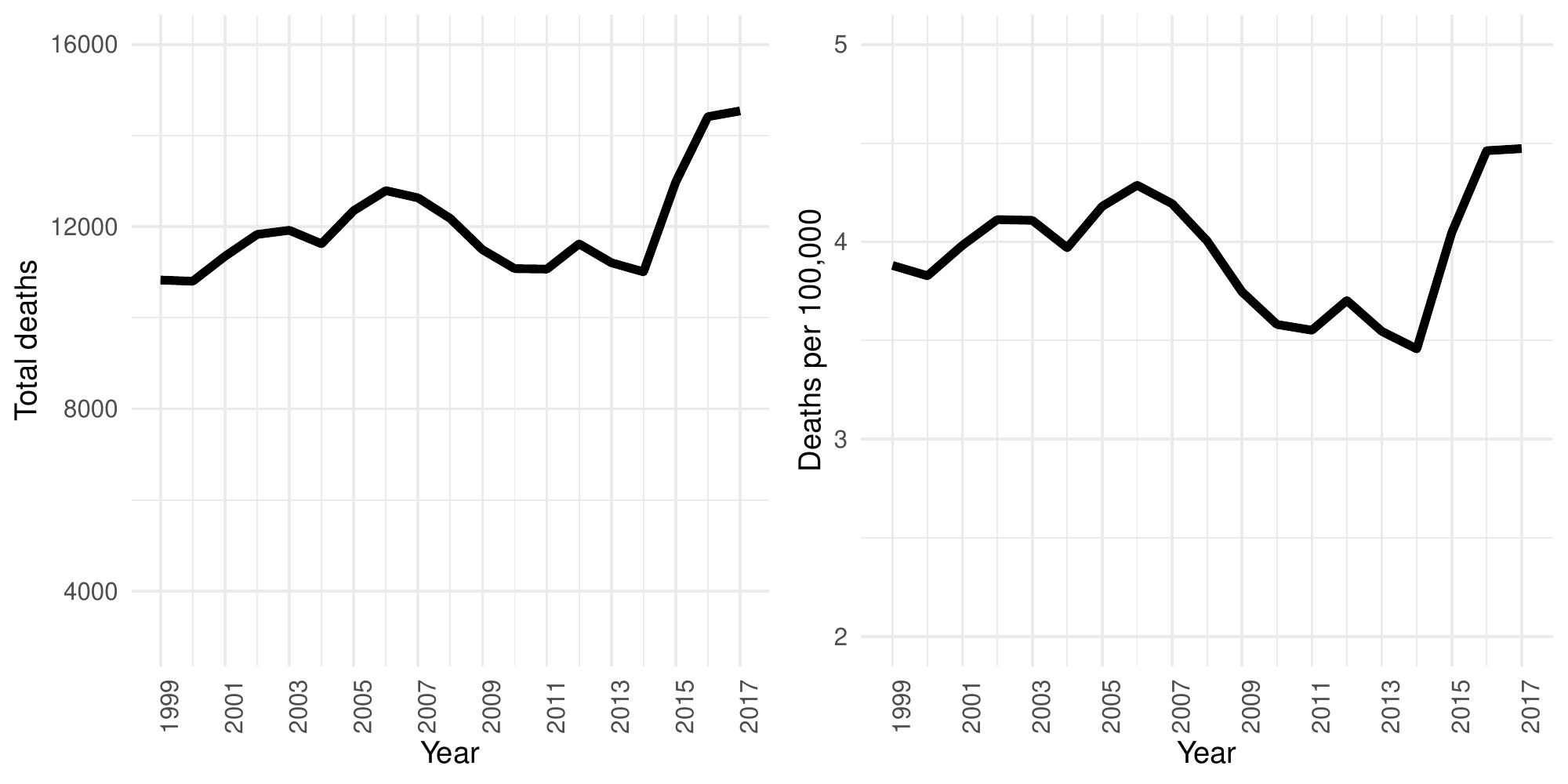}
	\end{center}
	\caption{Left panel: yearly number of homicides by firearms in the US over the years 1999--2017. Right panel: average number of homicides per 100,000 inhabitants.}
	\label{fig:homicides}
\end{figure}

\subsection{Trump's Twitter data}
\label{subsec:twitter}

Social media constitutes a natural field for application of point processes (cf.~\cite{lloyd15} and \cite{komsamo15}). In this subsection, we will analyse the tweet data from Donald J.~Trump's personal Twitter account {\tt{@realDonaldTrump}}. President Trump is an active Twitter user\footnote{A complete archive of Donald Trump's tweets, updated on an hourly basis, is available at \url{https://github.com/bpb27/trump_tweet_data_archive}}. His tweet data has already been a subject of quantitative studies in the past. Thus, using text mining tools, David Robinson (see \cite{robinson16} and \cite{robinson17}) analysed the question whether tweets posted on {\tt{@realDonaldTrump}} from Android and iPhone devices belong to different persons: Android tweets were posted or dictated by Donald Trump personally,\footnote{It is known that Donald Trump's personal mobile device was an Android, possibly a Samsung Galaxy S3. See, e.g., \url{https://www.androidcentral.com/which-android-phone-does-donald-trump-use} (Accessed on 26 March 2018)} whereas the iPhone tweets were written by his staff members. This question can be approached from the point process angle as well. Specifically, for each device type (Android and iPhone) we will model tweet arrival times as realisations from a Poisson point process. We will then infer the intensity functions and compare them. Significant differences in shapes of the intensity functions can be taken as an indicator of differing tweeting habits of the Android and iPhone users.

We model tweet times as realisations from a periodic Poisson process with a period of one day. The data we used in our analysis were all the tweets made from the Android and iPhone devices in the period between 16 June 2015 (the official launch of Donald Trump's presidential campaign) and 9 November 2016 (the date he won the presidential elections). For interpretability, we  took $N=48$ bins, which corresponds to bins of half an hour (note that the rule \eqref{eq:rule} would lead to $N=50$, a minor difference). The number of analysed tweets using  Android and  iPhone equals  $563$ and $952$ respectively.  Completion of each of the Gibbs samplers  took less than  a second.  Our estimation results are given in Figure~\ref{fig:twitter1}. The most prominent differences between the two intensity functions we inferred are that the iPhone user shows a high tweeting activity during the night hours and in the second half of the day. On the other hand, the Android user's main activity falls in the morning and during early afternoon hours.

\begin{figure}
	\begin{center}
		\includegraphics[width=0.9\textwidth]{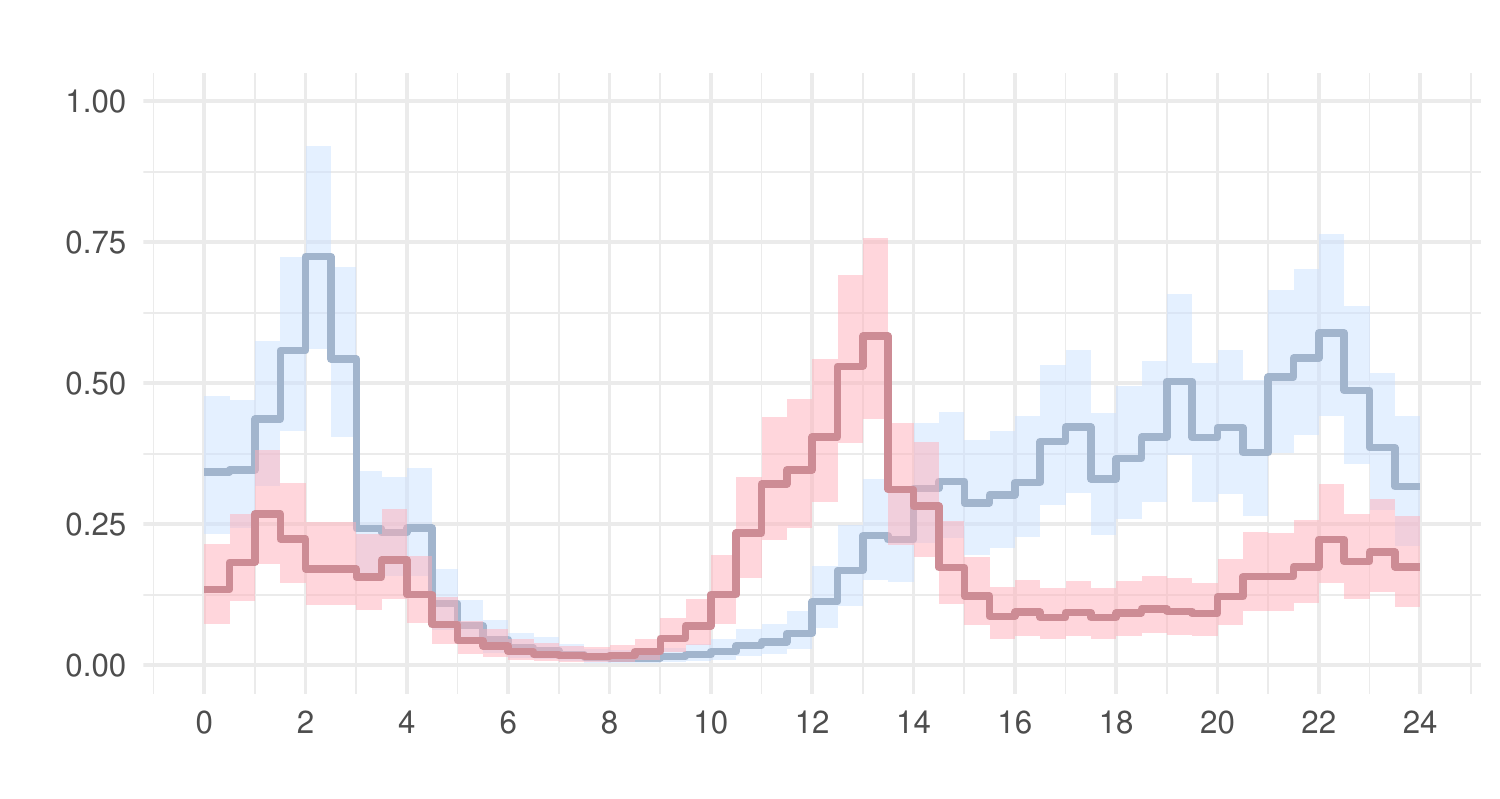}\end{center}
	\caption{Estimation results for the Twitter data from Subsection \ref{subsec:twitter}. Displayed are the posterior means with  $95\%$ credible bands. Blue: iPhone tweet data, pink: Android tweet data. The observation period in both cases is 16 June 2015 -- 9 November 2016.}
	\label{fig:twitter1}
\end{figure}

In Figure \ref{fig:twitter2} we compare the intensity function inferred from iPhone tweets made after 25 March 2017 (the date of the last Android tweet. As confirmed a few days later by  Dan Scavino Jr., the White House Director of Social Media and Assistant to the President, Donald Trump stopped using an Android device and switched to an iPhone\footnote{See, e.g., \url{https://www.businessinsider.nl/donald-trump-switches-apple-iphone-from-unsecure-android-2017-3/} (Accessed on 26 March 2018)}) until 1 January 2018 to the intensity function inferred from the combined Android and iPhone data over the period 16 June 2015 -- 9 November 2016, as well as only the Android data. A conclusion that emerges from the graphs is that although the inferred intensity functions do not match each other exactly, general tweeting habits of the Android user during 16 June 2015 -- 9 November 2016 and the iPhone user after 25 March 2017 are qualitatively similar. Specifically, tweeting activity over the night hours is largely the same for both these datasets, which also show a surge in tweeting around 11:00--13:00, and a subsequent decrease until a few hours prior to the midnight. On the other hand, observed differences between the two intensity functions can be plausibly explained by the fact that past the 25 March 2017 date, the iPhone is possibly still used by President Trump's aides, primarily in the morning and during early afternoon hours.

\begin{figure}
	\begin{center}
		\includegraphics[width=0.9\textwidth]{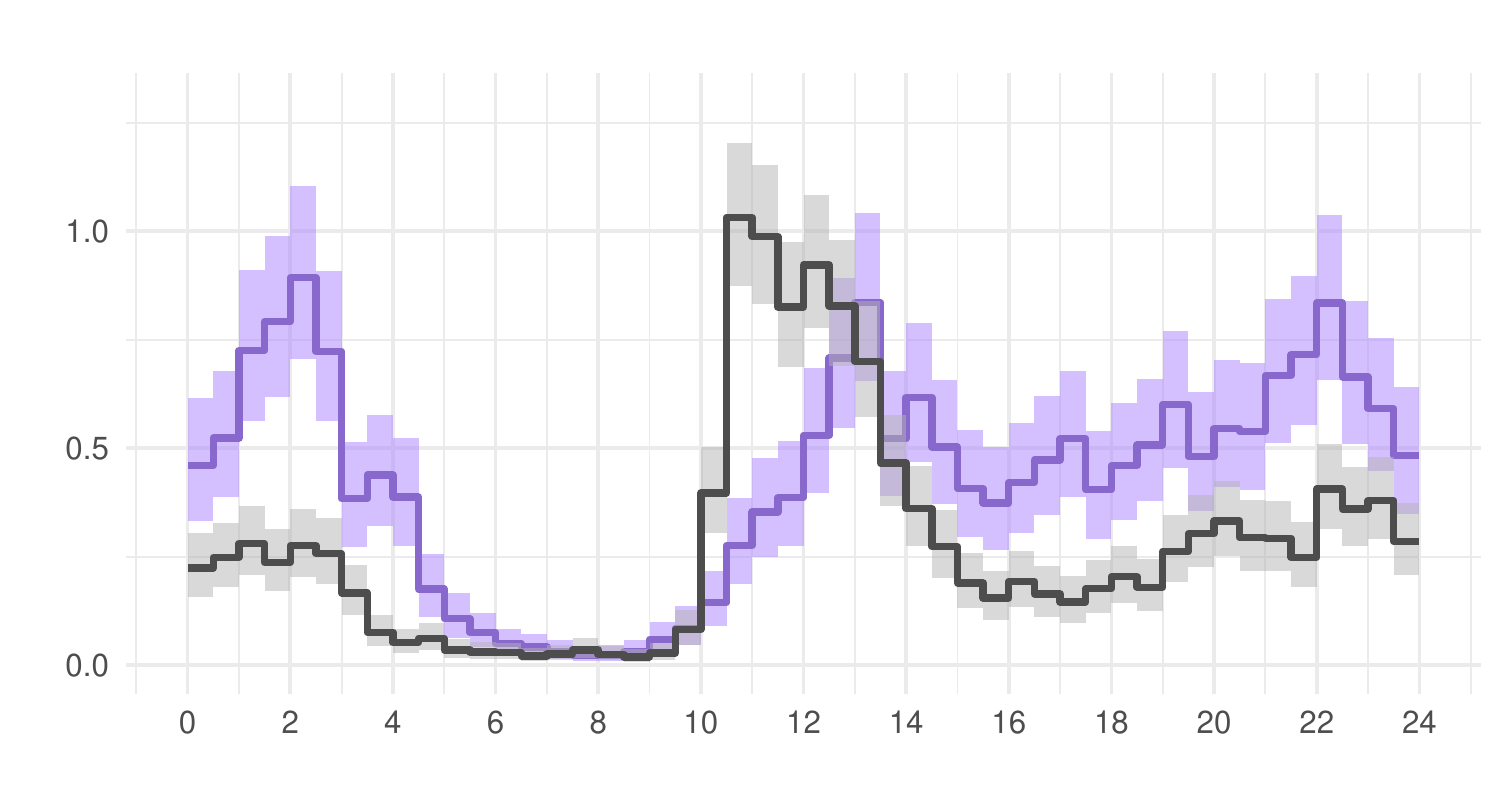}
		\includegraphics[width=0.9\textwidth]{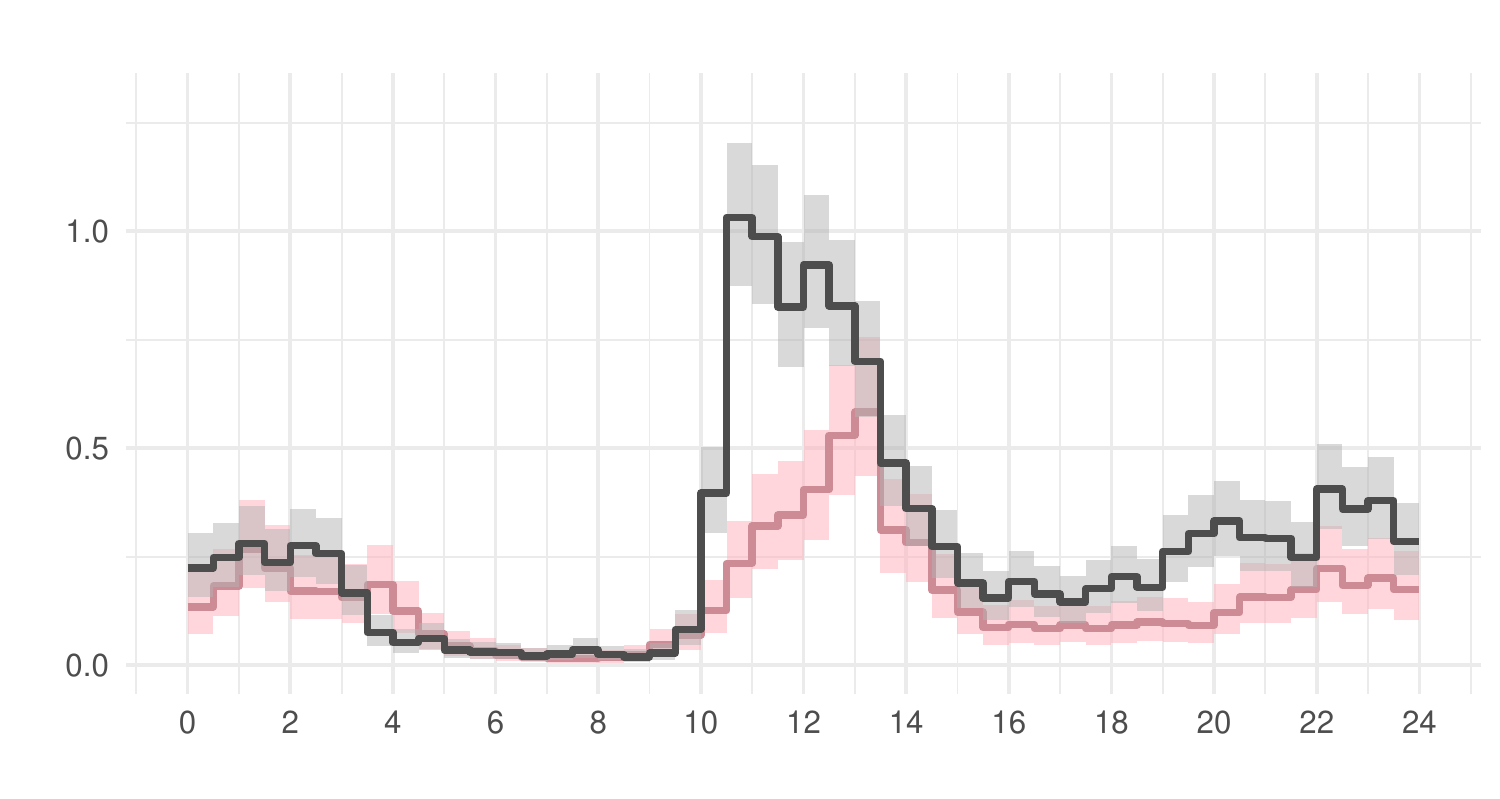}
	\end{center}
	\caption{Estimation results for the Twitter data from Subsection \ref{subsec:twitter}. Top plot: displayed are in black the  posterior mean with a $95\%$ credible band for all iPhone tweets over the period  3 March 2017 (date of last Android tweet) -- 1 January 2018, and in purple the posterior mean with a $95\%$ credible band for the combined iPhone and Android tweets over the period 16 June 2015 -- 9 November 2016. Bottom plot: displayed are in black the  posterior mean with a $95\%$ credible band for all iPhone tweets  over the period 3 March 2017 -- 1 January 2018, and in pink the posterior mean with a $95\%$ credible band for the iPhone tweets over the period 16 June 2015 -- 9 November 2016.}
	\label{fig:twitter2}
\end{figure}

\section{Discussion}
\label{sec:conclusions}

In this paper we studied two related non-parametric Bayesian approaches to estimation of the intensity function of a Poisson point process. Our methods are easy to implement, in that the posterior with the first method is available in closed form, whereas with the second method the posterior inference can be performed using a straightforward implementation of the Gibbs sampler. We believe that our methods are conceptually simpler than those previously developed in the statistical literature.
There also exists a solid body of machine learning research dedicated to inference in Poisson point processes. The methods used in this strand of the literature are largely based on optimisation techniques and the variational Bayes approach (see \cite{lloyd15}), or combinations of these techniques with MCMC (see \cite{hensman15}). However, they may underreport uncertainties in parameter estimates, which is an issue with the variational inference (see \cite{blei17}).
In that sense, the MCMC-based machine learning approaches to inference in point process models, such as \cite{adams09} and \cite{rao11}, hold an advantage over those employing optimisation techniques. On the downside, they typically do not scale well with large amounts of data. In contrast, implementation of our methods in {\bf{Julia}} is fast, with hundreds of thousands of data points and a large number of bins $N$ posing no significant computational challenges. 
Finally, we demonstrated good performance of our approach on simulated data examples and three real datasets: coal mining disasters data, the US mass shootings data and Donald Trump's Twitter data.

Two promising future research directions that we envision are an extension of our methods to a multi-dimensional setting and incorporation of covariates in the model. On the theoretical side, working out frequentist asymptotics for the method based on the GMC prior is interesting.

\section*{Acknowledgments}
The research leading to the results in this paper has received funding from the European Research Council under ERC Grant Agreement 320637. We would like to thank Adeline Leclercq Samson (Universit\'{e} Grenoble Alpes, Grenoble, France), Ryan Martin (North Carolina State University, Raleigh, North Carolina, USA) and Paulo Serra (Eindhoven University of Technology, Eindhoven, The Netherlands) for their constructive comments, that led to various improvements in the draft version of the paper. We have benefited from discussions on the GMC prior with Ali Taylan Cemgil (Bo\v{g}azi\c{c}i University, Istanbul, Turkey).

\appendix
\section{Proofs  of the technical results}\label{appendix}
\begin{proof}[Proof  of Lemma \ref{lem:indepposterior}]
	With $\lambda(x) = \sum_{k=1}^N \psi_k \ind_{B_k}(x)$, we get for the likelihood (as displayed in  \eqref{likelih})
	\begin{align*}
	L(X^{(n)};\lambda)&= \left(\prod_{j=1}^n\prod_{i=1}^{m_j}  \lambda(X_{ij})\right) \times \exp\left(nT - n\int_0^T \lambda(x) \dd x\right) \\ & \propto \left(\prod_{j=1}^n\prod_{i=1}^{m_j} \left[ \sum_{k=1}^N \psi_k \ind_{B_k}(X_{ij})\right]\right) \times \prod_{k=1}^N e^{-n\Delta_k \psi_k}.
	\end{align*}
	As the  factor with the double product can be rewritten as $\prod_{k=1}^N \psi_k^{H_k}$, we get
	\begin{equation}
	\label{eq:likelih}
	L(X^{(n)};\lambda)\propto \prod_{k=1}^N \psi_k^{H_k}e^{-n\Delta_k \psi_k}.
	\end{equation}
	The result now follows in a straightforward way by conjugacy of the prior, see Assumption~\ref{assump:prio}. 
\end{proof}

\begin{proof}[Proof of Theorem \ref{thm:mean}]
Recall from Subsection \ref{subsection:prior} that $B_k=[b_{k-1}, b_k)$. 
	We will use the fact that
	\[
	H_k \sim \operatorname{Poisson}\left(n \int_{B_k}\lambda_0(x)\dd x \right), \quad k=1,\ldots,N,
	\]
	which follows by independence of $X_1,\ldots, X_n$. 
	Hence, 
	\begin{equation}
	\label{meanvar}
	\ee[H_k]=\var[H_k]=n \int_{B_k}\lambda_0(x)\dd x, \quad k=1,\ldots,N.
	\end{equation}
	For notational simplicity we assume that $T=1$. This entails no loss of generality.
	Recall from \eqref{psik} that the posterior mean is piecewise constant with value equal to $\hat{\psi}_k = (n\Delta+\beta)^{-1}(H_k + \alpha)$ on the $k$-th bin $B_k$.

	Introduce
	\[
	\bar{\lambda}_0(x)=\sum_{k=1}^N \lambda_0(b_{k-1})\mathbf{1}_{B_k}(x).
	\]
	First we note
	\begin{equation}
	\label{biasvar}
	\ee[\|\hat{\lambda}-\lambda_0\|_2^2] \leq 2\left( \|\lambda_0-\bar{\lambda}_0\|_2^2 + \ee [\| \hat{\lambda} - \bar{\lambda}_0  \|_2^2]\right).
	\end{equation}
	Using H\"older-continuity of $\lambda_0$ (see Assumption \ref{lambda0}), we get for the first term on the right-hand side that
	\[
	\|\lambda_0-\bar{\lambda}_0\|_2^2 = \sum_{k=1}^N \int_{B_k} (\lambda_0(x)-\lambda(b_{k-1}))^2\dd x \leq L^2 \Delta^{2h}.
	\]
	For the second term on the right-hand side of \eqref{biasvar} we have
	\begin{align*}
	\ee [\| \hat{\lambda} - \bar{\lambda}_0  \|_2^2] &= \Delta \sum_{k=1}^N \ee [(\hat{\psi}_k-\lambda_0(b_{k-1}))^2] \\ &= \Delta\sum_{k=1}^N \{ \ee[\hat{\psi}_k] - \lambda_0(b_{k-1}) \}^2 + \Delta\sum_{k=1}^N \var[\hat{\psi}_k].
	\end{align*}
	By H\"older-continuity of $\lambda_0$, for $x\in B_k$, $\lambda_0(x)= \lambda_0(b_{k-1}) + O\left(\Delta^{h}\right)$. Combining this with  \eqref{meanvar}, we obtain
	\begin{equation}
	\label{mean:hk}
	\ee[H_k]=n\Delta \lambda_0(b_{k-1})+n O(\Delta^{1+h}), \quad k=1,\ldots,N,
	\end{equation}
	so that
	\[
	\ee[\hat{\psi}_k]-\lambda_0(b_{k-1})=O\left( \frac{1}{n\Delta} + \Delta^h \right)=O\left(\Delta^h \right), \quad k=1,\ldots,N.
	\]
	Here the last equality follows from the choice of $\Delta$.
	Furthermore, again by \eqref{meanvar} and Assumption \ref{lambda0},
	\[
	\var[\hat{\psi}_k] = \frac{1}{(n\Delta+\beta)^2}\var[H_k] \lesssim \frac{n\Delta}{(n\Delta+\beta)^2} \lesssim \frac{1}{n\Delta}.
	\]
	Thus,
	\[
	\ee [\| \hat{\lambda} - \bar{\lambda}_0  \|_2^2] \lesssim \Delta^{2h} + \frac{1}{n\Delta}.
	\]
	The statement of the theorem follows from the fact that $\Delta \asymp n^{-1/(2h+1)}$.
\end{proof}

\begin{proof}[Proof of Theorem \ref{thm:post}]
	By Chebyshev's inequality,
	\begin{equation}\label{eq:bv0}
	\ee [\Pi_N(\|\lambda-\lambda_0\|_2 \ge M_n \varepsilon_n \mid X^{(n)})] \le \frac{1} {M_n^2 \varepsilon_n^2}\ee \left[\ee_{\Pi_N} \left(\|\lambda-\lambda_0\|_2^2 \mid X^{(n)}\right)\right].
	\end{equation}
	Then, by the posterior bias-variance decomposition,
	\begin{equation}\label{eq:bv1}
	\ee \left[\ee_{\Pi_N} \left(\|\lambda-\lambda_0\|_2^2 \mid X^{(n)}\right)\right] =\ee \left [ \|\hat{\lambda}-\lambda_0\|_2^2\right] +\Delta\sum_{k=1}^N \ee [\var_{\Pi_N} (\psi_k \mid X^{(n)} )].
	\end{equation}
	We have for the second term on the right-hand side, using \eqref{post:gamma}, a formula for the variance of a gamma distribution and \eqref{mean:hk} that
	\begin{align*}
	&\Delta\sum_{k=1}^N \ee \left[\var_{\Pi_N} \left(\psi_k \mid X^{(n)} \right)\right]  = \Delta \sum_{k=1}^N \ee\left[ \frac{H_k+\alpha}{(n\Delta+\beta)^2} \right]\\
	&\qquad =\frac{\Delta}{(n\Delta+\beta)^2} \sum_{k=1}^N \ee[H_k]+O\left( \frac{1}{(n\Delta)^2} \right)
	= O\left( \frac{1}{n\Delta} \right) = O(n^{-2h/(2h+1)}).
	\end{align*}
	As far as the first term on the right-hand side of \eqref{eq:bv1} is concerned, by Theorem \ref{thm:mean} it is also of the order $n^{-2h/(2h+1)}$. Hence so is \eqref{eq:bv1}. Using this fact, the statement of the theorem now follows from equation \eqref{eq:bv0}.
\end{proof}

\begin{figure}
	\begin{center}
		\includegraphics[width=0.9\textwidth]{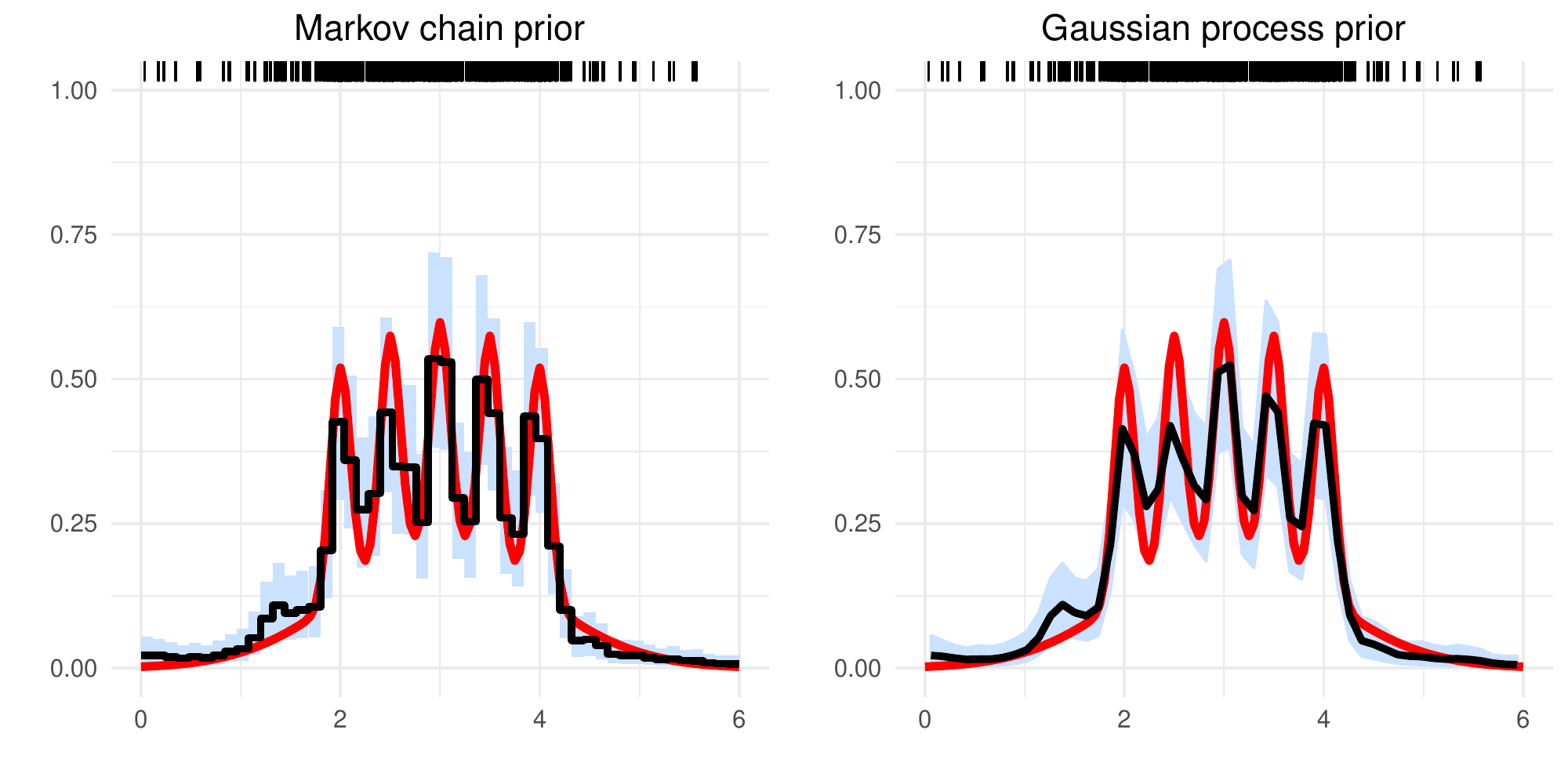}
	\end{center}
	\caption{Estimation results for the Bart Simpson function $\lambda_0$ from Subsection \ref{sec:simpson}, with the same $n=500$ replicated observations as used in that section. The settings for the GMC prior are described in the caption of Figure \ref{fig:bart200}. The Gaussian process prior is taken to have zero mean and covariance given by a Mat\'ern $3/2$ kernel, with its hyperparameters (on the log scale) assigned the independent $N(-2,4)$ prior distributions. $N=50$ bins are used to bin the count data and the GP posterior summaries are assigned to the midpoints of these bins.}
	\label{fig:bart:gauss}
\end{figure}

\section{Gaussian processes}\label{sec:gp}

Although the emphasis of the present paper is not on a comparison with other methods for point process inference per se, in this appendix we briefly report on the results we obtained with intensity estimation based on a Gaussian process prior (see the literature overview in Subsection \ref{subsec:literature}). As a test model, we chose the Bart Simpson function and used the same simulated dataset as in Subsection \ref{sec:simpson}, corresponding to $n=500$ observations. For numerical evaluation, we relied on the {\bf GaussianProcesses.jl} package of \cite{gaussian18}. As detailed in Section 4.3 of that technical report, the approach requires binning of the point process data (similar to our methods), and the model for the intensity of the process assumes the form
$
\Lambda_k = \exp(f_k),
$
for
$
k=1,\ldots,N,
$
where $\Lambda_k$ is the (integrated) intensity of the process on the $k$th bin, whereas $f_k$'s are assigned the Gaussian process (GP) prior. Posterior simulation relies on the Hamiltonian Monte Carlo (HMC) algorithm. From the assumed model it follows that in our notation, $\psi_k = e^{f_k}/(n\Delta)$ for $k=1,\ldots,N.$

To make results comparable with ours in Subsection \ref{sec:simpson}, we used $N=50$ bins. The GP had mean zero and a Mat\'ern $3/2$ kernel, with its hyperparameters (on the log scale) assigned the independent $N(-2,4)$ prior distributions. These choices correspond to the ones made in Section 4.3 in \cite{gaussian18}, and are not claimed to be optimal. After a trial-and-error procedure, we set the integration step size of HMC to $0.025$. This yielded a ca.\ $50\%$ acceptance rate of HMC. We kept other tuning quantities at their default values in {\bf GaussianProcesses.jl}. We ran the Markov chain for $30000$ iterations and dropped the first third of the posterior samples as burn-in; no thinning was used (cf.\ our settings in Section \ref{sec:simulations}). The posterior summaries were assigned to the midpoints of the bins.

The results are displayed in Figure \ref{fig:bart:gauss}. Overall, the performance of the GMC prior-based method is quite similar to that of the GP-based one  in this example and for this dataset. Our method appears to retain the heights of the sharp peaks of the Bart Simpson function somewhat better than the GP-based one, though admittedly the difference is marginal to be relevant in practice. More importantly, our method turned out to be much faster, with its computation completed in a fraction of a second, whereas the timing was in minutes for the GP-based approach.

\section{Imposing a prior on the number of bins}\label{sec:revjump}

Rather than fixing the number of bins $N$, it can be treated as an unknown parameter. From a Bayesian point of view this is a natural approach and the computational problem is commonly tackled using the reversible jump algorithm. Investigation of the performance of this approach, combined with the independent gamma prior, was suggested to us by  Ryan Martin. 

As different values of $N$ correspond to different binning of the data we write $H_k^{(N)}$  and $\Delta_k^{(N)}$ to highlight dependence on $N$. 
The marginal likelihood of the model indexed by $N$ is given by (see Equation \eqref{eq:marginal_lik}) \[
\operatorname{ML}_N(X^{(n)})=e^{Tn} \frac{\beta^{\alpha N}}{\Gamma(\alpha)^N} \prod_{k=1}^N \frac{\Gamma(\alpha+H_k^{(N)})}{(n\Delta_k^{(N)}+\beta)^{\alpha+H^{(N)}_k}},
\]
 Let $\pi_N$ be the prior probability of the model indexed by $N$. Define 
\[ R(N) = \log \operatorname{ML}_N(X^{(n)}) + \log \pi_N. \]
Suppose $q(N^\circ \mid N)$ is a Markov kernel on the model indices. Let the current iterate be $N$ with  value $R(N)$. The reversible jump algorithm iterates over the following steps:
\begin{enumerate}
  \item Select a new model $N^\circ$ from $q(\cdot \mid N)$. 
  \item Compute 
\[ A = R(N^\circ) - R(N) + \log q(N\mid N^\circ) - \log q(N^\circ \mid N). \]
\item Draw $U \sim \operatorname{Uniform}(0,1)$. If $\log U < A$, set $N$ to $N^\circ$.
\end{enumerate}
 This iterative scheme yields a sequence $\{N_i\}$ from the posterior of the model index (here $i$ is the Monte-Carlo iteration number). Within a model indexed by $N$, coefficients can be sampled from their posterior
\[ \psi_k^{(N)} \sim \operatorname{G}\left(\alpha + H_k^{(N)}, \beta + n\Delta_k^{(N)}\right), \qquad k \in \{1,\ldots, N\}.\] 
% Note that since $\operatorname{ML}_N(X^{(n)})$ is available in closed form 
Note that existence of a closed form expression for $\operatorname{ML}_N(X^{(n)})$ greatly simplifies derivation of the reversible jump algorithm.

For each proposed model $N$ the binning vector $(H_1^{(N)},\ldots, H_N^{(N)})$ needs to be computed. However,  for a particular value of $N$ this only needs to be done once over all MCMC iterations.

% \begin{exam}
To illustrate performance of the method, we revisit the example of Subsection \ref{subsec:exponential} and compare the results to those that were displayed in Figure \ref{fig:example1n1}. 

The proposal $q$ on the model index was defined as follows: Fix  $\eta\in (0,1/2)$. If the present model index is $\ge 2$, we propose to go one index up or one index down with probability $\eta$. With probability $1-2\eta$ we propose to stay within the present model. When the present model index is $1$, with equal chances we stay in this model or move up to model $2$.

We ran the sampler for  $30000$ iterations, discarding the first half of the iterates as burn-in. We chose $\eta=0.45$.  Within each model we assumed the $\operatorname{G}(0.1,0.1)$ prior distribution on $\psi_k$'s. We considered two different choices of the prior on $N$:
\begin{enumerate}
	\item $N \sim \operatorname{DiscreteUniform}(\{1,\ldots, 50\})$;
	\item $N \sim \operatorname{ShiftPoisson}(23)$, by which we mean the $\operatorname{Poisson}(23)$ distribution restricted to $\{1, 2, \ldots\}$. 
\end{enumerate}
The choice $23$ was motivated by the lower panel in Figure \ref{fig:example1n1}, where the number of bins equals $23$. Plots of the posterior along with marginal credible bands are in Figure \ref{fig:estimates-revjump}. The results do not appear to be satisfactory. Especially for the uniform prior on $N$ the posterior mean is visually oversmoothing. Under both priors, realisations from the posterior do not have any visual smoothness property, in the sense that conditional on the model index, the heights over the bins are both a priori and a posteriori independent. Hence our preference for the GMC prior. 

To illustrate difficulties with using a prior on the number of bins, we report a numerical experiment where we took 
$\lambda(x) = 2 + 0.2 \sin(30x) + \ind_{[0.7,1]}(x)$ and $n=10000$.  We employed the gamma Markov chain prior with $\operatorname{Exp}(0.1)$ on the smoothing parameter $\alpha$ and fixed number of bins $N=200$, as well as the independent gamma prior with $N \sim \operatorname{DiscreteUniform}(\{1,\ldots, 50\})$. 
The results are shown in Figure \ref{fig:last}. As can be seen from the bottom right panel there, the log-likelihood has many local maxima. The resulting difference in plots in the top row can be attributed to a defect of a sampler that only proposes visits to neighbouring sites. Having multiple local maxima in the marginal likelihood is not surprising, as any ``good'' approximation to $\lambda$ must be able to capture the discontinuity at $0.7$. One could instead use more advanced samplers, though this will come at a further computational cost. As a side remark, we note that estimation of discontinuous functions might be a delicate matter with the GMC prior. See \cite{gugu18wavelets} for further remarks.

Rather than employing a prior on $N$ for equal-sized bins, one could instead consider bins obtained from a non-equidistant grid of points. We conjecture that this will give visually more appealing results. However, such an approach requires a reversible jump algorithm over the locations of the bin endpoints. This is disadvantageous: the computing times of such  algorithms do not scale well with the number of observations.

%{\color{red} Could it be that no high order indices are visited simply because all bins are equally sized? This means that in regions where the intensity is low, we'll have few Poisson points and therefore require a large bin?  } {\color{blue} Shota: I don't think there is an easy way to adaptively determine the bin width.}

\begin{figure}
\begin{center}
		\includegraphics[width=0.45\textwidth]{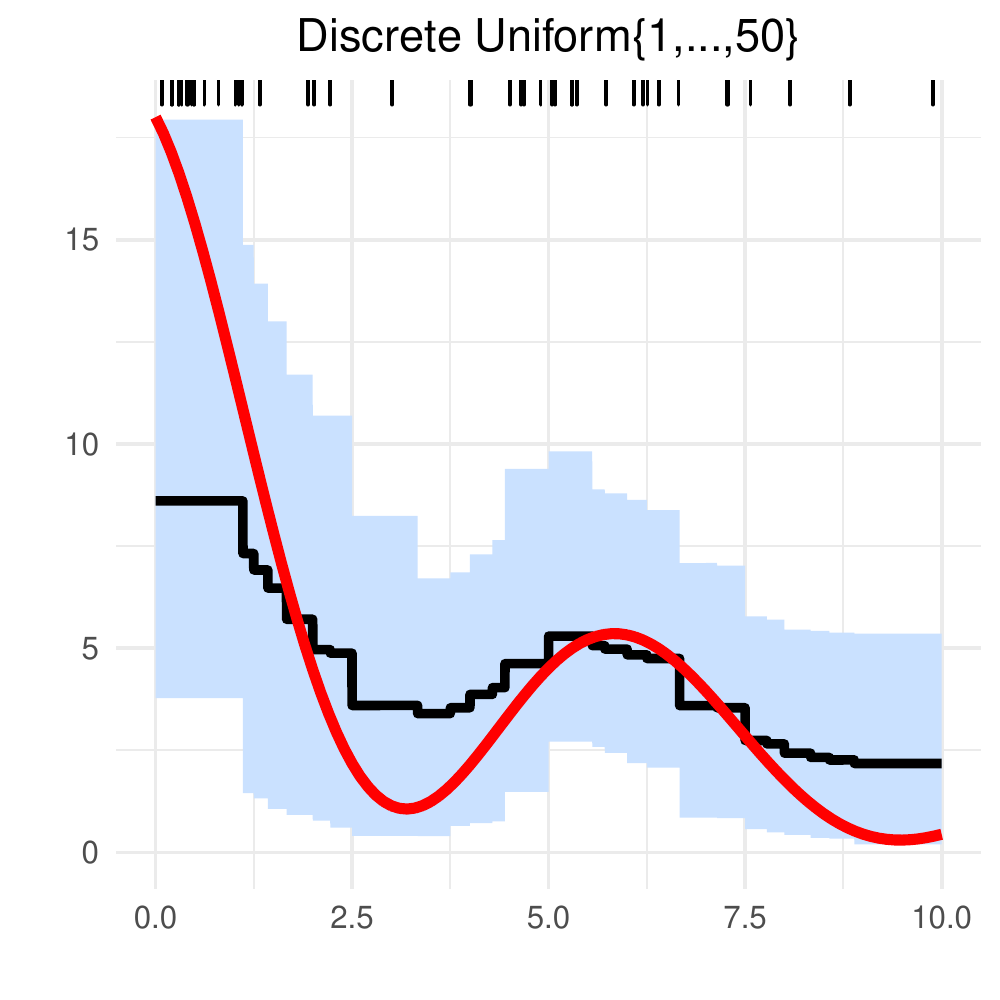}
	\includegraphics[width=0.45\textwidth]{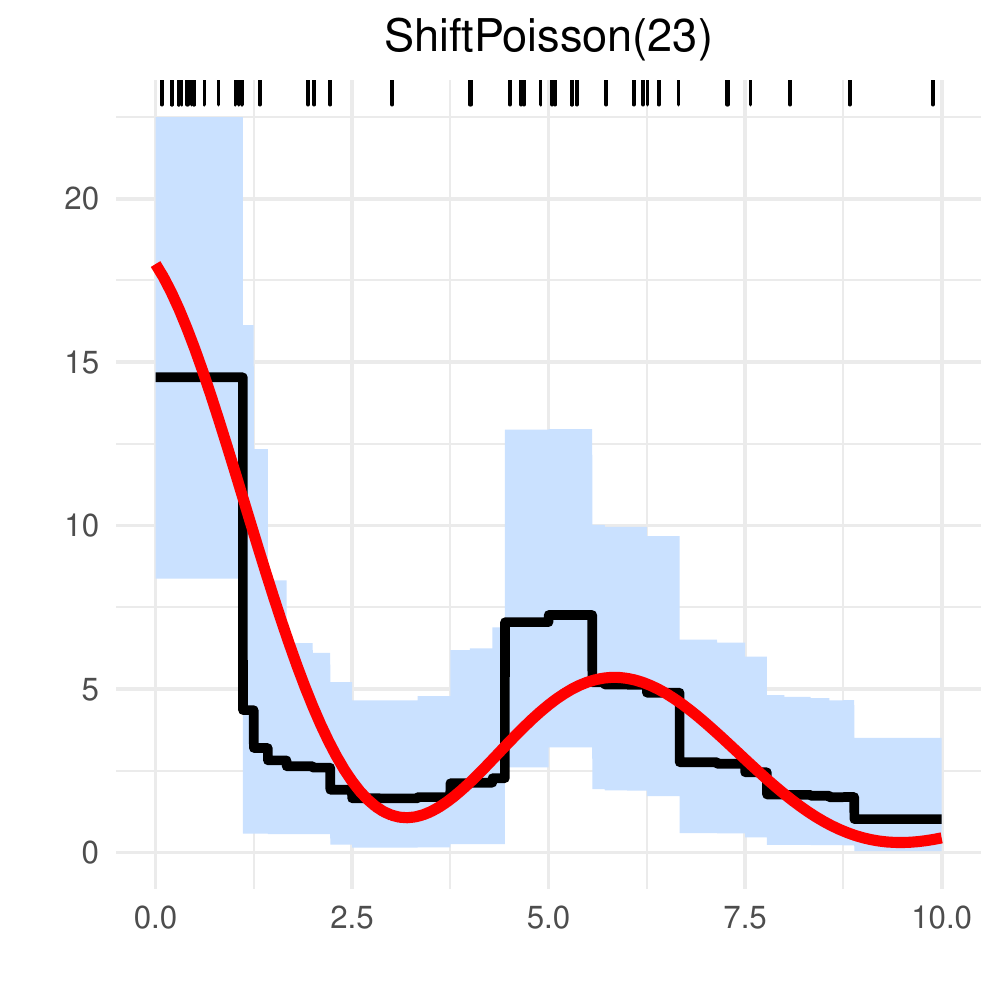}
	\caption{Estimation results for the oscillating exponential function $\lambda_0$ from Subsection \ref{subsec:exponential} with $n=1$ using a prior on the model index $N$. The prior on $\psi_k$'s is $\operatorname{G}(0.1,0.1)$.  Left: $N \sim \operatorname{DiscreteUniform}(\{1,\ldots, 50\})$. Right: $N \sim \operatorname{ShiftPoisson}(23)$. \label{fig:estimates-revjump}}
\end{center}
\end{figure}

\begin{figure}
\begin{center}
	\includegraphics[width=0.9\textwidth]{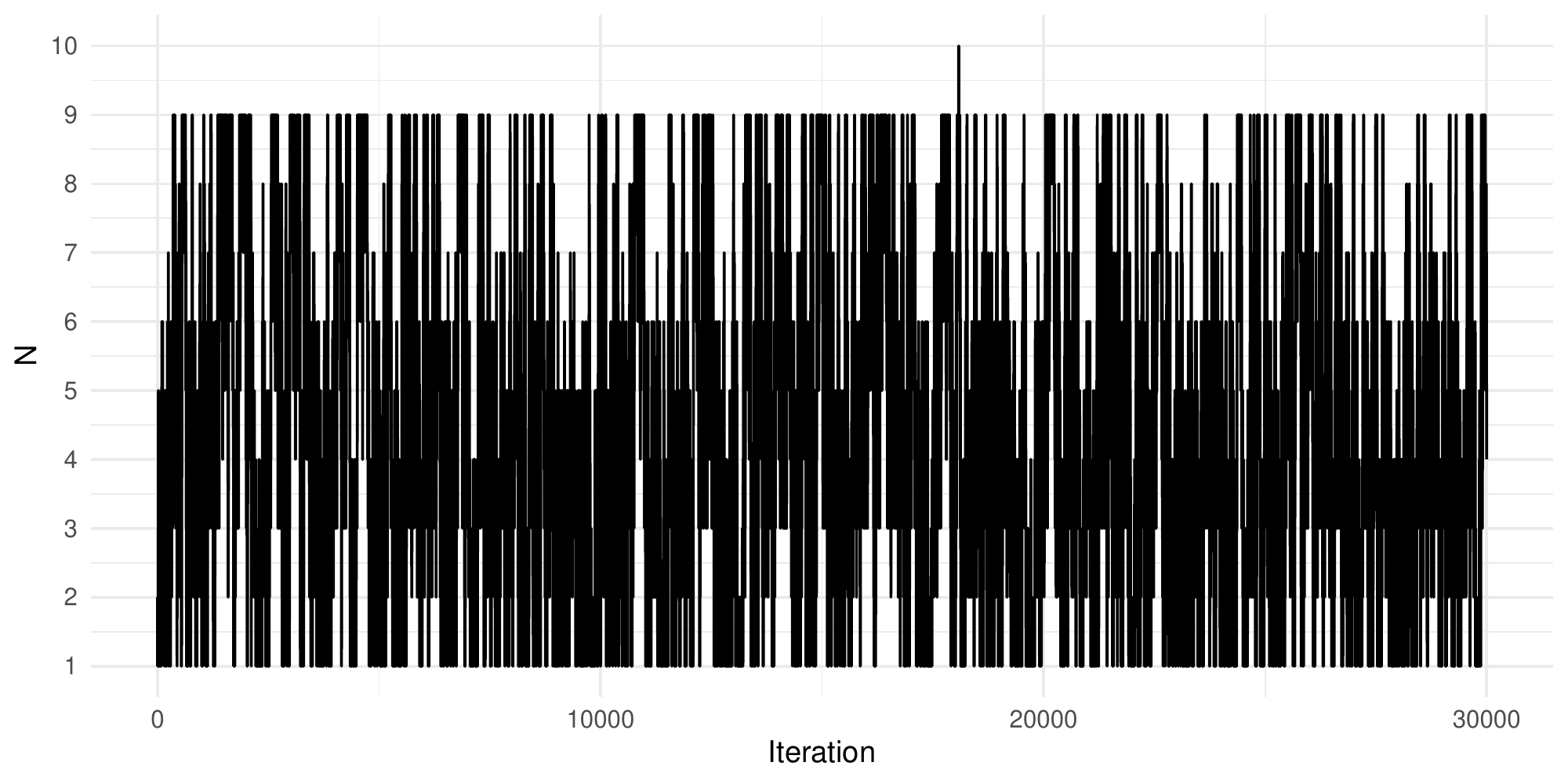}
	\includegraphics[width=0.9\textwidth]{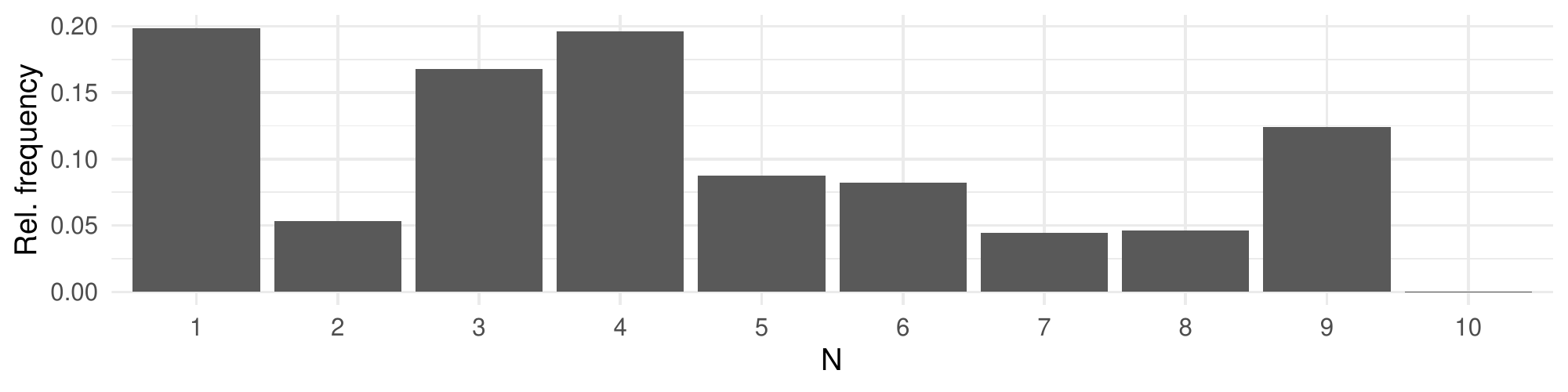}
	\end{center}
	\caption{ Results with $N \sim \operatorname{DiscreteUniform}(\{1,\ldots, 50\})$. Top: traceplot on the model index $N$. Bottom: relative frequencies of models.}
\end{figure}

\begin{figure}
\begin{center}
	\includegraphics[width=0.9\textwidth]{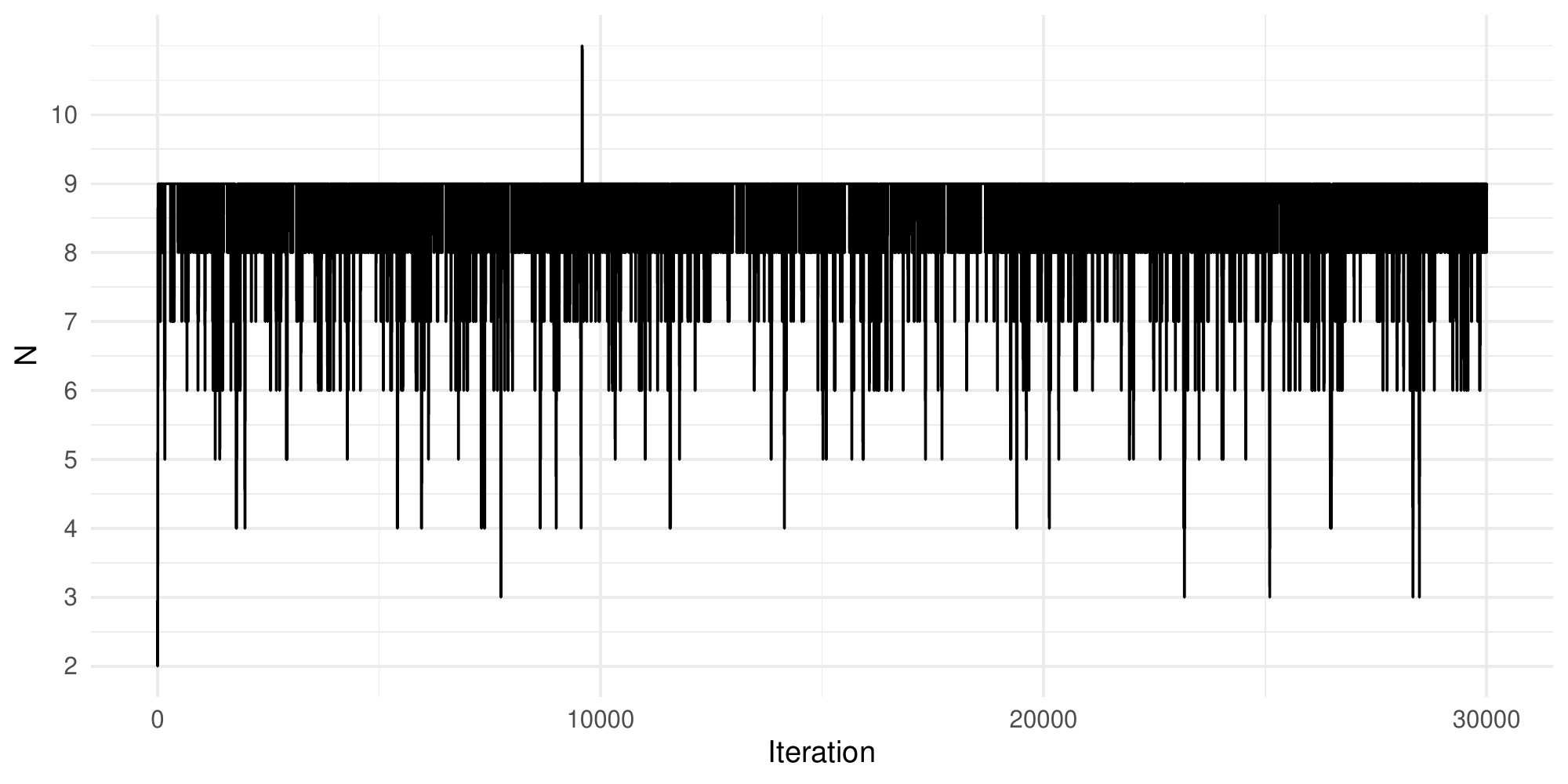}
	\includegraphics[width=0.9\textwidth]{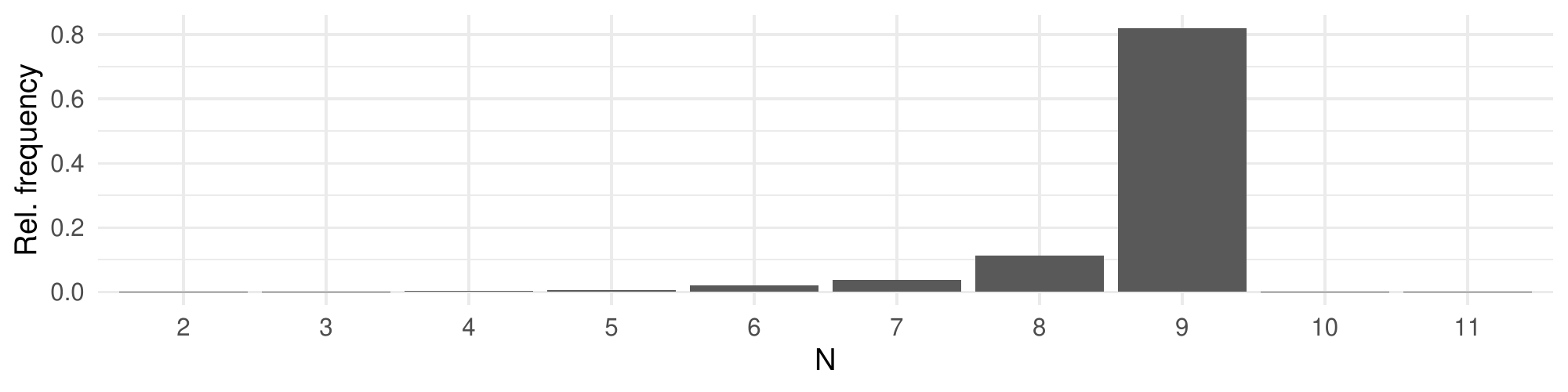}
	\end{center}
		\caption{Results with $N \sim \operatorname{ShiftPoisson}(23)$ prior. Top: traceplot on the model index $N$. Bottom: relative frequencies of models.}
\end{figure}

%\end{exam}
\begin{figure}
\begin{center}
	\includegraphics[width=0.45\textwidth]{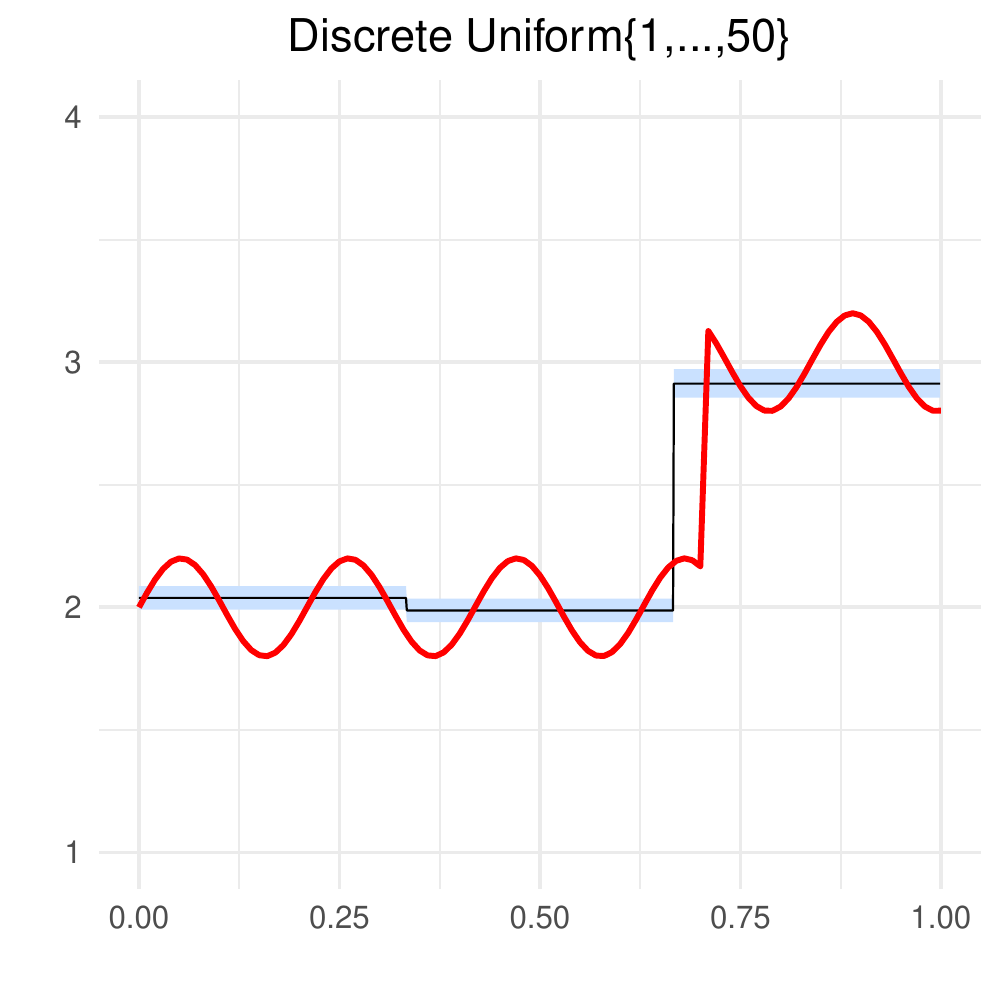}
	\includegraphics[width=0.45\textwidth]{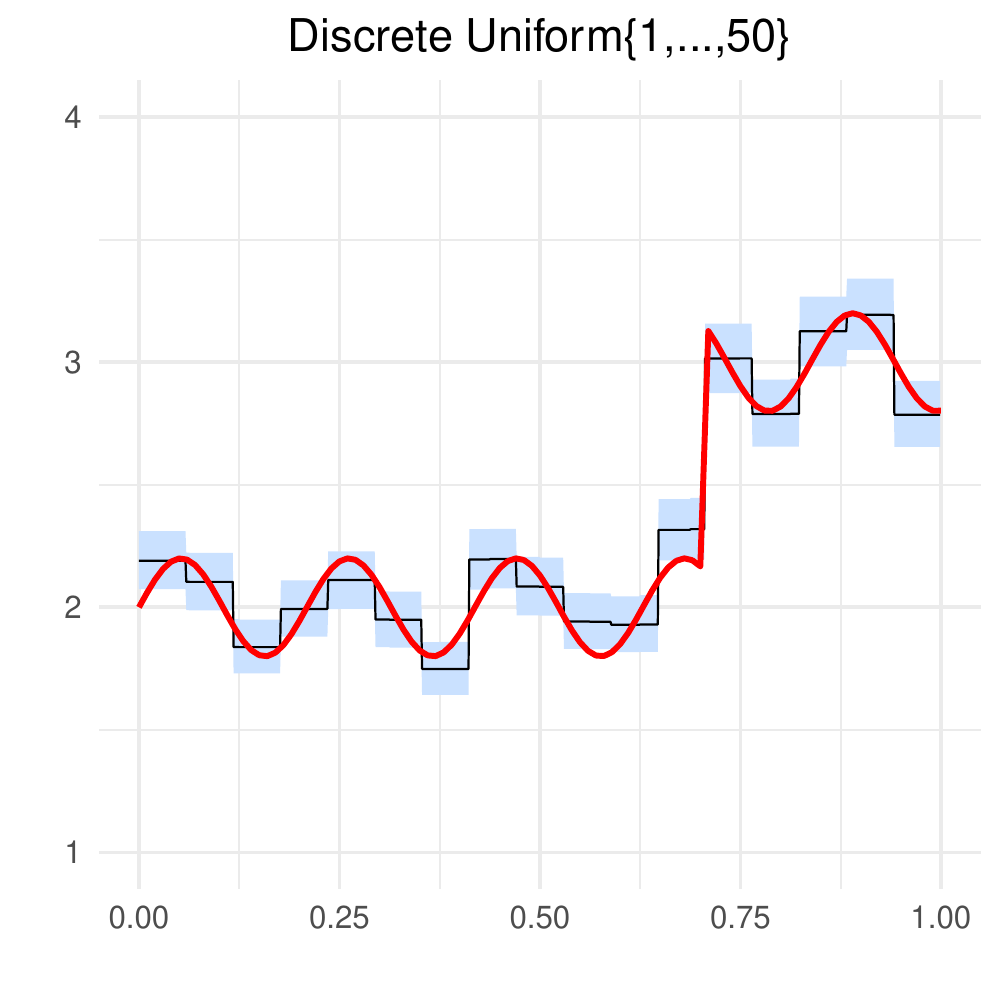}
	\includegraphics[width=0.45\textwidth]{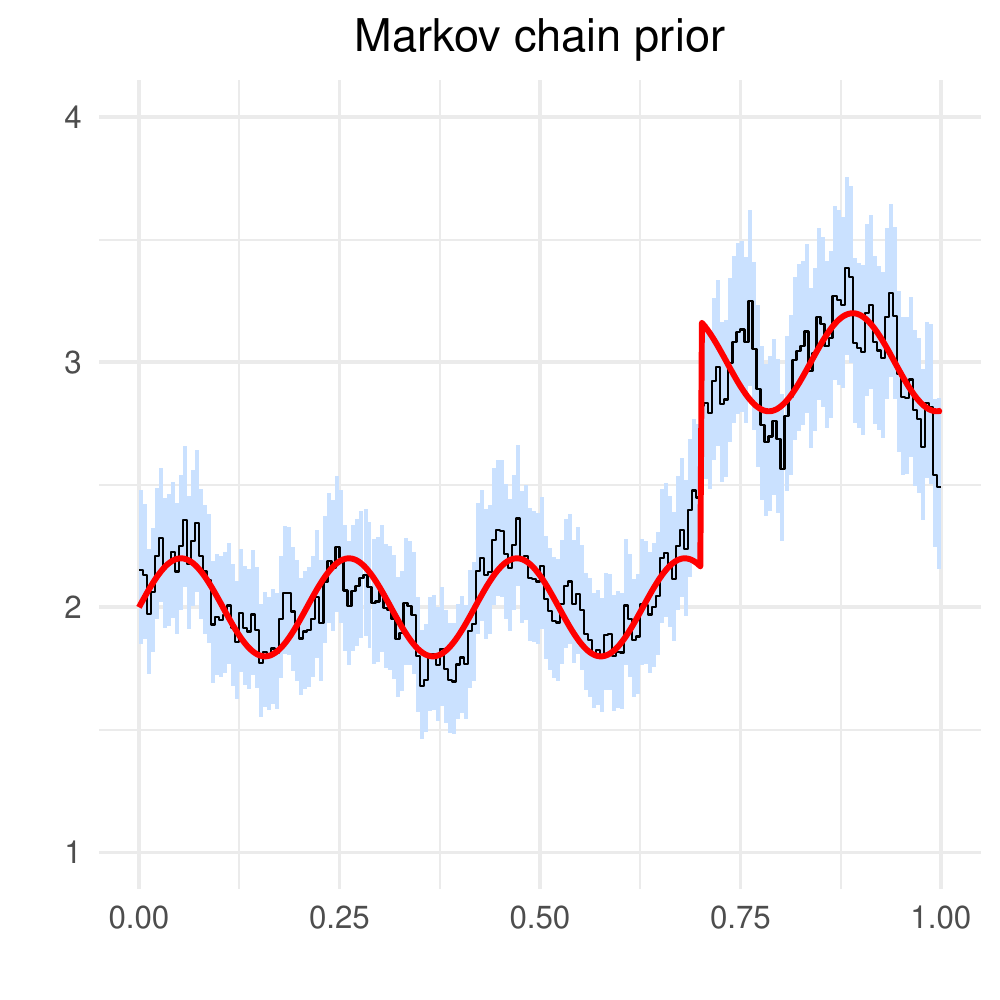}
	\includegraphics[width=0.45\textwidth]{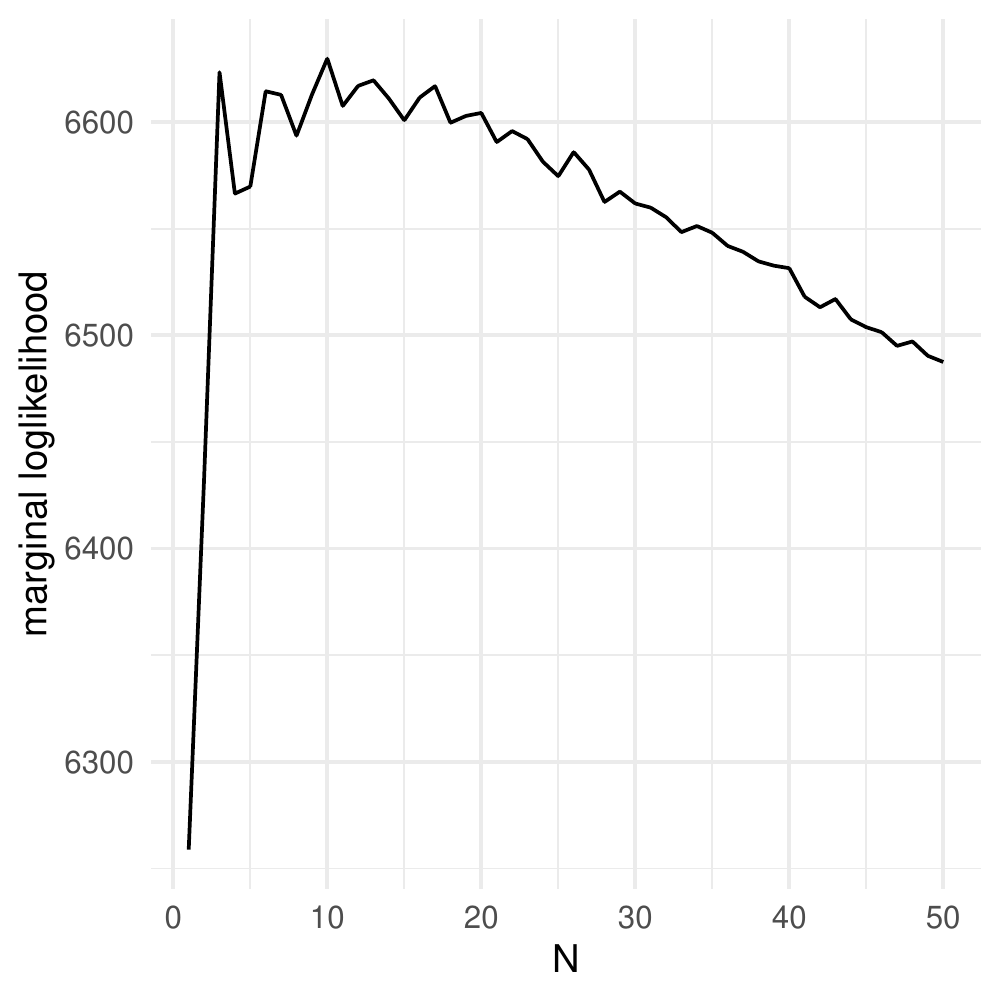}	
	
\end{center}
	\caption{Estimation results for true intensity $\lambda(x) = 2 + 0.2 \sin(30x) + \ind_{[0.7,1]}(x)$ with $n=10000$. Top left: independent gamma prior with reversible jump algorithm initialised at model index $2$. Top right: independent gamma prior with reversible jump algorithm initialised at model index $20$. Bottom left:  GMC prior with $N=200$ bins and $\operatorname{Exp}(0.1)$ prior on the smoothing parameter $\alpha$. Bottom right: marginal log-likelihood as function of the model index.  \label{fig:last}}
\end{figure}

\ifthenelse{1=2}{
	
	% If doing bibliography manually
	
	\section*{References}
	
	\begin{description}
		
		\item First reference
		
		\item Second reference
		
	\end{description}
	
}{
	
	% If doing bibliography with BibTeX
	
	\bibliographystyle{apa-good}
	\bibliography{bibppp}
	
}

\end{document}